\documentclass[pra,aps,twocolumn,groupedaddress,superscriptaddress,nofootinbib,preprintnumbers]{revtex4-2}

\usepackage{graphicx}
\usepackage[nointegrals]{wasysym} %
\usepackage[export]{adjustbox}
\usepackage{amsmath,amsfonts,amssymb,latexsym}
\usepackage{hhline}
\usepackage{bm}
\usepackage{verbatim}
\usepackage{enumitem}
\usepackage{mathrsfs}
\usepackage{slashed}
\usepackage{empheq}

\newcommand{\p}{\partial}

\newcommand{\lan}{\langle}
\newcommand{\ran}{\rangle}

\newcommand{\unit}{\mathbf{1}}

\newcommand{\da}{{\dagger}}

\newcommand{\ob}[1]{\mkern 1.5mu\overline{\mkern-1.5mu#1\mkern-1.5mu}\mkern 1.5mu}

\newcommand{\ra}{\rightarrow}
\newcommand{\lra}{\leftrightarrow}

\newcommand{\wt}{\widetilde}

\newcommand{\uvx}{{\mathbf{\hat x}}}
\newcommand{\uvy}{{\mathbf{\hat y}}}

\newcommand{\bfzero}{{\mathbf{0}}}

\renewcommand{\(}{\left(}
\renewcommand{\)}{\right)}
\renewcommand{\[}{\left[}

\newcommand{\mt}{\mapsto}

\newcommand{\tp}{\otimes}

\newcommand{\D}{\nabla}

\newcommand\bpm            {\begin{pmatrix}}
	\newcommand\epm           {\end{pmatrix}}

\newcommand{\ms}{\medskip}

\def\app#1#2{%
	\mathrel{%
		\setbox0=\hbox{$#1\sim$}%
		\setbox2=\hbox{%
			\rlap{\hbox{$#1\propto$}}%
			\lower1.1\ht0\box0%
		}%
		\raise0.25\ht2\box2%
	}%
}

\newcommand{\tw}{\textwidth}

\newcommand{\inv}{^{-1}}

\newcommand{\ope}\odot

\usepackage{manfnt}

\newcommand{\bi}{\begin{itemize}}
	\newcommand{\ei}{\end{itemize}}

\usepackage{amsthm}
\newtheorem{theorem}{Theorem}
\newtheorem{definition}{Definition}

\newtheorem{proposition}{Proposition}

\theoremstyle{definition}

\newcommand\bpro		  {\begin{proposition}}
	\newcommand\epro 		  {\end{proposition}}
\newcommand\bproof			  {\begin{proof}}
	\newcommand\eproof 		  {\end{proof}}

\newcommand\ed            {\end{definition}}

\newcommand\be            {\begin{equation}}
\newcommand\ee            {\end{equation}}
\newcommand\ba            {\begin{aligned}}
\newcommand\ea            {\end{aligned}}
\newcommand\bea{\begin{equation}\begin{aligned}}
	\newcommand\eea{\end{aligned}\end{equation}}
\usepackage{xcolor}

\usepackage{hyperref} 
\hypersetup{final}
\definecolor{darkblue} {rgb}{.1,.5,0.65}
\definecolor{darkgreen}{rgb}{.1,.18,.82}
\hypersetup{colorlinks,
	urlcolor   = blue,         
	linkcolor  = darkgreen,     
	citecolor  = darkblue,    
}
\newcommand{\sss}{\subsubsection}
\renewcommand{\ss}{\subsection}

\renewcommand{\a}{\alpha}
\renewcommand{\b}{\beta}
\renewcommand{\d}{\delta}
\newcommand{\De}{\Delta}
\newcommand{\g}{\gamma}

\newcommand{\s}{\sigma}

\newcommand{\ep}{\varepsilon} %
\renewcommand{\l}{\lambda}

\renewcommand{\t}{\theta}

\renewcommand{\o}{\omega}

\renewcommand{\c}{\chi}
\newcommand{\z}{\zeta}

\newcommand{\bfdel}{{\boldsymbol{\delta}}}

\newcommand{\bfn}{\mathbf{n}}

\newcommand{\bfr}{\mathbf{r}}

\newcommand{\bfv}{\mathbf{v}}

\newcommand{\bfx}{\mathbf{x}}

\newcommand{\bfy}{\mathbf{y}}

\newcommand{\zt}{\mathbb{Z}_2}
\newcommand{\zn}{\mathbb{Z}_N}

\newcommand{\EE}{\mathbb{E}}

\newcommand{\qq}{\qquad}

\newcommand{\zz}{\mathbb{Z}}

\newcommand{\mcc}{\mathcal{C}}

\newcommand{\mcb}{\mathcal{B}}

\newcommand{\mco}{\mathcal{O}}

\newcommand{\mch}{\mathcal{H}}
\newcommand{\mca}{\mathcal{A}}

\newcommand{\mcm}{\mathcal{M}}

\newcommand{\mcv}{\mathcal{V}}

\newcommand{\mcr}{\mathcal{R}}

\newcommand{\sfF}{\mathsf{F}}

\newcommand{\sfM}{\mathsf{M}}

\newcommand{\sfP}{\mathsf{P}}

\newcommand{\sfR}{\mathsf{R}}
\newcommand{\sfT}{\mathsf{T}}

\newcommand{\sfa}{\mathsf{a}}
\newcommand{\sfb}{\mathsf{b}}

\newcommand{\sfm}{\mathsf{m}}

\newcommand{\sfr}{\mathsf{r}}

\usepackage[mathscr]{eucal} %

\renewcommand{\tt}[1]{\mathtt{#1}}

\renewcommand{\k}[1]{|#1\rangle}

\usepackage{dcolumn}

\usepackage{pifont}
\usepackage[normalem]{ulem}

\usepackage[caption=false]{subfig}
\usepackage{enumitem}
\usepackage{amsthm}
\usepackage{physics}
\usepackage{booktabs}
\usepackage{algorithm}
\usepackage{algpseudocode}
\usepackage{bm}
 
\renewcommand\qq{\qquad}

\newcommand{\oEE}{\mathop{\mathbb{E}}}

\newcommand{\tmem}{t_{\rm mem}}
\newcommand{\poly}{{\rm poly}}
\newcommand{\maj}{{\sf maj}}
\newcommand{\dam}{{\sf dam}}
\usepackage{seqsplit}
\usepackage{array}

\begin{document}

	\title{Squeezing codes:  robust fluctuation-stabilized memories}
	
	\author{Ethan Lake}	
	\email{elake@berkeley.edu} 
	\affiliation{Department of Physics, University of California, Berkeley}	
	\author{Sunghan Ro}
	\affiliation{Department of Physics, Harvard University}
    
	\begin{abstract} 
		We introduce families of classical stochastic dynamics in two and higher dimensions which stabilize order in the absence of any symmetry. Our dynamics are qualitatively distinct from Toom's rule, and have the unusual feature of being fluctuation-stabilized: their order becomes increasingly fragile in larger dimensions. One of our models maintains an ordered phase only in two dimensions. The phase transitions that occur as the order is lost appear to realize new dynamical universality classes which are fundamentally non-equilibrium in character. 
	\end{abstract}
	
	\begin{figure*}
		\includegraphics[width=.97\tw]{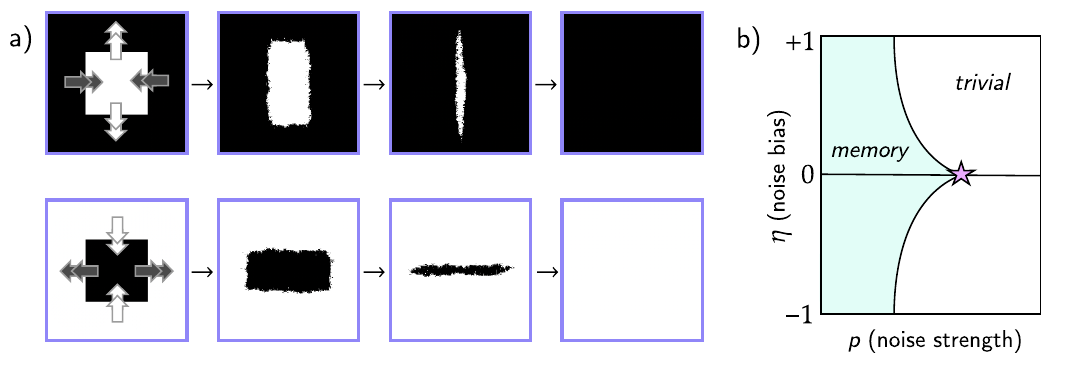}
		\caption{\label{fig:overview} {\sf a)} Error correction in the $\sfR$ squeezing code, shown for noiseless asynchronous updates. The top panels: a $+1$ (white) minority domain is annihilated by being squeezed in the direction indicated by the arrows. The bottom panels show the squeezing experienced by a $-1$ (black) minority domain, which is related to the process in the top panels by the combination of a spin flip and a $\pi/2$ rotation. {\sf b)} the schematic phase diagram of squeezing codes as a function of noise strength $p$ and bias $\eta$. The phase transition out of the memory phase is first order at $\eta \neq 0$ and second order at $\eta = 0$. The pink star at zero bias denotes a new type of  non-equilibrium critical point, which is discussed in Sec.~\ref{sec:crit}.}
	\end{figure*}	
	
	\maketitle
	\tableofcontents
	
	\section{Introduction and summary}\label{sec:intro}

	\ss{Many-body memories} 
	This work studies {\it robust many-body memories:} locally-interacting classical systems which maintain order even when subjected to arbitrary perturbations. These systems retain memory of their initial conditions for thermodynamically long time scales, and are capable of counteracting the effects of a noisy adversarial environment. From a condensed matter standpoint, they define absolutely stable open phases of matter; from an information theory standpoint, they define autonomous error-correcting codes. 
	
	The characterization of robust memories is an ongoing research program at the intersection of non-equilibrium statistical mechanics and condensed matter physics \cite{ponselet2013phase,rakovszky2024defining,sang2024mixed,pajouheshgar2025exploring}, fault-tolerant (quantum) computation \cite{gacs1978one,toom1980stable,cirel2006reliable,gacs2001reliable}, and artificial life \cite{mordvintsev2020growing,challa2024effect,pajouheshgar2024noisenca}. 
	Memory in locally-interacting systems must be achieved through collective, distributed error correction: errors in the information encoded into the initial state must be autonomously detected and corrected by the noisy dynamics itself. Most of the recent progress on understanding how this can be done has been made in the quantum error correction community, largely driven by an interest in parallelized decoding algorithms for quantum memories \cite{Harrington2004,ray2024protecting,balasubramanian2024local, paletta2025high,lake2025fast}. In spite of this, our understanding of the general mechanisms for stabilizing robust many-body memories remains far from complete. 
	
	This work introduces new classes of robust memories, which we define within the framework of probabilistic cellular automata. Cellular automata (CA) are general models of translation-invariant dynamics operating on lattices of discrete spins, with each spin updated according to a particular (potentially non-deterministic) function $\mca$ of the spins in its immediate vicinity. To study the robustness of a particular CA, the most common approach is to define noisy dynamics that performs $\mca$ with probability $1-p$, and suffers noise with probability $p$. When noise acts on a spin, we replace that spin by $\pm1$ with probability $(1\pm\eta)/2$. The parameter $p$ will be referred to as the noise strength, and $\eta$ referred to as the noise bias. Letting $s_{\{\bfr'\sim \bfr\}}$ denote the set of spins neighboring the site $\bfr$, the noisy CA dynamics acts on each spin $s_\bfr$ as
	\be s_\bfr \mt \begin{dcases} \mca(s_{\{\bfr'\sim \bfr\}}) & \text{w/ prob. $1-p$} \\ 
		\pm 1 & \text{w/ prob. $p(1\pm \eta)/2$} \end{dcases}.\ee 
	and the challenge is to find a choice of $\mca$ such that the system retains order in an open subset of the $(p,\eta)$ plane. 
	
	There are many CA known to be capable of retaining (and even processing) information when $p=0$: examples include Conway's game of life, several of Wolfram's elementary CA, the 2D Ising model (with the updates chosen probabilistically in a way that respects detailed balance), among others. 
	Unfortunately, nearly all such CA are fine-tuned: the information they retain is, in almost all cases, rapidly destroyed at any nonzero $p$ (the memory in the Ising model is destroyed in the presence of a nonzero magnetic field, which is analogous to having noise at bias $\eta \neq 0$). Robustness against noise mandates that these systems function as collective error correcting codes: errors must be identified and eliminated locally, despite the fact that any given spin does not have access to the full information that it is trying to protect. 
	
	The goal of this work is twofold: 
	\begin{itemize}
		\item We aim to better understand the mechanisms by which collective error correction can be performed, so as to further the classification of absolutely stable phases of matter, and 
		\item  By studying robust memories from the point of view of statistical physics, we aim to gain new insights into non-equilibrium statistical mechanics.
	\end{itemize}
	We are particularly interested in memories with simple update rules, and which can be analytically proven to have thresholds. 
	
	This work presents families of CA, dubbed {\it squeezing codes}, which meet these criteria. These CA---which we will mostly discuss in $d=2$, but which are defined in any $d\geq 2$---can be rigorously proven to robustly remember a single bit of information. They correct errors in a way qualitatively distinct from other existing memories, and this endows them with physical properties that are quite unusual from the point of view of statistical mechanics. 
	
	\ss{Overview of prior work} 
	
	Before describing squeezing codes, we first give a brief summary of the robust memories known at the time of writing. 
	
	The first example of a robust memory is a 2D CA known as {\it Toom's rule},\footnote{Toom's rule was in fact originally introduced not by Toom, but by the authors of \cite{vasil1969modelling} in 1969, where it was studied numerically. Toom was the first to give a rigorous proof of its noise robustness, which he did in 1974  \cite{toom1974nonergodic} (c.f. the discussion in sec.1.5 of \cite{swart2022peierls}).} which operates by sending each spin to the majority vote of itself and its northern and eastern neighbors, $\mca_{\sf Toom} \, : \, s_\bfr \mt \maj(s_{\bfr},s_{\bfr+\uvx},s_{\bfr+\uvy})$. Toom proved that this dynamics retains memory of the sign of its magnetization for a thermodynamically long time scale \cite{toom1974nonergodic,toom1980stable}. Although Toom's rule is a robust memory, from the perspective of statistical physics, it is rather similar to the 2D Ising model (in a way we make precise later on). 
	
	After this, Gacs produced an extremely complicated memory in 1D based on a technique referred to as ``hierarchical self-simulation'' \cite{gacs2001reliable} (see also \cite{gacs1989self}). A model of 1D local dynamics with an ordered phase has the potential to be very interesting from the point of view of statistical mechanics, but Gacs' automaton is unfortunately so complicated that even implementing it numerically has proven to be a prohibitively daunting task. 
	
	With the exception of thermally stable gauge theories (which exist only in dimensions $d\geq 3$ and lack local order parameters \cite{poulin2019self}), and the generalization of Toom's rule to different lattices given in \cite{kubica2019cellular}, Toom's rule remained the only {\it simple} known robust memory\footnote{In fact, Toom constructed another concrete example in $d=2$, defined by \cite{toom1980stable,fernandez2001non}
		\be s_\bfr \mt \max(\min(s_\bfr, s_{\bfr+\uvx}), \min(s_{\bfr+\uvy},s_{\bfr+\uvx+\uvy})).\ee 
		Apart from a brief mention in  \cite{ponselet2013phase}, this rule appears not to have been studied in detail, perhaps because it lacks any symmetry. The way it corrects errors is somewhat similar to the $\sfR$ squeezing code introduced below.} until a family of new models in $d=2$ dimensions were constructed by one of the authors \cite{pajouheshgar2025exploring}. This work showed that automata which perform majority votes over sufficiently asymmetric collections of sites are robust memories, and used machine learning methods to discover a large number of other candidate examples which were not mathematically proven to be memories, but were suggested to be so from numerics. The majority vote automata are in some sense qualitatively similar to Toom's rule, and the rules discovered with machine learning have the drawback of sometimes being hard to interpret, and not being amenable to analytic study.
	
	\ss{Building non-equilibrium memories by squeezing} 
	
	Let us first recall the central challenge faced when attempting to construct a robust memory. Consider a locally-interacting system which stores a single bit of logical information, which it encodes in the sign of the system's magnetization.\footnote{Other than Gacs' automaton and gauge theories, all currently known robust memories are of this form.} Let $s_\pm$ be the state in which all spins are $\pm1$, and assume that $s_\pm$ are invariant under the dynamics (the $s_\pm$ will be referred to as ``logical states'' in what follows). Consider taking $s_+$ and creating a minority domain of $-1$ spins of radius $R$ somewhere in the system. Let $R\gg 1$ be large, but vanishingly small compared to the system size $L$.  
	Due to noise, there is always a finite nucleation rate per unit area for minority domains of this size, and such droplets will inevitably appear in large systems. Therefore, for the dynamics to realize a robust memory, it must autonomously detect and correct errors introduced by these domains.

	Since $s_\pm$ are invariant under the dynamics, this domain can only be corrected by a gradual contraction of the domain wall at its boundary. 
	If $\eta<0$, so that the noise in the system is biased to favor $-1$, the minority domain will expand ballistically in the absence of error correction. In order to combat this ballistic expansion, the noiseless part of the dynamics must therefore {\it ballistically} erode minority domains. At first glance, it is not obvious how to do this, since the dynamics must act to correct errors as soon as they form: it does not have enough time to establish which side of the domain is the ``inside'' and which is the ``outside''. 
	
	Toom's rule moves domain walls normal to $\uvx-\uvy$ along $-\uvx-\uvy$ while keeping all other domain walls fixed, and Toom proved that this always results in the ballistic erosion of minority domains. There is, however, another way of eroding minority domains, which is simple to describe: by ``squeezing'' them. 
	
	Consider dynamics where 
	\begin{enumerate}
		\item $-1$ spins convert their north and south neighbors to $-1$,
		\item $+1$ spins convert their east and west neighbors to $+1$.
	\end{enumerate}
	What happens to a $-1$ minority domain under this dynamics? Point 1 means that it vertically expands, and point 2 means that it horizontally contracts. This leads to the minority domain being ballistically {\it squeezed}: it steadily becomes taller and narrower, becoming increasingly narrow until its width vanishes. The opposite is true for a $+1$ minority domain, which steadily becomes wider and shorter until its height vanishes. See Fig.~\ref{fig:overview}~$\sfa$ for an illustration. For an interactive demo of the dynamics, see \cite{supp}. 
	
	To define this dynamics more concretely, we will find it convenient to use nonstandard Boolean notation where a $+1$ spin is treated as $\tt{true}$ and a $-1$ spin is treated as $\tt{false}$. In terms of $\wedge$ ($\tt{and}$) and $\vee$ ($\tt{or}$) operations, one choice of a microscopic update rule implementing this squeezing process is 
	\bea \label{sqzrule_r} \sfR\, : \, s_\bfr(t+1) = \begin{dcases} s_{\bfr+\uvy}(t) \, \wedge \, s_{\bfr}(t) \wedge\, s_{\bfr-\uvy}(t)  & t \,\, {\rm even}  \\ 
		s_{\bfr+\uvx}(t)\, \vee \, s_\bfr(t) \, \vee \,  s_{\bfr-\uvx}(t) & t \, \, {\rm odd} 
	\end{dcases} \eea
	During even time steps $-1$ domains expand vertically, and during odd ones $+1$ domains expand horizontally; together, these updates implement the desired squeezing process. 
	While an update rule defined in terms of $\vee,\wedge$ may appear unusual, we will show in Sec.~\ref{ss:glauber} that it emerges naturally in a system undergoing heat-bath dynamics with respect to a time-periodic Ising model. 
	
	The dynamics written in \eqref{sqzrule_r} is {\it synchronous}, meaning that at each (discrete) time step, all of the spins in the system are updated simultaneously. This is the type of dynamics usually studied in the literature on cellular automata, but from a physical perspective, it is not desirable: it requires that all of the sites in the system, regardless of how far apart they are, have instantaneous access to a shared ``clock'', which tells them when to update. This renders the dynamics non-local (in the sense that it cannot be generated by a local Lindbladian), and requires that the clock be completely reliable, immune to the effects of noise. Both of these features are a departure from the desired features for a robust memory, and as such, we will mostly focus on {\it asynchronous} updates, where each site updates at times dictated by independent Poisson processes (doing so ensures that the dynamics can be described by a continuous time Markov process). For $\sfR$, we make the dynamics asynchronous by having the $\wedge$ and $\vee$ updates chosen randomly with equal probability. The snapshots in Fig.~\ref{fig:overview}~$\sfa$ were generated with updates of this type. 
	
	There are several reasons why $\sfR$ performs error correction in a way qualitatively distinct from Toom's rule. For one, convex minority domains do not always get monotonically smaller as they are corrected: the dynamics can sometimes make errors worse before ultimately correcting them (consider a minority domain of $-1$ spins which is taller than it is wide). Another reason is that it possesses a different symmetry. Let 
	\be \label{xdef} X \, : \, s \mt - s \ee 
	denote the spin-flip operation which exchanges $+1$ and $-1$ spins. Toom's rule is symmetric under $X$, since it is computed using majority votes. The symmetry of $\sfR$ is instead a spin flip {\it combined} with a $\pi/2$ rotation $R_{\pi/2}$, viz. the operation 
	\be \label{mixsym} X_{R_{\pi/2}} = X \circ R_{\pi/2}.\ee 
	As we will see in the following sections, this intertwining of a spin flip with a spatial symmetry has important consequences for several aspects of $\sfR$'s physics. 
	
	In addition to $R_{\pi/2}$, it is also possible to construct similar types of error-correcting dynamics symmetric under $X \circ C$, where $C$ is any spatial symmetry (except, as it turns out, a $\pi$ rotation).  Dynamics with these symmetries also corrects errors by squeezing minority domains, and we will therefore refer to these models as {\it squeezing codes}. 
	
	\ss{The physics of squeezing codes}
	
	We now give an overview of the physics of squeezing codes, which will also function as a summary of our main results. 
	
	\sss{Robust memory phases} 
	
	Our first result concerns the ability of squeezing codes to function as robust memories. We study four classes of squeezing codes in this work, denoted by $\sfR,\sfF,\sfM,\sfT$, which we introduce in Sec.~\ref{ss:setup}. Each class is distinguished by its symmetries: $\sfR$ combines a spin flip with a rotation, $\sfF$ and $\sfM$ combine spin flips with different types of reflections, and $\sfT$ has a pure spin-flip symmetry. In Sec.~\ref{ss:erosion}, we show that under synchronous updates, all squeezing codes are ``eroders'', meaning that they are able to correct errors in initial states in the absence of further transient noise (in the quantum error correction language, they function as ``offline decoders''). $\sfR,\sfF,$ and $\sfM$ remain eroders when the updates are made asynchronous, but $\sfT$ does not. In fact, under asynchronous updates, $\sfT$ very rapidly produces featureless disordered states from almost any initial state, even in the absence of noise. We will return to this point below. 
	
	In Sec.~\ref{ss:synchproof}, we give a rigorous proof that each of $\sfR,\sfF,\sfM$ and $\sfT$ is a robust memory under synchronous updates. This is done using a theorem of Toom \cite{toom1980stable} regarding ``monotonic eroders'', and guarantees that the system functions as a memory at sufficiently small values of the noise strength $p$ (regardless of the bias $\eta$). 
	Making rigorous proofs becomes more complicated for asynchronous updates, and in Sec.~\ref{ss:asynch_numerics} we use Monte Carlo numerics to establish the existence of memory phases in this setting. We find that the  $\sfR,\sfF$, and $\sfM$ squeezing codes remain robust memories under asynchronous updates, with phase diagrams resembling the schematic drawn in Fig.~\ref{fig:overview}~$\sfb$. 
	%	These squeezing codes have rather small critical noise strengths with asynchronous updates, less than $4\%$ or so at zero bias,\footnote{Compare this to the $p_c \approx 26\%$ of some of the automata studied in \cite{pajouheshgar2025exploring} (this being the threshold for asynchronous updates).} and as such are not particularly {\it good} memories, although they have not been designed for this purpose). 
	%	On the other hand, we find that not only is $\sfT$ {\it not} a robust memory under synchronous updates, it in fact very rapidly produces featureless disordered states, even in the absence of noise (we will return to this point below). 

    The intertwining of internal and spatial symmetries in the $\sfR$, $\sfF$, and $\sfM$ squeezing codes gives rise to unusual domain-wall dynamics in the memory phase, which we investigate in Sec.~\ref{ss:langevin}. 
    Such intertwining also appears in many models of living and active matter systems \cite{bowick2022symmetry,fruchart2021non,dinelli2023non,o2006fisher,giometto2021antagonism}, and the error-correcting dynamics of the $\sfF$ and $\sfM$ squeezing codes exhibit clear parallels with the physics of nonreciprocally interacting systems. 
    For example, the magnetization dynamics of the $\sfM$ squeezing code obey the Langevin equation
    \be \partial_t m = D \nabla^2 m + \gamma \partial_x m + \lambda m \partial_y m - r m + u m^3, \ee 
    where we omit the noise term for simplicity. This equation is not invariant under inversion of $x$ or $y$, indicating anisotropic dynamics. As will be shown in Sec.~\ref{ss:langevin}, this leads to different states of the squeezing code exerting direction-dependent nonreciprocal interactions, giving rise to propagating-band behavior reminiscent of active models with nonreciprocal couplings~\cite{dinelli2023non}. 

    In Sec.~\ref{ss:inftquench}, we study how the magnetization evolves following a quench from a random state into the ordered phase. Unlike Toom’s rule and the Ising model, the $\sfR$ squeezing code always rapidly converges to a uniform homogeneous state (whereas Toom’s rule and thermal Ising dynamics have a constant probability of becoming ``stuck'' in configurations with system-spanning strips of minority spins). In contrast, the $\sfF$ and $\sfM$ squeezing codes have a finite probability of quenching into a propagating-band steady state, in which a band of minority spins ballistically circles the system, dissipating only after a thermodynamically long time. The dynamics of each interface resemble the invading front of a species with higher fitness at the boundary between two antagonistic species, which roughens over time~\cite{o2006fisher,giometto2021antagonism}. 
	
	\sss{Fluctuation-stabilized order} 
	
	One of the most interesting aspects of squeezing codes is that they rely crucially on fluctuations to stabilize order. This is explained in Sec.~\ref{sec:meanfield}, where we analyze the squeezing codes in the dynamic mean-field approximation. Recall that this approximation is based on neglecting correlations between spins, and becomes increasingly accurate in the limit of large spatial dimension $d$. Because fluctuations are usually detrimental to the establishment of a state with uniform order, the mean-field treatment {\it over}estimates the tendency towards order. 
	Additionally, since fluctuations are usually stronger in lower dimensions, decreasing $d$ usually decreases its tendency to order. 
	
	Squeezing codes break this intuition, and in fact behave in precisely the opposite fashion: their ordering becomes weaker in higher dimensions, and stronger in smaller ones (they are only defined in $d\geq 2$). The mean-field analysis of Sec.~\ref{sec:meanfield} shows that they are in fact {\it always} disordered in mean-field, regardless of $d$, and it is in this sense that we refer to them as being ``fluctuation-stabilized''. 
	
	This phenomenon was also recently observed in the memories found numerically in \cite{pajouheshgar2025exploring}. In the present work, we show concretely how it arises in the context of squeezing codes (which, unlike the automata of \cite{pajouheshgar2025exploring}, have natural generalizations to $d>2$), and use the cluster variational method (CVM) \cite{kikuchi1951theory,pelizzola2005cluster} to include short-range fluctuations into the naive mean-field analysis. We show that doing this restores the existence of an ordered phase, and produces phase diagrams that match quite well with numerics. The CVM analysis also makes two interesting further predictions: 1) it predicts that the $\sfR$ squeezing code has an ordered phase {\it only} in two dimensions, and is disordered for all $d>2$; and 2) it predicts that if we drop the $s_\bfr$ on the RHS of \eqref{sqzrule_r}---which changes neither the symmetries of the dynamics nor the qualitative way in which it corrects errors---a memory phase is restored in {\it all} dimensions, albeit with a critical noise strength that scales as $p_c \sim 1/d^2$. Numerics are consistent with this prediction.

	\sss{New non-equilibrium critical points} 
	
	Our next results concern the nature of the phase transitions that occur as the noise strength is tuned at zero bias (the pink star in Fig.~\ref{fig:overview}~$\sfb$), which we investigate numerically in Sec.~\ref{sec:crit}. From previous theoretical work, the most reasonable prior is that these phase transitions lie in the model-A dynamic universality class \cite{hohenberg1977theory}, viz. the universality class of the critical Ising model under local, detailed-balance obeying updates. This is primarily due to RG calculations using the $\ep$ expansion, which suggest that critical model-A dynamics is stable against {\it arbitrary} perturbations, even those which break all symmetries and take the dynamics out of thermal equilibrium \cite{grinstein1985statistical,bassler1994critical,tauber2002effects}. As an example of this universality, the critical point for Toom's rule is known to belong to the model-A universality class\footnote{The agreement of the static exponents has been established in the literature, but not the agreement of the dynamic exponent $z$. In Sec.~\ref{sec:crit} we confirm that the value of $z$ indeed matches that of model-A dynamics.}  \cite{makowiec2000study,ray2024protecting},  implying that, despite the ordered phase in Toom's rule requiring the breaking of detailed balance, an effective equilibrium description emerges at the critical point. 
	
	In Sec.~\ref{sec:crit}, we use Monte Carlo numerics to show that, except possibly for  $\sfR$, the critical points of the squeezing codes are {\it not} in the model-A universality class. The static exponents $\nu,\b$ deviate from their model-A values by small but measurable amounts, and the dynamic exponent $z$ is quite different: for $\sfR,\sfF,$ and $\sfM$, we estimate $z\approx 1.942(3), 1.668(19),$ and $1.396(16)$, respectively (c.f. the model-A value of $z_A \approx 2.167$ \cite{adzhemyan2022dynamic,nightingale1996dynamic}). The (strongly) super-diffusive nature of these exponents is of particular interest, as it was recently proven that {\it any} local Markovian dynamics obeying detailed balance must have $z\geq 2$ \cite{masaoka2024rigorous}. Our observation of $z<2$ thus means that the critical points of the squeezing codes define {\it intrinsically non-equilibrium} universality classes. The aforementioned $\ep$ expansion calculations suggest that these critical points are strongly coupled, and describing them theoretically presents an interesting challenge for future work. 
	
	\sss{Synchronicity-protected order} 
	
	Our final results concern the role of synchronicity, where in Sec.~\ref{sec:synch_protection} we consider what happens when deviations from perfect synchronicity are introduced into the dynamics. To do this, we consider a discrete-time model whereby at each time step, any given site has a probability $\a \in [0,1]$ of being updated (see e.g. \cite{fates2007asynchronism,pajouheshgar2025exploring}). $\a=1$ corresponds to perfectly synchronous updates, while the $\a \ra 0$ limit corresponds to perfectly asynchronous ones. Taking $\a = 1-\ep$ for small $\ep$ may be regarded as adding small errors into the global synchronizing clock. 
	
	Recall the $\sfT$ squeezing code referred to above: it successfully operates as a robust memory when $\a=1$, but fails (catastrophically, as we will see) to be a memory when $\a \ra 0$. What then occurs at intermediate $\a$? We prove in Sec.~\ref{sec:memories} that $\sfT$'s memory phase is stable for small $\ep$, and so there must be a synchronicity-induced phase transition at an intermediate critical value of $\a$. In Sec.~\ref{sec:synch_protection}, we discuss this phenomenon and show that the critical point occurs at $\a_c \approx 0.38$. 
	Beyond the $\sfT$ automaton, many other squeezing codes are synchronicity-protected in this sense: this is the case for $\sfR$ in dimensions $d>2$, and in fact {\it all} squeezing codes are synchronicity protected when $d$ is odd. These examples complement the recent constructions of {\it asynchronicity}-protected memories in  \cite{pajouheshgar2025exploring}. 
	
	%	In the future, it will be interesting to understand more precisely the power of synchronicity as a resource, the nature of the transitions that occur as $\a$ is tuned, and how field theories can be developed for magnetization dynamics in synchronously-updated memories. 
	
	\begin{figure*}
		\includegraphics[width=.85\tw]{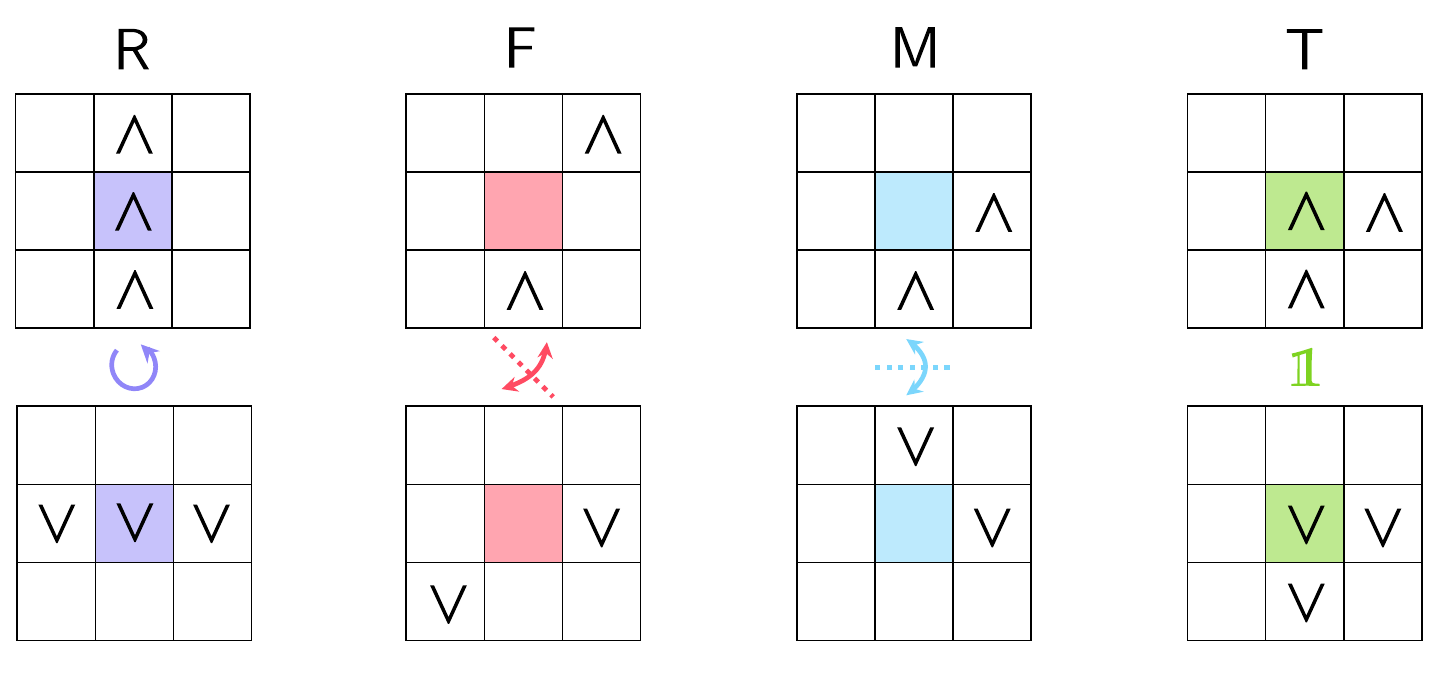} \caption{\label{fig:rules} Schematic definitions of the squeezing codes studied in this work. In each panel, the colored site (square) is updated to be the $\wedge$ (${\tt and}$) or the $\vee$ (${\tt or}$) of the indicated squares, as appropriate.  }
	\end{figure*}
	
	\section{Constructing squeezing codes}\label{sec:construction} 
	
	In this section, we define in detail the dynamics studied in the remainder of the paper. We will specialize to 2D dynamics for now, with generalizations to higher dimensions deferred to Sec.~\ref{sec:meanfield}. 
	
	\ss{Dynamics and symmetries}\label{ss:setup}

	For synchronous updates and in the absence of noise, a general squeezing code $\mca$ is defined by updates of the form (recall that $s=+1$ is associated with \texttt{true} and $s=-1$ is associated with \texttt{false})
	\be \label{general_sqz}\mca \, : \, s_\bfr(t+1) = \begin{dcases} \bigwedge_{\bfdel \in \mcr^\wedge } s_{\bfr+\bfdel}(t) & t \in 2\zz \\ 
		\bigvee_{\bfdel \in \mcr^\vee } s_{\bfr+\bfdel}(t) & t \in 2\zz+1 \end{dcases} \ee 
	where $\mcr^\wedge,\mcr^\vee$ are collections of vectors defining the particular dynamics; for the $\sfR$ squeezing code introduced in the previous section, $\mcr^\vee = \{\uvx,\bfzero,-\uvx\}$ and $\mcr^\wedge = \{\uvy,\bfzero,-\uvy\}$. For asynchronous updates, sites are updated according to independent Poisson clocks, and when an update occurs, $\mca$ performs either a $\vee$ or a $\wedge$ update, with equal probability. This definition ensures that the dynamics generated by $\mca$ can always be associated to a continuous-time Markov process. 
	
	All of the dynamics considered in this work will possess a certain type of $\zt$ symmetry in the absence of noise. This symmetry arises from choosing the regions $\mcr^\vee,\mcr^\wedge$ so that 
	\be \label{rrreln} \mcr^\wedge = C(\mcr^\vee), \qq \mcr^\vee = C(\mcr^\wedge),\ee 
	where $C$ is a symmetry of the square lattice, and is such that $C^2(\mcr^{\wedge/\vee}) = \mcr^{\wedge/\vee}$ (for $\sfR$, $C = R_{\pi/2}$). Letting $X$ denote the spin-flip operation as in \eqref{xdef}, the relation \eqref{rrreln} then means that the dynamics is symmetric under the operation 
	\be \label{xcdef} X_C = X \circ C,\ee 
	which combines a spin flip with a spatial symmetry. 
	The intertwining of internal and spatial symmetries is responsible for many of the most interesting phenomena in active matter (e.g. \cite{solon2015flocking,bowick2022symmetry}), and it will be seen to have similarly important consequences here. 
	
	We will see in Sec.~\ref{sec:memories} that $\mca$ is not a robust memory if $C = R_\pi$, but robust memories can be constructed for all other choices of $C$. In particular, $C$ may be either a $\pi/2$ rotation, a mirror reflection through a coordinate axis, a combination of the two, or the identity. The $\sfR$ automaton provides an example where $C$ is a $\pi/2$ rotation, and now we turn to defining three other automata $\sfF,\sfM,\sfT$ which use the other three possible choices of $C$. 
	
	The first rule $\sfF$ has $C$ equal to a mirror reflection through the line $x=-y$, and is defined by 
	\be \sfF \, : \, s_\bfr(t+1) = \begin{dcases} s_{\bfr+\uvx} \vee s_{\bfr - \uvx -\uvy} & t\in 2\zz  \\ 
		s_{\bfr +\uvx +\uvy} \wedge s_{\bfr- \uvy} & t \in 2\zz+1 \end{dcases} \ee 
	
	The rule $\sfM$ is defined with $C$ a mirror reflection through the $x$ axis: 	
	\be \sfM \, : \, s_\bfr(t+1) = \begin{dcases} s_{\bfr+\uvx} \vee s_{\bfr + \uvy} & t\in 2\zz  \\ 
		s_{\bfr +\uvx} \wedge s_{\bfr - \uvy} & t \in 2\zz+1 \end{dcases} \ee 
	
	Finally, the $\sfT$ rule has $C = \unit$ to be trivial:\footnote{The presence of $s_\bfr$ on the RHS is important: without it, $\sfT$ would be effectively one-dimensional under synchronous updates. It is easy to show that such automata cannot be robust memories.}
	\be \sfT \, : \, s_\bfr(t+1) = \begin{dcases} s_\bfr \vee s_{\bfr+\uvx} \vee s_{\bfr - \uvy} & t\in 2\zz  \\ 
		s_\bfr \wedge 	s_{\bfr +\uvx} \wedge s_{\bfr - \uvy} & t \in 2\zz+1 \end{dcases}.\ee 
	$\sfT$ possesses the same symmetries as Toom's rule, but it will be seen to have very different physics. 
	
	Schematic illustrations of these definitions are given in Fig.~\ref{fig:rules}.

	\ss{Erosion} \label{ss:erosion} 
	
	\begin{figure*}
		\centering
		\includegraphics[width=\tw]{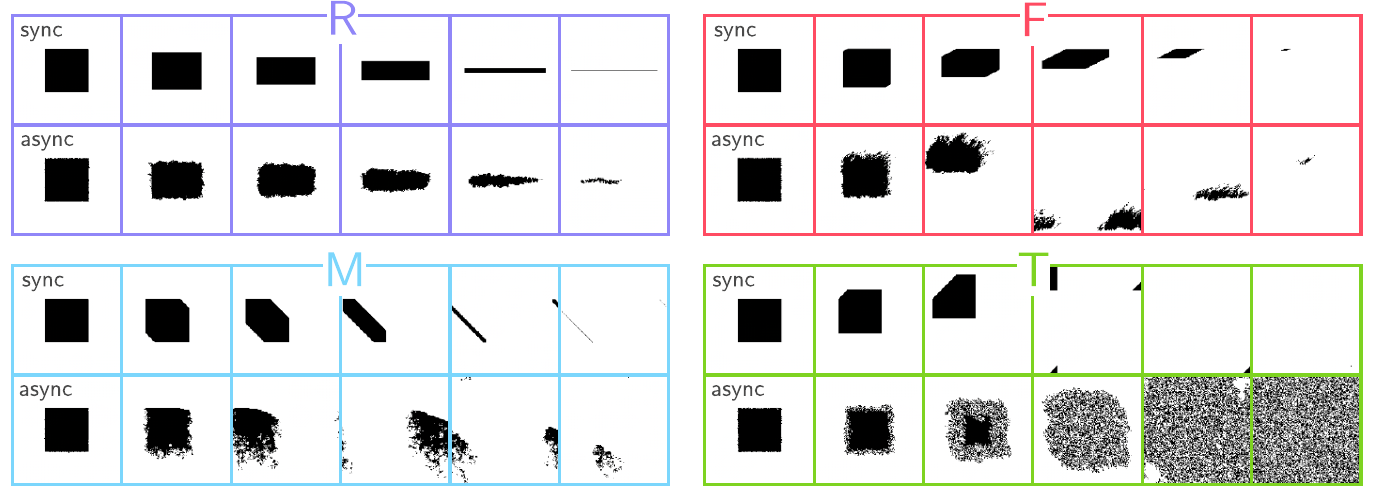}
		\caption{\label{fig:megafig} 
			The erosion of minority domains of $1$s in systems of size $L = 350$ under the different types of squeezing dynamics studied in this paper. Time runs left to right, and different time scales are used in each row. The top rows in each panel show synchronous updates, and the bottom rows show asynchronous updates. The rule $\sfT$ fails to erode the minority domain under asynchronous updates, with the domain boundary rapidly roughening and disordering the system. }
	\end{figure*}
	
	% file names like data/ca_fsqz_h_history_L350_p0.0_eta0.0_alpha0.0.jld2
	
	To understand how the $\sfR,\sfF,\sfM,\sfT$ squeezing codes perform error correction, we examine how they correct minority domains of errors. As a first step, in Fig.~\ref{fig:megafig} we visualize how noiseless dynamics acts on states containing a single square-shaped minority domain of $+1$ spins, for both synchronous (top rows) and asynchronous (bottom rows) updates (the correction of $-1$ minority domains is related by acting with $X_C$).  
	
	Under synchronous updates, the squeezing codes erase minority domains as follows (see also \href{https://pajouheshgar2025exploring.github.io/2D/floq/}{this interactive demo}): 
	\begin{itemize}
		\item $\sfR$ erases domains by squeezing them along the $\uvy$ direction.
		\item $\sfF$ erases domains by squeezing along $\uvy$ and shearing along $2\uvx+\uvy$.
		\item $\sfM$ erases domains by compression along $\uvx+\uvy$
		\item $\sfT$ erases domains in a way similar to a moving version of Toom's rule: domains are eroded from, and are ballistically displaced towards,  their northwest corners. 
	\end{itemize}
	
	Now consider asynchronous updates. As shown in Fig.~\ref{fig:megafig}, the situation for $\sfR$ is similar to the case of synchronous updates: minority domains roughen due to asynchronicity, but are still squeezed in a similar fashion. The error-correcting dynamics displayed by $\sfF$ are also not greatly modified by breaking synchronicity, although certain boundaries of the minority domain become very rough. For $\sfM$, asynchronicity causes certain domain walls to become extremely rough, with the overall shape of the minority domain being quite different from that for synchronous updates. The difference is most extreme for the $\sfT$ automaton, where asynchronicity causes the boundary of the minority domain to rapidly disorder, with the disordered region ballistically expanding across the entire system. The $\sfT$ automaton thus relies crucially on synchronicity to perform error correction.  
	
	To make our discussion more precise, we define the notion of an {\it eroder} as follows: 
	\begin{definition}[eroders]
		Let $s_{\pm,A}$ be a state which differs from a logical state $s_{\pm}$ on a  contiguous region $A$. A CA rule $\mca$ is a {\it synchronous eroder} if for all finite regions $A$, there is a finite time $t_A$ such that under synchronous updates, the errors on $A$ are erased by $t_A$:\footnote{In the math literature, one normally requires that this hold only for states $s_{-,A}$. Here we require that it hold for both $s_{+,A}$ and $s_{-,A}$.}
		\be \mca^t(s_{\pm,A}) = s_{\pm1} \,\,\forall\,\, t \geq t_A.\ee 
		
		$\mca$ is a {\it linear eroder} if domains are erased in a time proportional to their linear size, viz. if there is a constant $v$ such that if $A$ is contained in an $\infty$-norm ball of radius $R_A$, then 
		\be t_A \leq R_A / v.\ee  
		
		Finally, $\mca$ is said to be an {\it asynchronous eroder} if, under asynchronous updates, the probability that $\mca^t(s_{\pm,A})$ is not equal to $s_{\pm}$ is exponentially small in $t/t_A$; i.e. if there is a constant $c$ such that 
		\be \sfP[\mca^t(s_{\pm,A}) \neq s_\pm] < ce^{-t/t_A}.\ee 
		$\mca$ is analogously said to be an asynchronous linear eroder if there is a constant $v$ such that $t_A \leq R_A / v$.
	\end{definition}
	
	Naively, one might expect that $\mca$ must be a linear eroder if it is to be a robust memory. The argument for this is quite simple, and was essentially given above: consider a state $s_{+,A}$, where $A$ is a large (but finite) round domain. Under noise biased to favor $-1$ spins, the noise will cause this domain to ballistically grow. If the error-correcting dynamics does not cause the domain to ballistically shrink at a larger speed, it will grow and destroy the encoded logical information, and hence linear erosion is required. 
	
	With the present definition of eroder---which distinguishes $s_\pm$ as logical states---this intuition is in fact not true: there exist non-eroding automata which are nevertheless robust memories, as was shown (numerically) very recently in Ref.~\cite{pajouheshgar2025exploring}. These automata do not possess $s_\pm$ as absorbing states, and their logical states are instead probability distributions over similar states with non-extremal magnetization. Nevertheless, we expect that if $\mca$ has $s_\pm$ as absorbing memory states, as is the case in the present work, then linear erosion is indeed necessary for $\mca$ to be a robust memory.
	
	We will rigorously prove in Sec.~\ref{sec:memories} that $\sfR,\sfF,\sfM,\sfT$ are all synchronous linear eroders (as is indicated in the top rows of Fig.~\ref{fig:megafig}), and by a result of Toom, this will be shown to imply that these automata are robust memories under synchronous updates. 
	To establish what happens under asynchronous updates, we turn to numerics.\footnote{For certain types of asynchronous updates, this is not necessary. For example, if the update time intervals at each site are distributed uniformly on $[a,b]$ with $a,b > 0$ (which yields {\it non-Markovian} continuous-time dynamics), it is possible to give a rigorous proof of linear erosion  \cite{active_message_passing}. } Let us define the expected  erosion time as 
	\be \lan t_{\sf erode}(R) \ran = \max_{\pm}\( \oEE \min \{ t \, : \, \mca^t(s_{\pm,R}) = s_{\pm} \}\) ,\ee 
	where the expectation is over random update schedules, and $s_{\pm,R}$ is a state which equals $\mp1$ in the ball $r < R$ and equals $\pm1$ everywhere else. 
	%	(by monotonicity, considering only the erosion of balls is sufficient to determine erosion of arbitrary types of errors). 
	
	\begin{figure}
		\includegraphics[width=.47\tw]{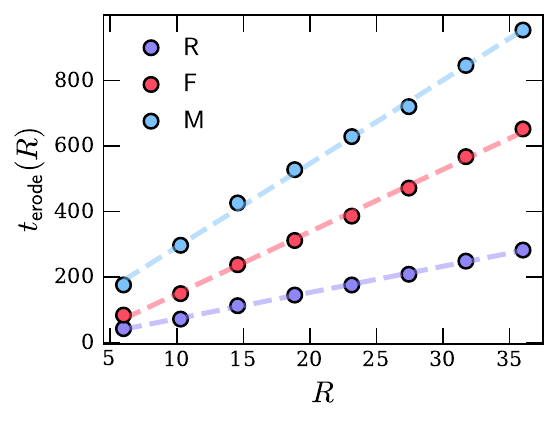} 
		\caption{\label{fig:erosion} Linear erosion under asynchronous updates in the automata $\sfR,\sfF,\sfM$. $\lan t_{\sf erode}(R)\ran$ is the average time taken by noiseless asynchronous dynamics to erode a circular domain of size $R$. Each data point is an average over 1000 samples on systems of size $L = 50R$, and the dashed lines are linear fits.}
	\end{figure}
	% p ca_plotter.py -fin data/ca_rsqz_h_erosion_L120_.jld2 data/ca_fsqz_h_erosion_L120_.jld2 data/ca_msqz_h_erosion_L120_.jld2
	
	\begin{figure}
\includegraphics[width=.47\tw]{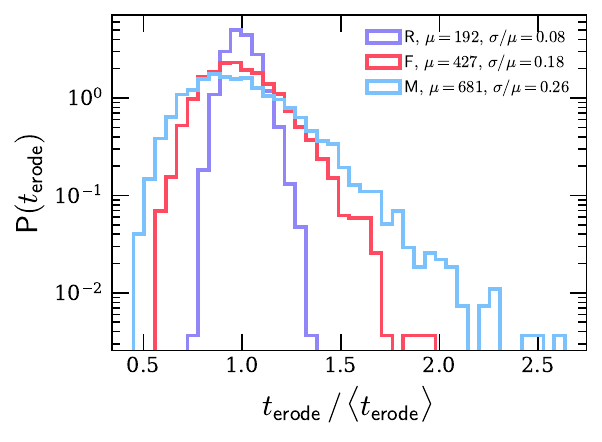} 
\caption{\label{fig:erosion_concentration} Concentration of erosion times. Normalized histograms of erosion times are shown for asynchronously updated automata with initial states containing a domain of radius $R = 25$ in systems of size $L =  1250$. $5000$ runs of the dynamics are used for each automaton. In the caption, $\mu$ is the average erosion time and $\s$ the standard deviation.  }
	\end{figure}
	
	As shown in Fig.~\ref{fig:erosion}, for each of the squeezing codes $\sfR,\sfF$ and $\sfM$, $\lan t_{\sf erode}$ scales linearly with $R$. Fig.~\ref{fig:erosion_concentration} furthermore shows that for these automata, the erosion time is exponentially concentrated about its mean. These automata are thus all asynchronous linear eroders, and from Fig.~\ref{fig:erosion} they can be seen to have erosion velocities ordered as 
	\be \label{erosion_speeds} v_\sfR> v_\sfF > v_\sfM\ee 
	$\sfR$ is thus the ``best'' asynchronous eroder (and from Fig.~\ref{fig:erosion_concentration} has the most tightly concentrated erosion times), while $\sfM$ is the ``worst'' (and has the least tightly concentrated). 
	
	As discussed above, $\sfT$ is observed not to be an asynchronous eroder at all. In fact, $\sfT$ exhibits a synchronicity-induced transition between an eroder at large amounts of synchronicity and something which rapidly reaches a trivial disordered state at small values of synchronicity; the details are explained in Sec.~\ref{sec:synch_protection}.

	\ss{Realization in Floquet Glauber dynamics}	\label{ss:glauber} 
	
	The update rules presented thus far may appear artificial from a physical point of view. However, it turns out that asynchronous dynamics under at least the $\sfR$ automaton admits a natural physical realization as Glauber dynamics under a time-periodic Ising model, which we now briefly explain (similar constructions for the other rules are also possible).
	
	The basic idea is as follows. Consider an Ising model with Hamiltonian 
	\be H = -J_x \sum_\bfr s_\bfr s_{\bfr+\uvx} - J_y \sum_\bfr s_\bfr s_{\bfr+\uvy} - h \sum_\bfr s_\bfr,\ee 
	and consider trying to engineer $H$ to implement the $\vee$ part of $\sfR$ dynamics, whereby $+1$ spins spread along $\pm\uvx$. This can be accomplished by taking $h<0$, and letting $J_x < J_y$. Having $h<0$ favors minority domains of $+1$ spins to expand, but since $J_x < J_y$, the domains expand more rapidly along the $\pm \uvx$ directions than they do along $\pm\uvy$ (growing a flat domain along $\pm\uvx$ requires roughening the domain surface by introducing domain walls normal to $\uvy$). With this choice of $H$, $+1$ domains thus grow with both directions, but expand along $\pm\uvx$ more quickly. 
	To implement the $\wedge$ part of the $\sfR$ dynamics, we then (by $X_{R_{\pi/2}}$ symmetry) need only choose $h>0$ and $J_x > J_y$: with these parameters, $+1$ domains shrink along both directions, but are eroded more quickly along $\pm\uvy$. Alternating between these two choices of $H$ sufficiently quickly thus allows the full squeezing dynamics to be simulated.
	
	To be completely explicit, consider performing Glauber dynamics under the Hamiltonian 
	\bea \label{floquet_glauber} H(t) & = -( \ob J +\wt J \cos(\o t)) \sum_i s_i s_{i + \uvx} \\ 
	& - (\ob J - \wt J \cos(\o t))\sum_i s_i s_{i + \uvy} \\
	& - \wt h \cos(\o t)\sum_i s_i.\eea 
	If we choose the parameters such that 
	\be 2(\ob J-\wt J) < h < 2(\ob J + \wt J),\ee
	then as $T\ra0$,  a $+1$ minority domain will expand along $\pm\uvx$ but stay fixed along $\pm\uvy$ when $\cos(\o t)<0$, while the opposite will occur when $\cos(\o t)>0$. In order for the time-dependent Glauber dynamics to mimic the CA dynamics, the switching should be faster than the timescale for minority bubble nucleation, which can be accomplished by taking $\o \gtrsim {\ob J} e^{-\b \ob J^2/\wt h}$. Numerics reveal the phase diagram under these types of updates to be qualitatively the same as those under asynchronous squeezing updates, and a more in-depth investigation is deferred to future work. 
	%Animations illustrating this are available at \cite{code}. 

	\section{Squeezing codes as robust memories} \label{sec:memories}
	
	In this section, we rigorously prove that the squeezing codes introduced in the previous section are robust memories under synchronous updates, and provide results from numerics which show that all but $\sfT$ are robust memories under asynchronous updates. 
	
	To make our discussion precise, we need two definitions: 
	
	\begin{definition}[memory time]
		The {\it memory time} $\tmem$ is the expected time by which the magnetization relaxes to half of its initial value, when initialized in a logical state: 
		\be \tmem =  \EE \min_\pm \{ t\, : \, \pm m[\mca^t(s_\pm)] < 1/2 \},\ee 
		where $m(s)= L^{-2}\sum_\bfr s_\bfr$ is the magnetization 
		%	\footnote{
			%		We take $ep = \O(L^0)$ since if $\mca$ simply implements i.i.d noise of strength $p$ (with no error correction), the sign of $M[\mca^t(s_\pm)]$ remains $\pm$ with high probability until a time of order $\log L$, which diverges as $L\ra\infty$. Taking $\ep>0$ gives $\tmem = {\rm const.}$ in this situation. 
			%		The exact choice of $\ep$ is not important, and for definiteness we will fix $\ep = 1/2$ in what follows. } 
		and the expectation value is over noise realizations and, if appropriate, patterns of asynchronous updates. 
	\end{definition}

	%\begin{definition}[logical error rate]
	%	We define the {\it logical error rate} $\plog$ as the probability for a magnetization flip to occur by time $O(L)$ when the system is initialized in its least stable logical state: 
	%	\be p_{\sf log} = \max_{\pm } \sfP[{\sf maj}(\mca^{cL}(s_\pm)) \neq \pm1],  \ee 
	%	where $\maj(s)$ computes the sign of the magnetization of a state $s$, and $c$ is a constant. 
	%\end{definition}
	%\be \tmem = \EE \min \{ t\, : \, \pm M[\mca^t(s_\pm)] < 0 \},\ee 
	%where $M(s)= L^{-2}\sum_\bfr s_\bfr$ is the magnetization, and  the expectation value is over noise realizations and (if appropriate) patterns of asynchronous updates. 
	
	We then define robust memories as dynamics where that $\tmem$ diverges in the thermodynamic limit $L\ra\infty$, even in the presence of a small amount of arbitrarily-biased noise: 
	\begin{definition}[robust memory]\label{def:memory}
		A CA rule $\mca$ is a {\it robust memory} if there is an open ball in the $(p,\eta)$  such that for all noise parameters in this ball,\footnote{Note that we do not place a restriction on how fast $\tmem$ must diverge with $L$, and some readers may be concerned with allowing for a weak divergence like $\tmem \sim \log L$. As an example of why one might be concerned, recall that under dynamics consisting of trivial i.i.d bit flip noise---which certainly should not count as a robust memory---a global majority vote can recover memory of a state's initial magnetization up to time $\log L$. However, our definition of $\tmem$ requires that the average magnetization be separated from the trivial value of $m=0$ by a {\it constant} amount (here $1/2$), and with this requirement, memory is lost under i.i.d noise in {\it constant} time. For this reason, we do not impose a particular requirement on how quickly $\tmem$ diverges with $L$.}
		\be \lim_{L\ra \infty} \tmem = \infty,\ee 
		with the radius of the ball remaining finite in this limit.
	\end{definition} 
    This means that in order to be a robust memory, the dynamics must be able to maintain memory for thermodynamically long times, even when {\it all} the symmetries and constraints it possesses are (weakly) broken. 
    
	A simpler definition would be to say that $\mca$ is a robust memory if $\lim_{L\ra\infty} \tmem = \infty$ as long as $p<p_c$ for some $\eta$-dependent constant $p_c$. This definition would be acceptable for the memories studied in this work (which have phase diagrams like the one in Fig.~\ref{fig:overview}~$\sfb$), but not in general: the ``open ball'' phrasing in def.~\ref{def:memory} is needed because there are some memories for which this is not true, viz. which are ordered only in an intermediate nonzero range of $p$ \cite{pajouheshgar2025exploring,marsan2025perturbed}. 
	
	\ss{Synchronous updates} \label{ss:synchproof}
	
	For synchronously updated squeezing codes, we are able to give a rigorous proof of robustness: 
	\begin{theorem}\label{thm:robust}
		The squeezing codes $\sfR, \sfM,\sfF,\sfT$ are robust memories under synchronous updates. 
	\end{theorem}
	The proof is a consequence of a general result of Toom \cite{toom1980stable} about asynchronous eroders, and to streamline the presentation, it is deferred in its entirety to App.~\ref{app:proof}. 
	
	As a complementary result, we also show that squeezing codes with $C = R_\pi$ are {\it not} robust memories:
	\begin{theorem}\label{thm:nonrobust}
		If $\mca$ is a squeezing code with $\mcr^\vee = R_\pi(\mcr^\wedge)$, then $\mca$ is not a robust memory.  
	\end{theorem} 
	This is similar to one of the results in \cite{pajouheshgar2025exploring}, which showed that automata which perform majority votes over collections of sites invariant under $R_\pi$ do not define stable memories. The proof is likewise deferred in its entirety to App.~\ref{app:proof}. 
	
	Because we can rigorously show that squeezing codes have ordered phases under synchronous updates, we will not provide detailed numerics of their performance in this setting. In general though, the threshold noise strengths for synchronous updates are (sometimes significantly) larger than for asynchronous updates. As an example, we find that $\sfR$ has a threshold of $p_{c,synch}\approx 15\%$ under synchronous updates, compared to $p_{c,asynch}\approx 3\%$ under asynchronous updates (to be demonstrated shortly). 
	
	\begin{figure*}
		\centering
		\includegraphics[width=.4\tw]{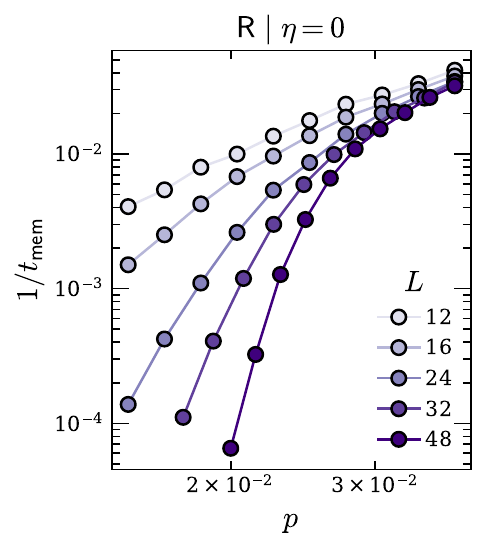}	\includegraphics[width=.4\tw]{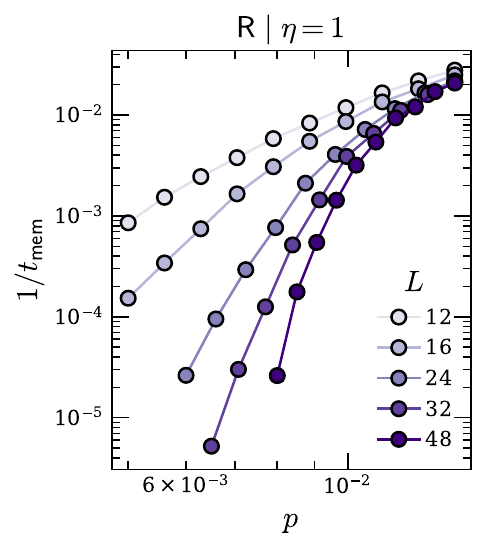}
		\caption{Phase transitions in the memory time under asynchronous $\sfR$ dynamics, at zero bias (left) and maximal bias (right). Error bars are smaller than the data points.}
		\label{fig:memory_fig}
	\end{figure*}
	
	% p ca_plotter.py -plog -fin data/ca_rsqz_h_trel_p0.005to0.014_L12_eta1.0_alpha0.0_3sqz.jld2 data/ca_rsqz_h_trel_p0.005to0.014_L16_eta1.0_alpha0.0_3sqz.jld2 data/ca_rsqz_h_trel_p0.006to0.014_L24_eta1.0_alpha0.0_3sqz.jld2 data/ca_rsqz_h_trel_p0.0065to0.014_L32_eta1.0_alpha0.0_3sqz.jld2 data/ca_rsqz_h_trel_p0.008to0.014_L48_eta1.0_alpha0.0_3sqz.jld2 -plog
	
	% p ca_plotter.py -fin data/ca_rsqz_h_trel_p0.015to0.0375_L12_eta0.0_alpha0.0_3sqz.jld2 data/ca_rsqz_h_trel_p0.015to0.0375_L16_eta0.0_alpha0.0_3sqz.jld2 data/ca_rsqz_h_trel_p0.015to0.0375_L24_eta0.0_alpha0.0_3sqz.jld2 data/ca_rsqz_h_trel_p0.0175to0.0375_L32_eta0.0_alpha0.0_3sqz.jld2 data/ca_rsqz_h_trel_p0.02to0.0375_L48_eta0.0_alpha0.0_3sqz.jld2 -plog -dynamic -pc .025 -raw
	
	\begin{figure*}
		\centering 
		\includegraphics[width=.4\tw]{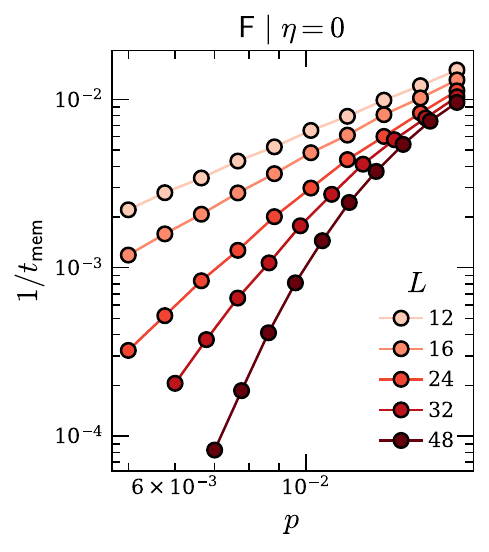} 
		\includegraphics[width=.4\tw]{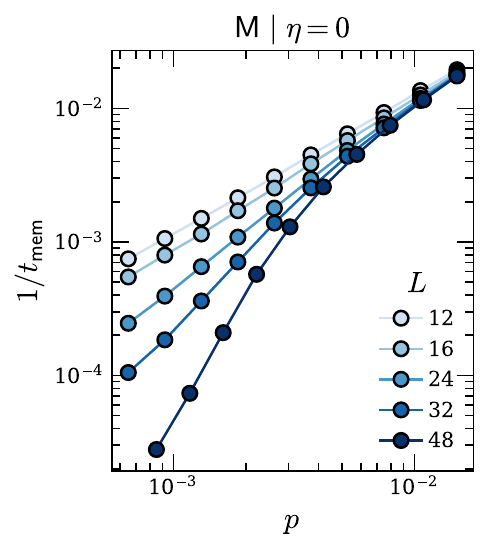}
		\caption{\label{fig:memory_fig2} Phase transitions in the memory time under asynchronous $\sfF$ dynamics (left) and asynchronous $\sfM$ dynamics (right), both at zero bias.}
	\end{figure*}
	
	% p ca_plotter.py -fin data/ca_fsqz_h_trel_p0.005to0.018_L12_eta0.0_alpha0.0.jld2 data/ca_fsqz_h_trel_p0.005to0.018_L16_eta0.0_alpha0.0.jld2 data/ca_fsqz_h_trel_p0.005to0.018_L24_eta0.0_alpha0.0.jld2 data/ca_fsqz_h_trel_p0.006to0.018_L32_eta0.0_alpha0.0.jld2 data/ca_fsqz_h_trel_p0.007to0.018_L48_eta0.0_alpha0.0.jld2 -plog
	
	% p ca_plotter.py -fin data/ca_msqz_h_trel_p0.00065to0.015_L12_eta0.0_alpha0.0.jld2 data/ca_msqz_h_trel_p0.00065to0.015_L16_eta0.0_alpha0.0.jld2 data/ca_msqz_h_trel_p0.00065to0.015_L24_eta0.0_alpha0.0.jld2 data/ca_msqz_h_trel_p0.00065to0.015_L32_eta0.0_alpha0.0.jld2 data/ca_msqz_h_trel_p0.00085to0.015_L48_eta0.0_alpha0.0.jld2  -plog

	\ss{Asynchronous updates}  \label{ss:asynch_numerics}
	
	\begin{table}[htbp]
		\centering
		\renewcommand{\arraystretch}{1.2}
		\begin{tabular}{c@{\hspace{1em}}c@{\hspace{1em}}c@{\hspace{1em}}c@{\hspace{1em}}c}
			
			$p_c$ (\%)& ${\sf R}$ & ${\sf F}$ & ${\sf M}$  \\
			\toprule 
			$\eta=0$ & 2.5 & 1.0  &  0.3
			\\ $\eta = 1$ & 1.1 & 0.4 & 0.1 \\ 
			\midrule
			\bottomrule
		\end{tabular}
		\caption{Approximate thresholds (in percent) of the $\sfR,\sfF,\sfM$ squeezing codes under asynchronous updates, for unbiased noise ($\eta=0$, top row) and maximally-biased noise ($\eta=1$, bottom row). These estimates are based on a computation of $\tmem$; more precise estimates at $\eta =0$ will be given in Sec.~\ref{sec:crit}. }
		\label{tab:thresholds}
	\end{table}
	
	Theorem~\ref{thm:robust} does not tell us about the robustness of asynchronously updated squeezing codes. Indeed, as the failure of $\sfT$ to be even an asynchronous eroder indicates, the breaking of synchronicity can have substantial effects in general, and it is not a priori obvious that $\sfR,\sfF,$ and $\sfM$ remain memories once synchronicity is broken. 
	
	To address this question, we use Monte Carlo numerics to estimate $\tmem$ for asynchronously updated $\sfR,\sfF,\sfM$ dynamics, excluding $\sfT$ since it is not an asynchronous eroder. 
	%	While our simulations compute $\tmem$, we will find it slightly more illuminating to plot instead the {\it logical error rate} $\plog$, which we here define as 
	%	\be \plog = L/\tmem.\ee 
	%	$1/\tmem$ may be thought of as the probability for a magnetization reversal (a ``logical error'') to occur during one step of the dynamics, and thus $\plog$ is a proxy for the probability for a magnetization reversal to occur by time $L$. 
	
	Before presenting the results, it is useful to first make some predictions about how $\tmem$ should scale with $p$ and $L$. For robust memories in the ordered phase, a loss of the initial magnetization requires the spontaneous nucleation of an extensively large minority domain. For each of $\sfR,\sfF,\sfM$, introducing a homologically nontrivial minority domain by flipping spins along an appropriately-oriented strip circling a cycle of the torus can always create an error that grows to span the entire system. As examples, for $\sfR$ and $\sfF$, inserting a vertical homologically-nontrivial strip of $+1$ spins into $s_-$ is guaranteed to produce a magnetization reversal under noiseless dynamics (and likewise a horizontal strip of $-1$ spins into $s_+$). For $\sfM$, inserting a homologically-nontrivial strip of $+1$ spins into $s_-$ similarly produces a reversal if the strip has boundaries parallel to $\uvx+\uvy$. Since these minority domains involve a number of flipped spins proportional to $L$, and since $\tmem$ is independent of $L$ in the disordered phase, we generically expect 
	\be \tmem = f(p) \cdot (\max(p_c/p,1))^{\a L},\ee 
	where $\a$ is a positive constant, $f(p)$ is a polynomial in $p$, and $p_c$ is the critical noise threshold. For the squeezing codes, the ordering of erosion speeds in \eqref{erosion_speeds} suggests that 
	\be \label{pc_ordering} p_{c,\sfR} > p_{c,\sfF} > p_{c,\sfM}.\ee 
	
	These predictions are consistent with numerical calculations of the memory time, shown in Figs.~\ref{fig:memory_fig}~and~\ref{fig:memory_fig2},  where $p_c$ here can be roughly identified as the point where $\tmem$ begins to increase rapidly with $L$. Fitting the scaling of $\tmem$ in the low-$p$ regime gives values of $\a$ at zero bias ranging from $\a\approx 1/3$ for $\sfR$ to $\a\approx 1/20$ for $\sfM$. The estimated values of $p_c$  in tab.~\ref{tab:thresholds} are consistent with \eqref{pc_ordering} (more precise estimates at $\eta=0$ will be given in Sec.~\ref{sec:crit} by studying cumulants of the magnetization). For each rule, we find that $p_c|_{\eta = 0}$ is approximately twice that of $p_c|_{\eta =1}$. This is roughly expected, since moving from $\eta = 0$ to $\eta=1$ at fixed $p$ doubles the number of noise events that orient spins against the majority. 
	
	\ss{Boundary conditions} 
	
	We close this section with some brief comments on the importance of boundary conditions. In all but this subsection we work with periodic boundary conditions, and this is the context within which theorem~\ref{thm:robust} was proved. In general, changing boundary conditions can significantly change how robust the memory is. As an example, consider Toom's rule, which becomes a much poorer memory under open boundary conditions (OBC).\footnote{For spins on the boundary we truncate the majority vote to only those sites within the system, and let the majority $\maj(\{a,b\})$ of two spins be a random choice between $a$ and $b$. } To see this, consider a system initialized in $s_-$ and subjected to weak noise of maximal bias $\eta=1$. As soon as the spin at the top-right corner of the lattice is subject to noise, it can never flip back to $-1$. After it is flipped, its west and south neighbors also can never return to $-1$ after flipping to $+1$. This means that the system will always reach a state with positive magnetization after time $O(L)$, giving $\tmem = O(L)$. Since our definition of a robust memory only required that $\tmem$ diverge as $L\ra\infty$, Toom's rule with OBC still counts as a robust memory---but it has a memory time exponentially shorter than with PBC.\footnote{We thank Charles Stahl for helpful discussions on this point.}
	
	For $\sfT$, the consequences of adopting OBC are similar to Toom's rule. In contrast, the squeezing codes $\sfR,\sfF,$ and $\sfM$ are more robust to adopting OBC. For these squeezing codes, boundary spins behave in essentially the same way as bulk spins, since there is no boundary site whose truncated interaction is trivial (like the northeast corner in Toom's rule). To guarantee a logical bit flip from $s_-$ to $s_+$ under weak noise with $\eta=1$, at least $L$ spins must be flipped, giving the same $\tmem \sim p^L$ scaling as was observed with PBC in the previous subsection.
	
	\section{Coarsening dynamics}\label{sec:coarsening}
	
	In this section, we investigate how the squeezing codes correct errors by studying their coarsening dynamics. We consider both the dynamics of isolated domains in the ordered state, as well as the dynamics of the magnetization following a quench from large $p$ into the memory phase. We will restrict to asynchronous updates in the entirety of this section, and since $\sfT$ is a  synchronicity-protected memory, we will restrict our attention to $\sfR,\sfF,$ and $\sfM$. 
	
	Since these models are robust memories, the dynamical exponent in the ordered phase is always $z=1$, and the coarsening dynamics is always ballistic (at least for states with only topologically trivial minority domains). This follows from the argument sketched earlier: minority domains must be ballistically eroded if they are to have a chance of being corrected in the presence of biased noise. The way that coarsening happens, however, is different in different squeezing codes, as we now explain.

	\ss{Langevin equations and domain wall dynamics} \label{ss:langevin}

	\begin{figure*}
		\includegraphics[width=.98\tw]{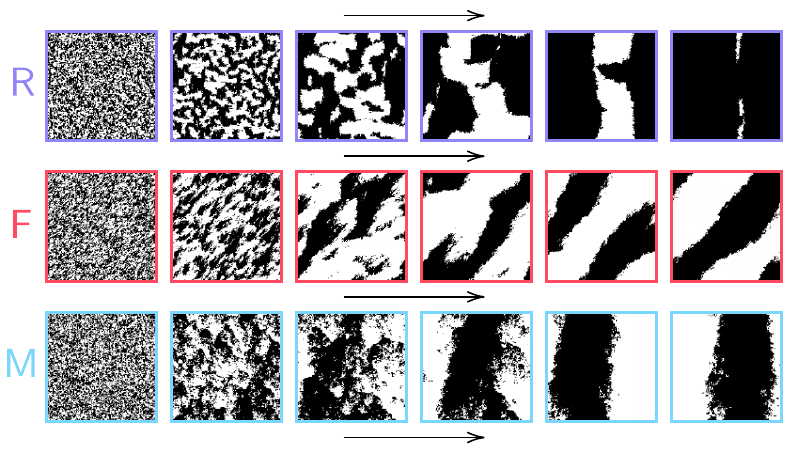}
		\caption{\label{fig:coarsening_histories} Coarsening dynamics from random initial states in systems of size $L=350$; time runs left to right. The stripes in the final panels for the $\sfF$ and $\sfM$ rules move continuously to the southeast (for $\sfF$) and the east (for $\sfM$) and are stable for thermodynamically long time scales.  }
	\end{figure*}
	
	% file names like data/ca_fsqz_h_history_L350_p0.0_eta0.0_alpha0.0_randquench.jld2
	
	We begin by using general symmetry principles to write down differential equations describing the evolution of the coarse-grained magnetization field $m$ in the memory phase under asynchronous squeezing dynamics.
	Keeping the most relevant symmetry-allowed terms that are needed to properly describe the dynamics, and omitting noise terms to lighten the notation, we have 
	\bea \label{langevins} \sfR \, : \, \p_t m & = h + \l m(\D_x^2 - \D_y^2)m - \frac{\d F_A[m]}{\d m} \\ 
	\sfF\, : \, \p_t m & = h + \g (\D_x-\D_y) m + \nu m(\D_x+\D_y) m  \\ & \qq + \l m(\D_x^2 - \D_y^2)m - \frac{\d F_A[m]}{\d m}  \\ 
	\sfM \, : \, \p_t m & = h+ \g \D_x m + \nu m\D_y m   + \l m \D_x \D_y m - \frac{\d F_A[m]}{\d m} 
	%\sfT \, : \, \p_t m & = \g (\D_x + \D_y)m  +  \mca_A[m]
	\eea 
	where $\g,\l$ are $p$-dependent constants, $h\propto \eta$, and 
	\be F_A[m] = \frac K2 (\nabla m)^2+ V[m]\ee 
	is the free energy of model-A dynamics \cite{hohenberg1977theory},\footnote{A priori the $(\D_xm)^2$ and $(\D_y m)^2$ terms in $F_A$ will not have equal coefficients for $\sfF$ and $\sfM$, but it will not be important to explicitly account for this.} with the interaction chosen as e.g.  $V[m] = \frac r2 m^2 + \frac u4 m^4$. For $\sfF$ and $\sfM$, the reason for keeping the terms proportional to $\l$---which are naively less relevant than those proportional to $\nu$---will become clear below when we discuss domain wall velocities. 	
	
	In the case of the $\sfR$ rule, this Langevin equation can in fact be derived microscopically by performing a Magnus expansion on the time-dependent Glauber dynamics model of \eqref{floquet_glauber}; the details are provided in App.~\ref{app:magnus}. 
	
	None of these Langevin equations can be obtained by varying a free energy, $i.e.$ there is no functional $F_{\rm eff}[m]$ such that $\p_t m \neq -\frac{\d F_{\rm eff}[m]}{\d m},$ {\it even} if we allow $F_{\rm eff}[m]$ to an arbitrarily non-local polynomial of the $m$. This can be derived from observing that the $\g,\l$ terms on the RHS of these Langevin equations all fail the Helmholtz integrability test, and is a formal way of showing that the type of error correction performed by the squeezing codes is not possible in equilibrium. 
	
	These Langevin equations can be used to understand the dynamics of domain walls in the memory phase. The most basic aspect of this concerns how large straight domain walls move under the dynamics.	
	To this end, consider a state with a single straight domain wall of infinite extent, with the domain wall normal to the direction $\bfn = (\cos(\theta),\sin(\theta))$ (with the orientation chosen so that $\bfn$ points in the direction of positive magnetization). By translation invariance, the dynamics must move the domain wall along $\bfn$ with a (potentially negative) angle-dependent velocity $v_\theta$, the functional form of which determines the dynamics of domain walls up to terms suppressed in the domain wall curvature. 
	
	To determine $v_\theta$, we substitute the functional form
	\be m(\bfr,t) = \sfm(\bfr \cdot \bfn - tv_\theta)\ee  
	into the Langevin equation, where the function $\sfm(x)$ is odd in $x$,\footnote{For general directions $\bfn$, the lack of symmetry means that $\sfm$ needn't be odd in general. Symmetry does however dictate that $\sfm$ is odd when $\bfn \in \{ \pm \uvx, \pm \uvy\}$ for $\sfR$, $\bfn \in \{\pm (\bfx - \bfy)/\sqrt2\}$ for $\sfF$, and $\bfn \in \{\pm \uvy\}$ for $\sfM$. } and goes quickly from $-m_0$ to $+m_0$ as $x$ increases past $0$, with $m_0$ the positive minimum of $V(m)$. We then multiply both sides by $(\bfn \cdot \nabla)m$, and integrate along the $\bfn$ direction. For $\sfR$, this yields 
	\be 0 = \int dx\, \(v_\t \sfm'^2 + h\sfm' + (F_A[\sfm])'  + \mco\), \ee 
	where primes denote derivatives along $\bfn$, and the term $\mco$ is 
\begin{equation}
	\mathcal{O} = \begin{dcases} 
		\lambda \cos(2\theta) \mathsf{m} \mathsf{m}' \mathsf{m}'' & \mathsf{R} \\[1ex]
		\begin{aligned}
			&\sqrt{2}\gamma \sin(\pi/4-\theta) \mathsf{m}'^2 + \nu \sin(\pi/4+\theta) \mathsf{m} \mathsf{m}'^2 \\
			&\quad + \lambda \cos(2\theta) \mathsf{m} \mathsf{m}'\mathsf{m}'' 
		\end{aligned} & \mathsf{F} \\[2ex]
		\gamma\cos(\theta) \mathsf{m}'^2 + \nu \sin(\theta) \mathsf{m}\mathsf{m}'^2 + \lambda \sin(2\theta) \mathsf{m}\mathsf{m}' \mathsf{m}'' & \mathsf{M} 
	\end{dcases}
\end{equation}
	Since $\sfm(x)$ minimizes $F_A[\sfm]$ as $x \ra\pm\infty$, the term involving $F_A[\sfm]$ integrates to zero. Using the fact that $\sfm'^2$ has a smoothened-delta-function functional form centered on $r=R$, and that the form of $\sfm \sfm' \sfm''$ has the same sign profile as the derivative of $\sfm'^2$, the integrals yield
	\be \label{vtr} v_\t  = -\wt h -  \wt \l \cos(2\t)  \qq (\sfR)\ee 
	for $\sfR$,
	\be \label{vtf} v_\t = -\wt h - \wt \nu \sin(\pi/4+\t) -\wt \g \sin(\pi/4-\t) -  \wt \l \cos(2\t)  \qq( \sfF) \ee 
	for $\sfF$, and 
	\be \label{vtm} v_\t = - \wt h - \wt \nu \sin(\t) - \wt \g \cos(\theta) - \wt \l \sin(2\t) \qq(\sfM) \ee 
	for $\sfM$; here $\wt h$, $\wt \nu$, $\wt \g$, and $\wt \l$ are constants proportional to $h$, $\nu$, $\g$, and $\l$, respectively.
	
	This level of analysis does not tell us how rough the domains along different directions become, but it does tell us how minority domains are eroded, and yields dynamics in qualitative agreement with the behavior observed numerically in Fig.~\ref{fig:megafig}. The interpretation for $\sfR$ is simplest: $\wt h$ causes domains to ballistically expand or contract (depending on whether they align with or against the bias), and the $\cos(2\t)$ profile of the velocity ballistically squeezes domains by anisotropically compressing them along one axis. For $\sfF$ and $\sfM$, the terms proportional to $\wt \nu$ and $\wt \g$ combine to a single term of the form $a \sin(\t + \phi)$ for constants $a,\phi$. This term alone is readily seen to induce a uniform motion of the entire domain along the $\widehat\phi$ direction at speed $a$ (accounting for the behavior seen in Fig.~\ref{fig:megafig}). The terms proportional to $\wt \l$ then ballistically squeeze the domain (in the same sense as $\sfR$ for $\sfF$, and in a $\pi/4$-rotated sense for $\sfM$), and are the terms responsible for endowing the dynamics with a robust memory. This is the reason that we kept the naively less relevant terms proportional to $\l$ in \eqref{langevins} (a term in the Langevin equation must possess at least two spatial derivatives in order to do something beyond modifying the uniform component of the domain's velocity). 
	
	 For all of the squeezing codes, the equations \eqref{vtr}, \eqref{vtf}, \eqref{vtm} eventually produce domains whose boundaries have ``corners'', regions where the domain boundary has a radius of curvature independent of the domain's size (see e.g. the examples in Fig.~\ref{fig:megafig}). Understanding the domain wall dynamics near the corners requires an analysis beyond that presented above, which we defer to future work.

	% At large $R$, the equation of motion for $R$ can be derived by multiplying both sides of each Langevin equation by $\p_r \sfm$ and integrating over $r$. As an example, consider doing this for the $\sfR$ automaton, whose Langevin equation differs from conventional model A dynamics by the presence of an anisotropic KPZ term. Using $\p_t m = -\p_t R \p_r\sfm$, we have 
	% \bea  0 & =\int_r \p_r \sfm \big( \p_tR \p_r \sfm +h + \l \sfm \p_r^2 \sfm (\cos^2(\t) - \sin^2(\t)) \\ & \qq  + \mcl_A[\sfm] + \cdots\big), \eea 
	% where the $\cdots$ are terms going as $\p_t R / r$, which will be seen to be unimportant at large $R$. The last term is a total derivative $\p_r \sfm \mcl_A[\sfm] = \p_r F_A[\sfm]$, where $F_A[\sfm] = K(\D \sfm)^2/2 + V[\sfm]$. Integrating the second term by parts and using the fact that $(\p_r \sfm)^2$ has a delta-function-like functional form centered on $r=R$ (which is why the $\cdots$ drop out at large $R$), we obtain 
	% \be \label{dw_dynamics} \p_tR= \s (\wt \l \cos(2\t)  - \wt h),\ee 
	% where $\wt h, \wt \l$ are proportional to $h$ and $\l$, respectively. The $\cos(2\t)$ on the RHS then qualitatively reproduces the correct squeezing behavior. 

	\ss{Quenching into the memory phase} \label{ss:inftquench}
	
	We now examine what happens when we perform a quench from the disordered phase into the memory phase. We do this by initializing the system in a random state, and then quenching by evolving under noiseless dynamics. As illustrated in Fig.~\ref{fig:coarsening_histories}, the time-dependent magnetization patterns generated by this process can be quite rich. 
	
	The most basic question we can ask is the probability that the system eventually reaches one of the logical states $s_\pm$. For Toom's rule, as well as zero-temperature Glauber dynamics, this probability is strictly smaller than unity. This is because both types of dynamics have a large number of absorbing states, beyond just $s_\pm$ (indeed, one easily sees that for these dynamics, any state obtained from $s_\pm$ by flipping spins along a non-contractible loop of the torus is an absorbing state). 
	
	As a first result, we show that for the squeezing codes studied in this work, the logical states are the only exact absorbing states: 
	\begin{proposition}[unique absorbing states]
		Under noiseless asynchronous updates, the squeezing codes $\sfR,\sfF,\sfM$ have only $s_\pm$ as absorbing states. 
	\end{proposition}
	\begin{proof}
		The proof is quite direct. Consider a state $s$ with $s_\bfr =-1$. If the spin at $\bfr$ is to be protected from flipping under the dynamics, the other spins in $\{\bfr+\bfdel \, : \, \bfdel\in \mcr^\vee\}$ must all be $-1$. Requiring that the $-1$ spins in this set be protected from flipping then requires that all spins in the set $\{\bfr + \bfdel + \bfdel'\, : \, \bfdel, \bfdel' \in \mcr^\vee\}$ be $-1$, and so forth. For $\sfF,\sfM$, iterating this process with periodic boundary conditions mandates that $s=s_-$, and so the only absorbing states are $s_\pm$. For $\sfR$, this instead mandates that all spins with the same $y$ coordinate as $\bfr$ be $-1$.  However, if any of the spins on adjacent rows are $+1$, then they can be flipped to $-1$ by $\mcr^\wedge$ updates. Therefore, the spins on these rows must also all be $-1$, and hence the only absorbing states are again $s_\pm$. 
	\end{proof}
	
	When we quench with noiseless dynamics, the system will thus always find its way to a logical state at infinite times. However, this fact does not a priori say anything about the time it takes to reach the logical states, which can vary greatly depending on the dynamics. 
	
	For zero temperature Glauber dynamics, it is known from percolation arguments that the probability of a quench reaching a non-logical absorbing state is approximately $1/3$ \cite{spirin2001freezing}, and a similar result can be shown to hold for Toom's rule. It turns out that the $\sfF$ and $\sfM$ squeezing codes behave somewhat similarly, although the metastable states to which they relax are dynamical: with constant probability, the dynamics reaches a state containing large topologically nontrivial stripes of minority spins that ballistically propagate, $i.e.$ ``flock'', around the system. The flocks move along $\uvx$ for $ \sfM$ and $\uvx-\uvy$ for $\sfF$, as seen in the second and third rows of Fig.~\ref{fig:coarsening_histories}. This motion can be understood from the domain wall motions derived in \eqref{vtf} and \eqref{vtm}; for example, the latter shows that domains with unit normal $\pm \uvx$ move along $+\uvx$ at constant speed. The positions of the flock boundaries move diffusively, which sets the lifetime of the metastable states (the lifetimes of these states are thus $O(L^\alpha)$ for some constant $\alpha>1$, the determination of which would require an investigation into the roughness of interfaces in the squeezing codes). 
	While violations of detailed balance often do not manifest in macroscopic dynamics, this flocking phenomenon shows that the large-scale dynamics of the $\sfM$ and $\sfF$ models can be intrinsically nonequilibrium (a feature that will reappear in the critical exponent measurements of Sec.~\ref{sec:crit}). 
	%	This phenomenon is closely analogous to the discrete flocking behavior observed in the active Ising model \cite{solon2015flocking}, and provides a direct manifestation of how the dynamics breaks detailed balance.
	
	In contrast, the $\sfR$ squeezing code {\it always} rapidly reaches a logical state: this can be directly seen from the coarsening dynamics of \eqref{vtr}, and is reflected in the first row of Fig.~\ref{fig:coarsening_histories}.
	
	\begin{figure}
		\centering 
		\includegraphics[width=.45\tw]{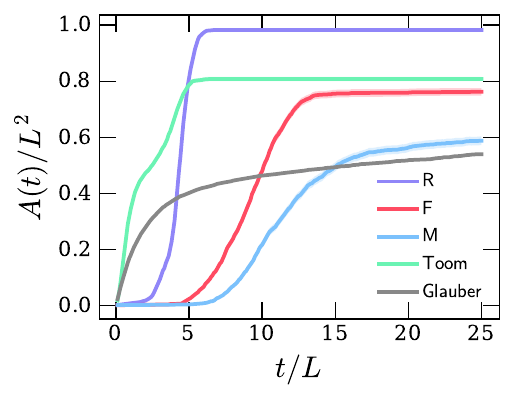} \caption{\label{fig:weighted_area} Time evolution of the average domain area following a noiseless quench from a random initial state on a system of size $L=200$, shown for the $\sfR,\sfF,\sfM$ squeezing codes, Toom's rule, and zero-temperature Ising-model Glauber dynamics. Each point on the curve is the average over 1000 independent runs of the dynamics. Only the $\sfR$ squeezing code quenches to a logical state with probability 1; all other dynamics have a constant probability of freezing out into a state that contains multiple domains for thermodynamically long times (infinitely long for Toom's rule and Glauber dynamics, and $\poly(L)$ for $\sfF$ and $\sfM$). }
	\end{figure}
	% p ca_plotter.py -plot areas -fin data/ca_rsqz_h_quench_L150_p0_eta0.0_alpha0.0_nonweighted.jld2 data/ca_fsqz_h_quench_L150_p0_eta0.0_alpha0.0_nonweighted.jld2 data/ca_msqz_h_quench_L150_p0_eta0.0_alpha0.0_nonweighted.jld2 data/ca_toom_quench_L150_p0_eta0.0_alpha0.0_nonweighted.jld2 data/ca_zeroT_glauber_quench_L150_p0_eta0.0_alpha0.0_nonweighted.jld2
	
	To demonstrate this more quantitatively, we numerically compute the average domain area
	\be \lan A(t) \ran = \left\lan \sum_{d \in {\sf Domains}(t)} {\sf Area}(d) \right \ran \ee 
	where ${\sf Domains}(t)$ is the set of domains in the system at time $t$ following the quench, ${\sf Area}(d)$ is the area enclosed by $d$, and $\lan \cdot \ran$ indicates an average over initial states and stochastic update schedules. When initialized in a random state, $\lan A(t)\ran$ is a small $O(L^0)$ number. On the other hand, when the system reaches a logical state, there is only one domain (of size $L^2$), and $\lan A(t)\ran = L^2$. 
	
	In Fig.~\ref{fig:weighted_area} we plot $\lan A(t)\ran / L^2$ for the $\sfR,\sfF,\sfM$ squeezing codes, Toom's rule, and zero-temperature Glauber dynamics for the 2D Ising model. All but $\sfR$ asymptote to a constant less than unity at long times (with Ising-Glauber taking the longest to saturate). For Ising-Glauber and Toom, these plateaus continue out to $t = \infty$. For $\sfF, \sfM$, the system will eventually reach $s_\pm$ (and hence $A = 1$) at times of order $O(L^{\alpha > 1})$, for the reasons discussed above.

	%	To study the dynamics of this process, 
	%	consider the bond density $b_{\bfr,\bfr'}= s_\bfr \oplus s_{\bfr'}$. 
	%	\be \lan \ell_{\sf dom}(t)\ran = \frac{L}{\sqrt{\sum_{\bfr\sim \bfr'} \langle b_{\bfr,\bfr'}(t) \rangle}},\ee 
	%	where $\langle \cdot \rangle$ is taken over different quenches. For a quench under noiseless dynamics from a random initial state, the domain wall density grows as a universal power law $\lan \ell_{\sf dom}(t)\ran \sim t^{1/z}$.  
	%	For curvature driven dynamics (e.g. the Ising model), one obtains $z=2$. For the squeezing code, 
	%	\be \lan \ell_{\sf dom} \ran \sim t,\ee 
	%	which we confirm numerically in the right panel of fig.~\ref{fig:coarsening}. 

	\section{Cluster mean-field and fluctuation-stabilized order}\label{sec:meanfield} 
	
	In this section, we analyze the squeezing codes in the dynamic mean field (MF) approximation. We will show that standard MF theory badly {\it underestimates} the existence of an ordered phase, falsely predicting that asynchronously-updated squeezing codes are disordered at {\it all} nonzero noise strengths. Since MF becomes a better approximation in larger dimensions, this suggests that the squeezing codes lose their ability to function as memories in larger dimensions, and we will see that this is indeed the case. We will also see how the memory phase can be recovered by including short-distance fluctuations into the MF analysis. Since a local MF approach does not appear to be possible for synchronously-updated automata, whose dynamics is not generated by a Lindbladian expressible as a sum of local operators, in the entirety of this section we will focus on the $\sfR,\sfF,\sfM$ automata under asynchronous updates. Extending the analysis to synchronous updates is an interesting problem for future work. 
	
	\ss{Squeezing codes in higher dimensions}
	
	In preparation for the discussion to follow, we first generalize the squeezing codes to general dimensions $d\geq2$. The generalization with the most natural $d\ra\infty$ limit, which we will focus on in what follows, 
	takes $d$ to be even, pairs the dimensions as $(1,2), \dots, (d-1,d)$, and subjects each pair $(i-1,i)$ to 2D squeezing dynamics. 
	%is one where $d$ is taken to be even, the dimensions are paired up as $(1,2), \dots,(d-1,d)$, and each pair of dimensions $(i-1,i)$ undergoes $2d$ squeezing dynamics. 
	This means that when a given site is chosen to be updated, a pair of dimensions and a choice of update ($\vee$ or $\wedge$) are chosen uniformly at random, the update is then applied with probability $1-p$, while noise is applied with probability $p$. 
	
	For odd $d>1$, we choose to pair up $d-1$ of the dimensions as before, and to take the remaining unpaired final dimension to undergo $\wedge$ updates. This breaks the $X_C$ symmetry of the dynamics, and we find that under asynchronous updates, this causes all of the squeezing codes to fail to be memories at any nonzero value of $p$. This is despite the fact that the arguments in the proof of theorem~\ref{thm:robust} can be used to show that these models are {\it always} memories under synchronous updates, regardless of $d$. Therefore, in odd dimensions, the $\sfR,\sfF,\sfM$ automata are synchronicity-protected in the way that the $\sfT$ automaton is in $d=2$. Since we are interested only in asynchronous updates in this section, we will assume $d$ to be even in what follows. 
	
	\ss{Master equations}
	
	To set up our MF analysis, we first derive master equations governing the dynamics of the magnetization. 
	When doing asynchronous updates, each spin has a probability $dt = 1/L^d$ of being updated at a given time step. Let the neighborhoods controlling the $\vee$ and $\wedge$ updates be denoted by $\mcr^{\vee/\wedge}_a$, where $a \in \{1,\dots,d/2\}$ labels the different two-dimensional planes on which updates happen. Then $\lan s_\bfr\ran$ evolves from one time step to the next according to 
	\bea \label{master} \p_t \lan & s_\bfr(t+dt) \ran  = -\lan s_\bfr(t) \ran + \eta p \\ & \,\,+ \frac{1-p}{d} \sum_{a=1}^{d/2} \left\lan \bigwedge_{\bfdel \in \mcr^\wedge_a} s_{\bfr+\bfdel} + \bigvee_{\bfdel' \in \mcr^\vee_a} s_{\bfr + \bfdel'} \right\ran,\eea 
	where we have taken the $dt\ra0$ limit. 
	This equation is exact, but solving it requires obtaining evolution equations for the terms like $\lan \bigwedge_{\bfdel\in \mcr^\wedge_a} s_{\bfr+\bfdel} \ran$; these involve expectations of yet higher body terms, leading to an infinite hierarchical family of equations. This infinite regress is halted using an appropriate type of mean-field ansatz, which truncates the equations after a certain weight of operator is reached, and generates a hierarchy of evolution equations that can be analyzed jointly. % which are then jointly analyzed. 

	\ss{The failure of mean-field}
	
	%Consider what happens in an initial state $s_{+,\{\bfr_0\}}$ (viz. a state which equals $+1$ everywhere except at site $\bfr_0$). How does the number of $-1$s in this state spread? The only way for $s_{\bfr_0}$ to be flipped to $+1$ is if $\bfr$ is chosen to be $\bfr_0$. On the other hand, if $\bfr_0$ is chosen to be a member of the set $\mcr$, it will infect other spins with probability $1/2$. 
	
	First, we discuss the standard MF analysis, which truncates \eqref{master} to a single equation by assuming that $n$-point functions of the spins factorize as 
	\be \label{mfass} \left\lan \prod_{i=1}^n s_{\bfr_i} \right\ran = \prod_{i=1}^n \lan s_{\bfr_i} \ran = m^n ,\ee 
	where we defined the expected magnetization $m \equiv \lan s_\bfr\ran$, with the expectation value assumed to be independent of $\bfr$ (as is the case for all of the steady states of interest). 
	
	As a warmup and to calibrate our expectations, we first perform the MF analysis for Toom's rule, for which the master equation reads 
	\be \p_t \lan s_\bfr \ran = - \lan s_\bfr\ran + \eta p + (1-p) \lan \maj(s_\bfr, s_{\bfr+\uvx},s_{\bfr+\uvy} )\ran.\ee 
	Factoring this with \eqref{mfass} gives (see App.~\ref{app:toom_mf} and \cite{lebowitz1990statistical})
	\be \p_tm = p\eta + \frac{1-3p}2m - \frac{1-p}2m^3 \quad ({\sf Toom}).\ee 
	The effective free energy\footnote{Even though we are out of equilibrium, since we have assumed $\lan s_\bfr\ran$ is independent of $\bfr$, the RHS of our MF equations are always polynomials in $m$. Such equations can always be interpreted in terms of relaxational dynamics under an appropriate free energy, simply by integrating with respect to $m$. } whose relaxational dynamics produces this equation via $\p_t m = -\p F_{\rm eff} / \p m$ is simply 
	\be F_{\rm eff} = -p\eta m + \frac{3p-1}4 m^2 + \frac{1-p}8m^4 \quad ({\sf Toom}).\ee 
	At zero bias $\eta=0$, $F_{\rm eff}$ develops multiple minima at $p_{c,MF} = 1/3$, which as expected, is an over-estimate of the true value of $p_c$ (around $14\%$). 
	
	At nonzero values of $\eta$, the MF phase diagram is obtained by finding the minimum of $F_{\rm eff}$ that is encountered by the logical state $m=-{\sf sgn}(\eta)$ aligned against the noise bias as it evolves according to the dynamic MF equation. When $\eta \neq 0$, this minimum is metastable, and a solution with ${\sf sgn}(m) = {\sf sgn}(\eta)$ always has lower free energy. The nontrivial statement about Toom's rule (and the squeezing codes) is that the lifetime of this metastable state is exponentially long (in $L$), but there is indeed no way of assessing this statement within the context of the present MF theory.\footnote{The {\it infinite} time non-equilibrium steady state is always dominated by configurations with ${\sf sgn}(m) = {\sf sgn}(\eta)$, so the phase diagram is only a dynamic one: the ``thermodynamics'' of the non-equilibrium steady states are trivial when $\eta\neq0$.} In any case, having an effective free energy with multiple minima is a necessary---but not sufficient---condition to have a memory.\footnote{Any dynamics that proceeds by performing 3-site majority votes will yield the same phase diagram, but only votes that occur in a specific geometric pattern (see \cite{pajouheshgar2025exploring}) will produce an extensively long lifetime for the metastable state. } In what follows we will assume that the metastable minima in $F_{\rm eff}$ indeed have extensively long lifetimes; for the squeezing codes this was shown analytically for synchronous updates and numerically for asynchronous ones in Sec.~\ref{sec:memories}. 
	
	%For this purpose, consider dynamics where the following happens at each time step:\footnote{For a system of $N$ spins, each step takes time $dt = 1/N$.} 
	%\begin{enumerate}
	%	\item A random site $\bfr$ is chosen to be updated.
	%	\item With probability $1-p$, a random set $\mcr$ of $n$ sites is chosen, and the spin at site $s_\bfr$ is updated as $s_\bfr \mt \bigwedge_{\bfr'\in \mcr} s_{\bfr'}$ with probability $1/2$, and $s_\bfr \mt \bigvee_{\bfr'\in\mcr} s_{\bfr'}$ otherwise. 
	%	\item Noise acts by sending $s_\bfr \mt \pm1$ with probability $(1\pm \eta)/2$.
	%\end{enumerate}
	%The corresponding master equations are  
	%\be \p_t \lan s_\bfr\ran = - \lan s_\bfr\ran + \eta p + \frac{1-p}2 \left\lan \bigwedge_{\bfr'\in \mcr} s_{\bfr'} + \bigvee_{\bfr'\in \mcr} s_{\bfr'} \right\ran,\ee 
	%where $\lan \cdot\ran$ is taken in the steady state, and includes an average over the choice of $\mcr$ (a random set of $n$ spins; since the interactions are all-to-all, there is no distinction between $\mcr^\wedge$ and $\mcr^\vee$). Using 
	We now perform the same analysis for the squeezing codes. The factorization of the RHS of \eqref{master} is accomplished using, for any set of sites $\mcv$, the relations 
	\bea \left\lan \bigwedge_{\bfr'\in \mcv}s_{\bfr'} \right\ran & =2 \lan \prod_{\bfr'\in \mcr} (1+s_{\bfr'})/2\ran - 1 \\ 
	\left\lan \bigvee_{\bfr'\in \mcv} s_{\bfr'} \right\ran & =1-2 \lan \prod_{\bfr'\in \mcr} (1-s_{\bfr'})/2\ran.\eea
	With this, the mean field factorization assumption \eqref{mfass} reduces \eqref{master} to 
	\be \p_t m = -m + \eta p + \frac{1-p}{2^{n-1}} \sum_{k=0}^{\lfloor (n-1)/2 \rfloor} {n \choose 2k+1} m^{2k+1},\ee 
	where $n = |\mcr^{\wedge/\vee}_a|$ ($n= 3$ for $\sfR$ and $n=2$ for $\sfF,\sfM$). 
	
	Consider first the rules $\sfF,\sfM$. Setting $n=2$, we obtain the extremely simple 
	\be \label{twosite_mf} \p_t m = p(\eta - m) \qq (\sfF,\sfM),\ee 
	which is the evolution equation we would obtain in the presence of noise alone. At the mean-field level, the error correction completely drops out of the dynamics, and the spins simply relax exponentially in time to the average value set by the noise bias. This corresponds to a trivial effective free energy of 
	\be F_{\rm eff} = -p\eta m + \frac p2m^2 \qq(\sfF,\sfM),\ee 
	which is, indeed, inconsistent with the existence of an ordered phase, regardless of $p$. This result is rather remarkable because MF theory usually {\it over}estimates the extent of the ordered phase (as was true for Toom's rule). Here, it misses the ordered phase entirely. 
	
	Fluctuation corrections to the standard MF factorization can, in principle, restore the existence of an ordered phase, but in $d$ dimensions, deviations from mean-field are expected to scale as $O(1/d)$, which we will confirm shortly below. Fluctuation corrections can thus yield an effective free energy with nontrivial minima only if they are larger than the $pm^2/2$ term, and therefore the critical noise strength must vanish as $d\ra\infty$ at least as $p_c \sim 1/d$ (we will see below that in fact $p_c \sim 1/d^2$). 
	
	Now consider $\sfR$ squeezing, for which $n=3$. We obtain 
	\be \label{threesite_mf} \p_t m = p\eta - \frac{1+3p}4m + \frac{1-p}4m^3 \qq (\sfR), \ee 
	corresponding to an effective free energy of 
	\be F_{\rm eff} = -p\eta m  + \frac{1+3p}8m^2 - \frac{1-p}{16}m^4\qq(\sfR).\ee 
	The $m^4$ term has the ``wrong'' sign, and the coefficient of the $m^2$ term is bounded below by $1/8$ at all values of $p$. Even as $p\ra0$, fluctuation corrections can only produce nontrivial minima in $F_{\rm eff}$ if $d$ is sufficiently small. This demonstrates that there is an upper critical dimension $d_c$, for which the model is always disordered when $d>d_c$. It is interesting to recall that in spite of this, $\sfR$ actually has a {\it higher} threshold (and larger erosion velocity) for $d=2$ than either $\sfF$ or $\sfM$.

	\begin{figure*}[t]
		\centering
		\setlength{\tabcolsep}{0pt}
		\begin{tabular}{@{}%
				m{.24\textwidth}%
				m{.24\textwidth}%
				@{\hspace{1.2em}\color[gray]{0.7}\vrule width 0.6pt\hspace{1.2em}}%
				m{.24\textwidth}%
				m{.24\textwidth}%
				@{}}
			\centering\arraybackslash {$\sfR_2 \, | \, 2d$} 
			\includegraphics[width=\linewidth]{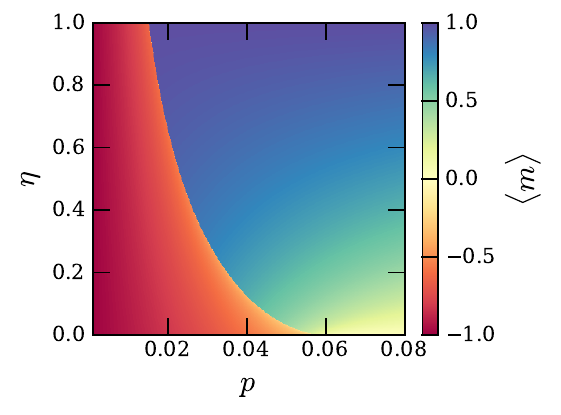} &
			\centering\arraybackslash{$\sfR_2 \, | \, 4d$} 
			\includegraphics[width=\linewidth]{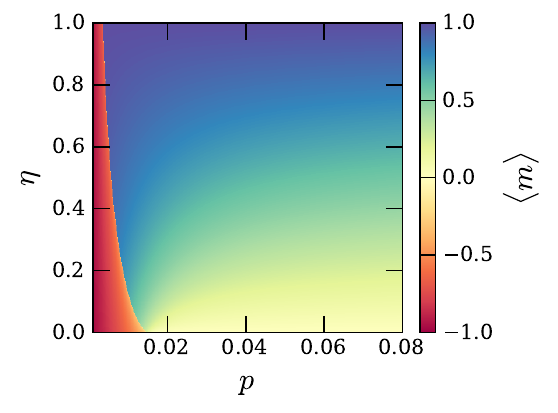} &
			\centering\arraybackslash{$\sfR_3 \, | \, 2d$} 
			\includegraphics[width=\linewidth]{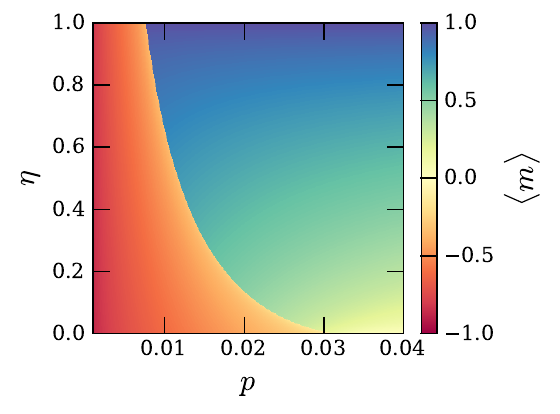} &
			\centering\arraybackslash{$\sfR_3 \, | \, 4d$} 
			\includegraphics[width=\linewidth]{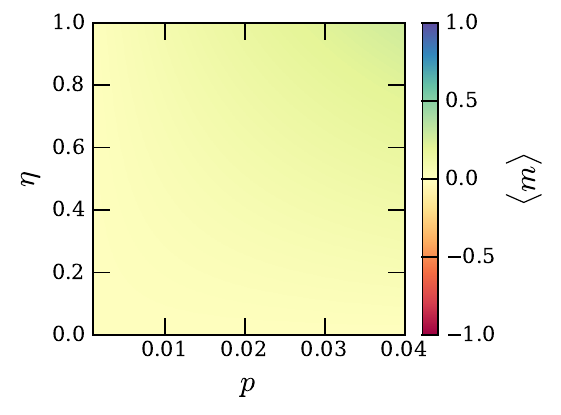}
		\end{tabular}
		\caption{Extended cluster mean-field phase diagrams in the $(p,\eta)$ plane for the  $\sfR_2$ (left) and $\sfR_3$ (right) squeezing rules in $d=2$ and $d=4$ dimensions. Each phase diagram is generated by initializing the system in $m=-1, f^+ = 0, f^-=0$ and iterating the MF equations until convergence.  }
		\label{fig:mffig}
	\end{figure*}

	\ss{Cluster expansions and the restoration of order}
	
	We now summarize calculations which show how the incorporation of fluctuations can change this conclusion, and restore the existence of an ordered phase in a way consistent with numerics. This is done by taking correlation functions to factorize beyond a range $a$ (as measured with the 1-norm): 
	\be \label{nntrunc} \lan s_\bfr s_{\bfr'} \ran = \begin{dcases} \lan s_\bfr\ran\lan s_{\bfr'}\ran & ||\bfr-\bfr'||_1 > a \\ 
		d_{\bfr,\bfr'}  & {\rm else} \end{dcases},\ee 
	where the $d_{\sfr,\sfr'}$ are additional fields that we retain in our analysis to incorporate short-scale fluctuations. In statistical physics, this approach is known as the ``cluster variational method'' \cite{kikuchi1951theory,pelizzola2005cluster}.
	
	With this assumption, any $n$-point function of the $s_\bfr$ may be expressed in terms of $\lan s_\bfr\ran$ and the $d_{\bfr,\bfr'}$ using rules for conditional probabilities (see App.~\ref{app:mf} and \cite{pelizzola2005cluster}). As an example, when $a=1$, we may factor a product of three neighboring spins as 
	\be \lan s_{\bfr-\uvx}s_\bfr s_{\bfr+\uvx}\ran =  \frac{d_{\bfr-\uvx,\bfr} d_{\bfr,\bfr+\uvx}}{\lan s_\bfr\ran}. \ee 
	
	This and similar relations let us truncate the hierarchy of equations generated by \eqref{master} to a finite number of ODEs, which we will refer to as the cluster MF equations. The derivation for the squeezing codes considered in this work is technically involved, and the details are deferred in their entirety to App.~\ref{app:mf}, where the cluster MF equations are generated systematically using the Doi-Peliti operator formalism. 
	
	The cluster MF analysis turns out to be simplest for a ``two-body'' version of the $\sfR$ squeezing code, which we will refer to as $\sfR_2$. In the language of \eqref{general_sqz}, it is defined by the regions 
	\bea \label{sqzrule_r2} \sfR_2\, : \, \mcr^\vee = \{\uvx,-\uvx\}, \qq \mcr^\wedge = \{\uvy,-\uvy\},\eea
	which differs from \eqref{sqzrule_r} only in the absence of $\bfzero$ in $\mcr^\vee$ and $\mcr^\wedge$. To better distinguish between $\sfR_2$ and the ``three-body'' rule $\sfR$, we will refer to the latter as $\sfR_3$. It is straightforward to show that $\sfR_2$ is a robust memory under synchronous updates, and has the same symmetries and qualitative coarsening behavior as $\sfR_3$. In fact, numerics reveal that $\sfR_2$ is a {\it better} memory than $\sfR_3$: in 2D with asynchronous updates and unbiased noise, it undergoes a transition at $p_{c,\sfR_2} \approx 3.9\%$ (compared with $p_{c,\sfR_3} \approx 2.5\%$). 
	
	Because the spins on the RHS of \eqref{sqzrule_r2} are next-nearest neighbors, running the cluster MF analysis on $\sfR_2$ requires setting $a=2$. Assuming that the steady states in the ordered phase do not break any spatial symmetries (an assumption which we may justify a posteriori from numerics), the cluster MF equations are most conveniently parametrized in terms of the magnetization $m = \lan s_\bfr\ran$ and linear combinations of its connected correlation functions along the $\uvx$ and $\uvy$ directions. To this end, we define 
	\be f^{l\pm} = (d_{\bfr,\bfr+l\uvx} - m^2) \pm (d_{\bfr,\bfr+l\uvy} - m^2),\ee 
	where $l\in\{1,2\}$, the choice of $\bfr$ is arbitrary by translation invariance, and where if $d>2$, we take $x,y$ to be the first two spatial directions. The $f^{l-}$ transform in the same way as $m$ under $X_C$, while the $f^{l+}$ are neutral. 
	
	\begin{widetext} 
		The calculation, given in App.~\ref{app:mf}, shows that the cluster MF equations for $\sfR_2$ are 
		\bea\label{r2mfs}
		\p_t m   & = p(\eta-m) + \frac{1-p}2 f^{2-} \\ 
		\p_t f^{1-} & = -2f^{1-} + \frac{1-p}d \(m(1-m^2) + f^{2-} - mf^{2+} \) \\ 
		\p_t f^{2-} & =  -2f^{2-} + \frac{1-p}d \big( f^{1-} + mf^{1+} + \frac1{1-m^2} \( f^{1-}f^{2+} + f^{1+}f^{2-} - m(f^{1-}f^{2-} + f^{1+} f^{2+})\)\big)\\ 
		\p_t f^{1+} & = -2f^{1+} + \frac{1-p}d \(1-m^2 + f^{2+} - mf^{2-} \) \\ 
		\p_t f^{2+} & =  -2f^{2+} + \frac{1-p}d \big( f^{1+} + mf^{1-} + \frac1{1-m^2} \( f^{1+}f^{2+} + f^{1-}f^{2-} - m(f^{1+}f^{2-} + f^{1-} f^{	2+})\)\big).
		\eea 
	\end{widetext}
	There are several things to note about these equations. First, note that they correctly reduce to \eqref{twosite_mf} upon neglecting fluctuations (viz. upon setting $f^{l\pm} = 0$). Second, one may verify that the disordered state, which has $m = f^{1-} = f^{2-} = 0$, possesses $X_C$-neutral fluctuations of size
	\be f^{1+} = \sqrt{f^{2+}} = \frac{1-\sqrt{1-\z^2}}\z, \ee 
	where $\z = \frac{1-p}d$.
	%	That $f^{2+} = (f^{1+})^2$ in the disordered state is not surprising, since next-nearest-neighbor correlation functions are mediated by two nearest-neighbor ones. 
	In the $d\ra\infty$ limit, we thus have 
	\be f^{1+} = \frac{1-p}{2d} + O(1/d^2),\ee 
	so that $f^{1+} \sim 1/d$ and $f^{2+} \sim 1/d^2$ to leading order, consistent with deviations from naive MF being suppressed as $1/d$.
	In the ordered state, one may similarly verify the scalings 
	\be m \sim d^0, \quad f^{1-} \sim 1/d,\quad f^{2-} \sim 1/d^2,\ee 
	and from this scaling and the equation for $\p_tm$, we thus conclude that $p_c \sim 1/d^2$. Indeed, one can find $p_c$ analytically by linearizing the mean field equations about the disordered state, and then determining the point at which the disordered solution becomes unstable. The details are unilluminating and will be skipped; here we quote only the result
	\be \label{r2pcmf} p_c = \frac1{1+4d^2}.\ee 
	In 2d, this predicts $p_c = 1/17 \approx 0.058$, which is about a factor of $1.5$ larger than the value observed numerically. 
	
	Given the collections of ODEs produced by the cluster MF, a phase diagram in the $(p,\eta)$ plane may be obtained by initializing the fields in the state $m=-{\sf sgn}(\eta),f^{l\pm} = 0$, and numerically updating the fields until convergence is achieved. 
	This is done for the $\sfR_2,\sfR_3$ squeezing rules in Fig.~\ref{fig:mffig}, for both two and four dimensions (the cluster MF equations for $\sfR_3$ are more complicated, and are deferred to App.~\ref{app:mf}). For $\sfR_2$, we find a phase diagram like in the schematic illustration of Fig.~\ref{fig:overview}, with the ordered phase being weaker in $d=4$ in accordance with \eqref{r2pcmf}. For $\sfR_3$, we find an ordered phase in $d=2$ and a completely disordered phase diagram in $d=4$. For $\sfR_3$ the cluster MF analysis is remarkably accurate: in 2D it predicts a phase transition at $p_{c} \approx 3\%$ (c.f. the estimate of $p_c\approx 2.5\%$ from Fig.~\ref{fig:memory_fig}), and in $4d$ its prediction of $p_c = 0$ can be numerically verified to be correct (in fact, in $4d$ $\sfR_3$ is not even an eroder---it is synchronicity-protected in the same way that $\sfT$ is in 2D).

	\section{Phase transitions and criticality}\label{sec:crit}
	
	We now investigate the phase transitions that occur in the squeezing codes as the memory is lost at zero bias, focusing exclusively on asynchronous updates (synchronous updates have different physics, and will be studied in detail elsewhere). This transition is numerically observed to be continuous for all of  $\sfR_2,\sfR_3,\sfF$, and $\sfM$. 
	
	\ss{Expectations}
	
	Before we begin, we set our expectations by reviewing what is known about non-equilibrium systems where a single non-conserved order parameter undergoes a continuous ordering transition.
	
	The consensus in the literature is that in such situations, the critical point is usually described as an effectively equilibrium system, with the non-equilibrium nature of the dynamics being irrelevant in the RG sense.\footnote{The Langevin equations derived in Sec.~\ref{ss:langevin} contain terms proportional to a single spatial derivative of the magnetization. While these are naively relevant by power-counting, they can be absorbed in a shift of the time coordinate, and do not affect the long-distance scaling of correlation functions.  } 
	An early work coming to this conclusion in the context of noisy cellular automata with conventional $\zt$ spin-flip symmetry is \cite{grinstein1985statistical}, where it was argued that phase transitions in such models belong to the dynamic Ising universality class, $i.e.$ critical model-A dynamics \cite{hohenberg1977theory}. This conclusion is known from numerics to be correct for Toom's rule, at least with regards to static correlation functions (we will discuss dynamic correlation functions in Toom's rule shortly). After this work, Ref.~\cite{bassler1994critical} used a 1-loop epsilon expansion calculation in $d=4-\ep$ to argue that model-A dynamics is stable with respect to {\it any} perturbation which keeps the model at criticality, even those which break symmetries and detailed balance (see also \cite{tauber2002effects}). Indeed, at present, the only exceptions to this scenario known to the authors require fine-tuning to a multicritical point \cite{young2020nonequilibrium,agrawal2024dynamical}. If this conclusion holds for squeezing codes, we should therefore expect model-A exponents, viz. a correlation length exponent of $\nu_A = 1$, a magnetization exponent of $\b_A = 0.125$, and a dynamic exponent of $z_A \approx 2.167$ (the latter being known from numerics and a 5-loop $\ep$ expansion calculation \cite{adzhemyan2022dynamic,nightingale1996dynamic}).

	\begin{table}[htbp]
		\centering
		\renewcommand{\arraystretch}{1.2}
		\begin{tabular}{c@{\hspace{1em}}c@{\hspace{1em}}c@{\hspace{1em}}c@{\hspace{1em}}c}
			
			& ${\sf R}_2$ & ${\sf F}$ & ${\sf M}$  & model-$A$ \\
			\toprule
			$p_c$ ($\%$) & $3.828(6)$  & $1.1431(15)$ & $0.3186(24)$ & n/a \\
			$\nu$ & $0.952(11)$   & $0.972(16)$   & $0.99(4)$  & $1.0$ \\
			$\b$ & $0.165(5)$ & $0.1826(20)$ & $0.227(5)$ & $0.125$ \\
			$z$   & $1.942(3)$ & $1.668(19)$  & $1.396(16)$ & $\approx2.167$  \\
			\midrule
			\bottomrule
		\end{tabular}
		\caption{Estimates of $p_c$ and critical exponents for the asynchronously updated automata at zero bias, compared with the critical exponents of model-$A$ dynamics. $p_c,\nu$, and $\b$ are obtained from finite-size collapses of $B$ and $m$ at different values of $p$ (Fig.~\ref{fig:static_stuff}), and the reported value of $z$ is obtained from collapsing dynamic quenches of $B$ at $p=p_c$ (Fig.~\ref{fig:Bt_collapses}; see also Fig.~\ref{fig:mt_quench}). The value of $z$ for model-A dynamics comes from \cite{adzhemyan2022dynamic,nightingale1996dynamic}. }
		\label{tab:exponents}
	\end{table}
	
	\ss{Statics} \label{ss:statics}

		\begin{figure*}
		\centering
		\includegraphics[width=0.32\textwidth]{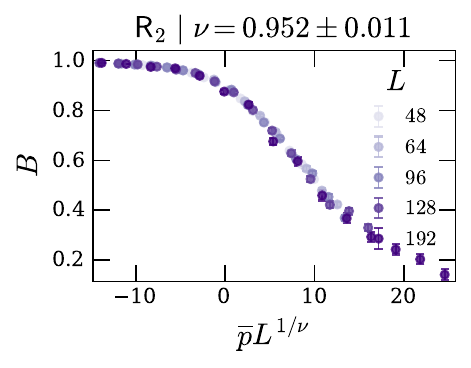} \includegraphics[width=.32\tw]{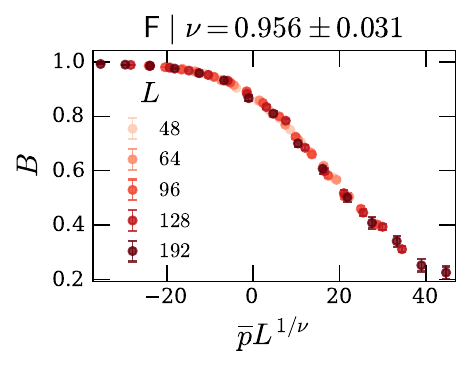}
		\includegraphics[width=.32\tw]{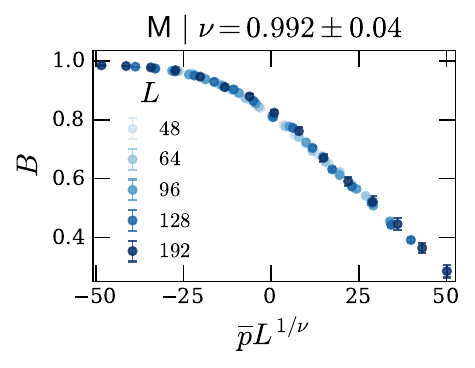} \\ 
		\includegraphics[width=0.32\textwidth]{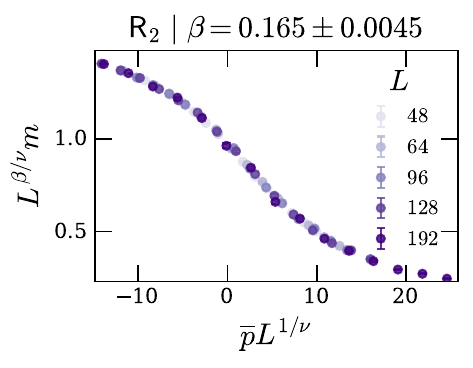} \includegraphics[width=.32\tw]{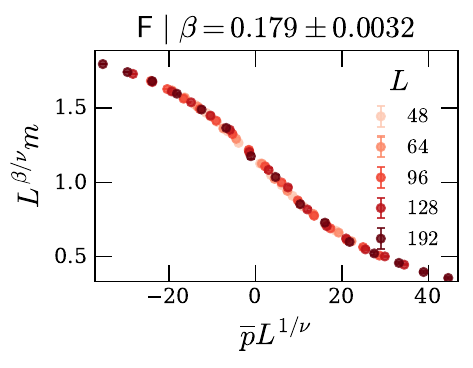}
	\includegraphics[width=.32\tw]{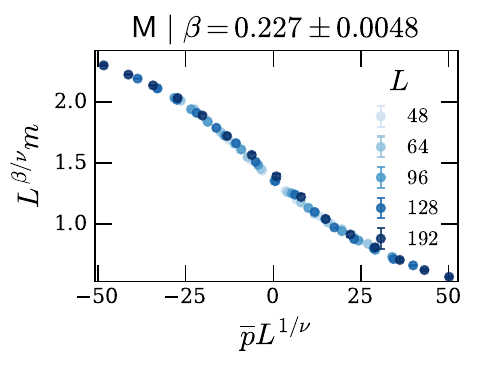}
		\caption{\label{fig:static_stuff} Scaling collapses of the Binder cumulant (top row) and magnetization (bottom row) for the $\sfR_2,\sfF$, and $\sfM$ squeezing codes (left to right).  {\it Top row:}  Scaled Binder cumulants near the critical point, with the estimated values of $\nu$ and their uncertainties shown in the subfigure headings. {\it Bottom row:} Scaled magnetization near the critical point, with the values of $\b$ shown in the subfigure headings. All code and data needed to reproduce these plots, and all other plots in this section, are available at \cite{code}.}
	\end{figure*}

	We first estimate the critical exponents associated with static correlation functions. By hyperscaling only two are independent, and we will therefore focus only on obtaining $\nu$ and $\b$. $\sfR_2$ and $\sfR_3$ appear to have very similar values of $\nu$ and $\b$ (which is unsurprising on account of their identical symmetries), and as such we will only present data for $\sfR_2$. The procedures we use for extracting critical exponents and determining the uncertainties thereof are described in detail in App.~\ref{app:numerics}.
	
	We identify $p_c,\nu$ and $\b$ by performing a scaling collapse on the magnetization $m$ and the Binder cumulant
	\be \label{bdef} B = \frac32 - \frac{\lan m^4 \ran}{2\lan m^2\ran^2},\ee 
%	and the susceptibility 
%	\be \chi = L^2 (\lan m^2 \ran - \lan m \ran^2),\ee 
	where in this section we define 
	\be m = \frac1{L^2} \left|\sum_\bfr s_\bfr  \right|\ee 
	as the absolute value of the magnetization. Our definition in \eqref{bdef} ensures that $B = 1$ at 
	the ordered phase with 
	$p=0$, that $B= 0$ at 
	the disordered phase with
	$p=1$. Right at the critical point, $B$ is a universal $L$-independent number whose value is determined by the universality class of the transition. 
	
	In Fig.~\ref{fig:static_stuff}, we plot scaling collapses of $B$ and $m$ for $\sfR_2,\sfF,$ and $\sfM$. Both $B$ and $L^{\b / \nu} m$ are plotted against $\ob p =  (p/p_c-1)$, with $p_c,\nu,\b$ simultaneously optimized to give the best possible collapse (for details, see App.~\ref{app:numerics}). This analysis yields the values collected in Tab.~\ref{tab:exponents}. In App.~\ref{app:numerics}, we also perform collapses that simultaneously include the (slightly noisier) magnetic susceptibility $\chi = L^2 (\lan m^2 \ran - \lan m\ran^2)$. Doing this only very slightly changes the values of $p_c,\nu,\b$, and yields values of $\g$ consistent with hyperscaling within error bars. 
	
	As expected from Sec.~\ref{ss:erosion}, $p_c$ is highest for $\sfR_2$ and lowest for $\sfM$, with the values of $p_c$ agreeing well with the estimates from Sec.~\ref{ss:asynch_numerics} derived from the computation of $\tmem$. The values of $\nu$ are very slightly smaller than the model-A value of $\nu_A = 1$, while the values of $\b$ are all appreciably larger than the model-A value $\b_A= 0.125$. While the values of $\nu$ and $\b$ are different for different squeezing codes, they are such that the magnetic susceptibility exponent assumes a value of $\g \approx 1.58$ for all three models (see App.~\ref{app:numerics}), to be compared with the model-A value of $\g_A = 1.75$. 
	
	As a check, we also perform a similar collapse with $\nu,\b$ fixed to $\nu_A,\b_A$, and the collapse is worse by a statistically significant amount (see  App.~\ref{app:numerics} for details). On the other hand, when we apply the same analysis to Toom's rule, our best scaling collapse (not shown) occurs at exactly model-A values (within error bars). 
	
	Summarizing, if the static correlations of the squeezing code critical points are indeed in the model-A universality class, they would need to have much stronger subleading corrections to scaling than Toom's rule. In App.~\ref{app:numerics} we investigate how the extracted values of the exponents drift with system size, and find no evidence for a gradual convergence to model-A values. We are thus led to tentatively claim that the static exponents lie in universality classes distinct from the Ising model, although future numerical work will be needed to make a definitive conclusion.

	\ss{Dynamics} 
	
		% p plotter.py --files   data/R_quench_L48_n_samples25000_T1500_p0.03828_bind_pc.jld2  data/R_quench_L64_n_samples50000_T2500_p0.03828_bind_pc.jld2 data/R_quench_L96_n_samples50000_T4000_p0.03828_bind_pc.jld2 data/R_quench_L128_n_samples50000_T5500_p0.03828_bind_pc.jld2 data/R_quench_L196_n_samples50000_T7500_p0.03828_bind_pc.jld2 --fit --Lmin-sweep --n-bootstrap 500
		\begin{figure*}
		\centering
		\includegraphics[width=0.32\textwidth]{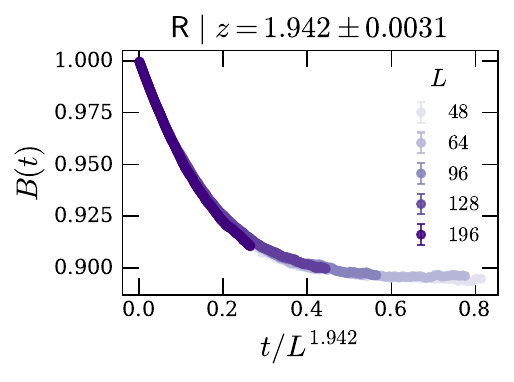} \includegraphics[width=.32\tw]{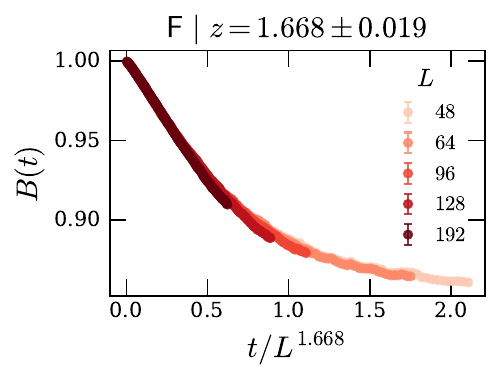}
		\includegraphics[width=.32\tw]{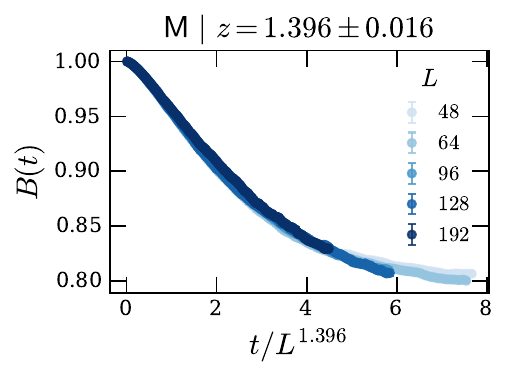} 
		\caption{\label{fig:Bt_collapses} Scaling collapses of the time-dependent Binder cumulant following a quench from the ordered state at $p=0$ to the critical point. The values used for $p_c$ are the same as those extracted from the collapses in Fig.~\ref{fig:static_stuff}. Each data point is an average of $15,000$ samples.}
	\end{figure*}

	% p ca_plotter.py -plot mt -fin data/sqztest_quench_rsqz_h_300_0.038425.jld2 data/sqztest_quench_rsqz_h_300_0.02876.jld2 data/ca_fsqz_h_quench_L300_p0.01165_eta0.0_alpha0.0.jld2 data/sqztest_quench_msqz_h_300_0.0032875.jld2 data/sqztest_quench_toom_300_0.13395.jld2 data/sqztest_quench_zeroT_glauber_300_0.141294.jld2
	
	% current quenches 
	% R3: data/sqztest_quench_rsqz_h_300_0.02876.jld2 
	% R2: data/sqztest_quench_rsqz_h_300_0.038425.jld2 
	% F: data/ca_fsqz_h_quench_L300_p0.01165_eta0.0_alpha0.0.jld2 
	% M: data/sqztest_quench_msqz_h_300_0.0032875.jld2
	% toom: data/sqztest_quench_toom_300_0.13395.jld2
	% glauber: data/sqztest_quench_zeroT_glauber_300_0.141294.jld2 

		Regardless of whether or not there is a slow drift of the static exponents to model-A values at large system sizes, results from dynamics provide a very strong indication that the squeezing code critical points indeed lie beyond the model-A {dynamic} universality class. To show this, we now  estimate the dynamic exponent $z$.
		
		\sss{Time-dependent Binder cumulant collapse}
		
		The first method we use is to perform a scaling collapse on the time-dependent Binder cumulant $B(t)$, following a quench from the ordered state at $p=0$ to the critical point at $p=p_c$ (with the value of $p_c$ identified using the collapses in Fig.~\ref{fig:static_stuff}). Since the Binder cumulant is dimensionless, on a system of size $L$ at $p=p_c$ it scales as 
		\be B(t,L) = \Psi_B(t/L^z) \ee 
		for a scaling function $\Psi_B(x)$ which begins at $\Psi_B(0) = 1$ and plateaus to a non-trivial late-time number $\Psi_B(\infty)$, whose value is characteristic of the universality
		class. With our conventions for $B$, the late-time value with periodic boundary conditions for model-A dynamics is \cite{kamieniarz1993universal}
		\be \Psi_B(\infty) \approx 0.916. \ee 
		
		The results of performing a collapse of $B(t)$ for each of the squeezing codes are shown in Fig.~\ref{fig:Bt_collapses}, and lead to the values of $z$ reported in Tab.~\ref{tab:exponents}. These values are all significantly smaller than the model-A value of $2.167$ and are all super-diffusive, each having $z<2$ within error bars. They also satisfy 
		\be \label{zordering} z_\sfM < z_\sfF < z_{\sfR_2},\ee
		so that the worse the memory (as measured by $p_c$), the smaller the value of $z$. 
		
		We additionally see that all of the squeezing codes have a value of $\Psi_B(\infty)$ less than that of model-A dynamics by statistically significant amounts: extracting the late-time asymptotics of the curves in Fig.~\ref{fig:Bt_collapses} and averaging them over $L$ gives 
		\be  \label{psibinf} \Psi_B(\infty) \approx 
		 \begin{dcases} 
		 	0.88 & \sfR_2 \\ 
		 	0.84 & \sfF\\ 
		 	0.76 & \sfM
		 \end{dcases}. \ee 
		The smaller $\Psi_B(\infty)$ is, the larger $\lan m^4\ran$ is relative to $\lan m^2 \ran^2$, and the more Gaussian the correlations. The ordering of $\Psi_B(\infty)$ in \eqref{psibinf} is thus intuitively consistent with the ordering $\b_{\sfR_2}< \b_\sfF < \b_\sfM$, since correlations away from the Gaussian fixed point lead to $\b$ being suppressed down from its mean-field value of $\b_{MF} = 1/2$. 
		
		\sss{Power-law decay in magnetization quench}
		
				\begin{figure}
			\centering
			\includegraphics[width=.49\tw]{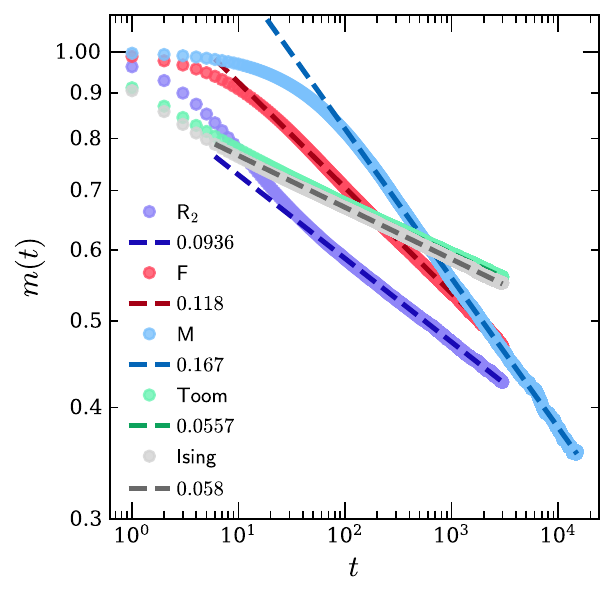} 
			\caption{\label{fig:mt_quench}  Time evolution of the magnetization following a quench from an ordered state to criticality for the squeezing codes, Toom's rule, and critical Glauber dynamics. In all cases, we set the system size at $L=300$, estimate $m(t)$ from $10^3$ samples (error bars not shown), and use the values of $p_c$ in tab.~\ref{tab:exponents}. For $\sfM$ we run the dynamics to $t_{\sf max} = 50L$; for all other automata we take $t_{\sf max} = 10L$. The dashed lines are fits to power laws $ m(t) \sim t^{-\t}$, with $\t $ as indicated in the legend.}
		\end{figure}
		
		\begin{figure}
			\includegraphics[width=.49\tw]{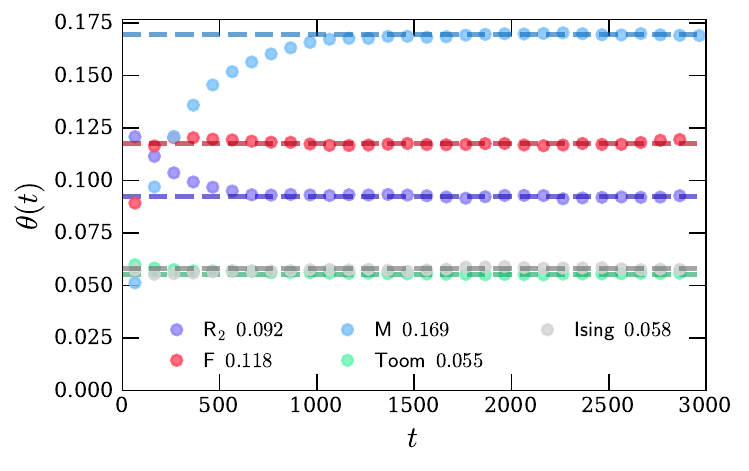}
			\caption{\label{fig:thetat} The time-dependent magnetization exponent $\t(t)$ (c.f. \eqref{thetadef}) for the quenches shown in Fig.~\ref{fig:mt_quench} up to time $t=3000$ (for $\sfM$, $\t(t)$ remains close to $0.17$ out to at least $t = 1.5 \times 10^4$). The horizontal dashed lines are averages of the data over the last $50\%$ of the displayed time interval, and the values of $\t$ at which they are drawn are indicated in the legend. }
		\end{figure}                
		
		As an alternate check on these values of $z$,  we simulate the same quench from $p=0$ to $p=p_c$, and fit the falloff in the time-dependent magnetization to the universal scaling form \cite{janssen1989new,hinrichsen2000non}
	\be\label{mquench_scaling}  m(t)  \sim t^{-\t}, \qq \t = \frac\b{\nu z},\ee
	which is valid for times $1 \ll t \lesssim L^z$.\footnote{As a sanity check, we find that the collapse in Fig.~\ref{fig:static_stuff} gives within error bars the same estimates of $p_c$ as does a search for the noise strength at which $m(t)$ is a clean power law. }
	
	The result of performing these quenches is shown in Fig.~\ref{fig:mt_quench} for systems of size $L=300$, where the squeezing codes are compared to Toom's rule and critical Ising-model Glauber dynamics. The power law behavior takes a while to set in for the squeezing codes (especially for $\sfM$), although we have checked that the onset time appears to be independent of $L$. The dynamics is run out to a maximum time of $t_{\sf max} = 10L$ for all squeezing codes except $\sfM$, for which we take $t_{\sf max} = 50L$ on account of the longer onset time.\footnote{With the exception of $\sfM$, $t_{\sf max}$ is comfortably smaller 
	than our estimates of $L^z$, and so the entire time window should be within the regime where \eqref{mquench_scaling} is applicable. For $\sfM$ this separation does not hold, but the extracted value of $\t$ gives a good fit to $\lan m(t)\ran$ for the entire time range past the initial onset. }
	
	The quality of the power-law scaling can be seen more transparently by plotting the time-dependent exponent 
	\be \label{thetadef} \t(t) = \log_{10}\(\frac{\lan m(t/10)\ran}{\lan m(t)\ran}\).\ee 
	This is done in Fig.~\ref{fig:thetat}, which demonstrates a good quality power-law dependence that sets in rapidly for Toom's rule and Glauber dynamics, less rapidly for $\sfR_{2}$ and $\sfF$, and rather slowly for $\sfM$. The quality of the fits observed here indicates that the values of $p_c$ identified from the scaling collapses of Fig.~\ref{fig:static_stuff} at smaller system sizes ($L\leq 196$) are indeed quite close to the true critical point. Additionally, the near-identical slopes for Toom's rule and Glauber dynamics show that the former is in the same {\it dynamic} universality class as model-A dynamics. 
	
	From these values of $\t$, the relation $\t = \b / \nu z$, and the values of $\b,\nu$ provided in Fig.~\ref{fig:static_stuff}, we extract the values 
	\be \label{zvalues_quench} z = \begin{dcases} 1.88(6) & \sfR_2 \\ 
	1.60(3) & \sfF\\ 
1.35(6) & \sfM \end{dcases}. \ee 
	The values in \eqref{zvalues_quench} are close to those obtained from the collapse of $B(t)$,  satisfy the ordering in \eqref{zordering}, and are all significantly less than 2. They are, however, all smaller than the values obtained from collapsing $B(t)$, and for $\sfR_2$ and $\sfF$ this difference is larger than one standard deviation. A more detailed analysis will need to be done to determine to what extent this discrepancy is due to the data in Fig.~\ref{fig:mt_quench} being taken at a larger system size than the data in Fig.~\ref{fig:static_stuff}, as compared to the quenches being done not exactly at $p=p_c$.

%	For Glauber dynamics and Toom's rule, we obtain $z_A \approx 2.16$ and $z_{\sf Toom} \approx 2.23$, close to the known value of $2.167$ for model-A dynamics. We thus conclude that the phase transitions in Toom's rule and equilibrium Glauber dynamics are very likely in the same dynamic universality class. 
%	
%	For the squeezing codes, we obtain $z_{\sfR_2} \approx 2.10, z_{\sfR_3}\approx 1.93, z_\sfF \approx 1.44$, and $z_{\sfM} \approx 1.38$.  Except for $\sfR_{2,3}$, these estimates are all appreciably smaller than $z_A$. If we instead use $\nu_A,\b_A$ to extract $z$ from $\t$, we get estimates of $z$ which are smaller still, with $z=1.33$ for $\sfR_2$, and an unphysically small $z = 0.75$ for $\sfM$. Compared to $z_A$, we have also observed that these values of $z$ give significantly better fits for the magnetization autocorrelation time at the critical point. 

	\ss{Intrinsically non-equilibrium critical points}

	From the available data, we conclude that even if the static exponents $\b,\nu$ do eventually asymptote to model-A values at very large system sizes (which, according to the analysis of App.~\ref{app:numerics}, seems unlikely), the transitions in the squeezing codes nevertheless lie in a distinct {\it dynamic} universality class, separate from critical model-A dynamics. 
		
	What is the nature of these new universality classes? We will not answer this question in this work, although there is one immediate observation we can make from the estimated values of $z$. Ref.~\cite{masaoka2024rigorous} proved that any critical Markov process which is local and obeys detailed balance must have $z\geq2$. Our estimates of $z$ for all three squeezing codes are  less than 2, and are significantly smaller than 2 for $\sfF$ and $\sfM$. This means that not only are the critical points in these squeezing codes distinct from critical model-A dynamics, but also they are in fact {\it intrinsically non-equilibrium}, viz. do not admit an effective equilibrium description at long distances. This is in contrast to Toom's rule, which as we have just seen, has $z_{\sf Toom} \approx z_A$. The other continuous phase transitions in local stochastic models with $z < 2$ the authors are familiar with are KPZ and percolation, neither of which has similar exponents. Understanding what is happening in the present case of the squeezing codes thus constitutes an outstanding challenge for future work.

	\section{Synchronicity-protected memories} \label{sec:synch_protection} 
	
	In this section, we discuss in more detail the phenomenon of synchronicity protection observed in the $\sfT$ automaton, which acts as a robust memory only under synchronous updates. 
	
	\begin{figure}
		\includegraphics[width=.48\tw]{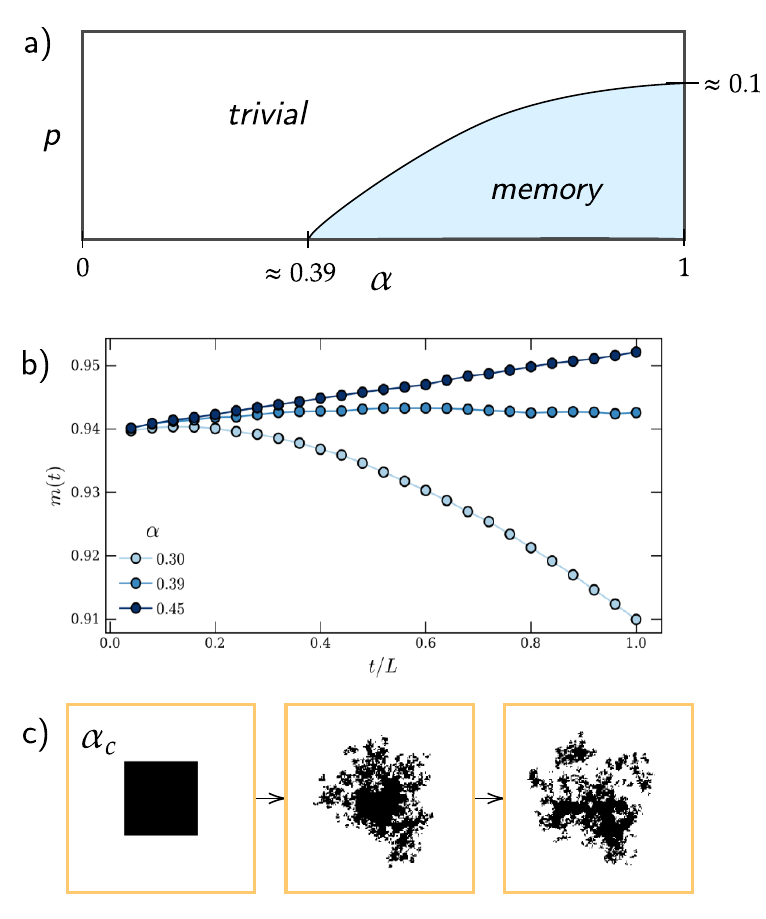} 
		\caption{\label{fig:synch_fig} The synchronicity transition in the $\sfT$ automaton. {\sf a)} A schematic phase diagram in the $(p,\a)$ plane. {\sf b)} The average magnetization $m(t)$ following a quench under noiseless dynamics in a state with a small circular minority domain; the critical point occurs at $\a_c \approx 0.385$. {\sf c)} Snapshots of the system undergoing such an evolution at $\a = \a_c$ (time moves left to right). The blob moves to the northwest. }
	\end{figure}
	
	We will say that an automaton rule $\mca$ is a {\it synchronicity-protected memory} if $\mca$ is a robust memory under synchronous updates, but not under asynchronous ones. In Sec.~\ref{ss:erosion}, we saw that $\sfT$ is synchronicity-protected. Additionally, in the discussion of Sec.~\ref{sec:meanfield}, we mentioned that {\it all} of the squeezing codes studied here are synchronicity-protected when defined in an odd number of spatial dimensions. As a further example, adding $s_\bfr$ to both $\vee$ and $\wedge$ updates of $\sfM$ turns out to produce a synchronicity-protected memory. 
	
	Going beyond this, we have empirically observed that if $\mca$ is a linearly-eroding squeezing code defined by regions $\mcr^\vee$ and $\mcr^\wedge$ which are not related by any symmetry, then $\mca$ is highly likely to be synchronicity protected (a randomly chosen example has $\mcr^\wedge = \{-\uvx,-\uvy,3\uvx+\uvy\}, \mcr^\vee = \{\uvx, 2\uvy, -\uvx+\uvy\}$). This, however, is not always the case: the automaton obtained by combining the $\sfR_2$ and $\sfR_3$ automata as $\mcr^\wedge = \{\uvy,\bfzero,-\uvy\}, \mcr^\vee = \{\uvx,-\uvx\}$ is not synchronicity-protected. 
	
	For a synchronicity-protected memory, it is possible to quantify the degree of synchronicity needed to sustain logical information. Following similar investigations in the cellular automaton literature \cite{fates2007asynchronism,pajouheshgar2025exploring}, we define updates with an intermediate degree of asynchronism $\a$ by taking, at each time step, a given site to update with probability $\a \in [0,1]$ (with the length of each step being $dt = 1/\a$). $\a=1$ thus corresponds to fully synchronous updates, while $\a=0$ corresponds to fully asynchronous updates. All of the synchronicity-protected memories we have investigated turn out to undergo a percolation-type transition between a robust memory and trivial disordered dynamics at an intermediate critical synchronicity parameter $0<\a_c<1$. 
	
	The synchronicity transition is illustrated in Fig.~\ref{fig:synch_fig} for the $\sfT$ automaton, for which $\a_c \approx 0.385$ at $p=0$ ($\a_c$ increases with $p$, leading to the phase diagram shown in the top panel of the figure). 	
	One way of probing this transition is to look at the time evolution of the magnetization $m(t)$ following a quench from an initial state with a small minority domain. When $\a>\a_c$, the dynamics is error-correcting, and the magnetization (ballistically) returns to its logical-state value as the minority domain is eroded. When $\a < \a_c$, by contrast, the minority domain turns into a blob of disordered spins which ballistically expands, causing the magnetization to drop towards $0$. At $\a=\a_c$, the size of the minority domain is unchanged in time, with the domain forming a disordered blob of unchanging size. This is illustrated in panels b and c of Fig.~\ref{fig:synch_fig}.

	In all examples considered so far, synchronicity is beneficial for maintaining a memory. However, we note that it is also possible for synchronicity to have deleterious effects. Examples of such {\it asynchronicity}-protected memories were constructed by making the $s_\pm$ states unstable in \cite{pajouheshgar2025exploring}, and the squeezing codes $\sfR,\sfF,\sfM$ can likely be made asynchronicity-protected in the same way. 
	
	\section{Outlook}\label{sec:outlook}
	
	This work has explored new ways of stabilizing robust many-body memories. We found examples of memories stabilized by fluctuations and by synchronicity, as well as new intrinsically non-equilibrium phase transitions. We end by discussing a few open questions raised by our work. 
	
	\ss{Mechanisms for robust memories } 
	
	One important goal is to elucidate the general mechanisms by which robust memories can be constructed.
	Staying within the context of squeezing codes, it would be valuable as a first step to rigorously identify the necessary and sufficient conditions that $\mcr^\wedge,\mcr^\vee$ must satisfy (c.f. \eqref{general_sqz}) in order for a squeezing code to constitute a robust memory under synchronous updates; one strategy for approaching this problem may be to make use of Toom's zero-set technology (see App.~\ref{app:zerosets}). Making fully rigorous statements for asynchronous updates seems to be more difficult, but it would nevertheless be valuable to obtain at least a qualitative understanding of the scenarios in which synchronicity is important.
	%	 One could also consider defining squeezing codes on different types of lattices, where different types of internal-spatial symmetry intertwinment are possible.  
	
	All of the robust memories in $d=2$ discovered so far protect only a single bit of information (with the exception of a 2D version of Gacs' infamous 1D automaton \cite{gacs1989self}). What distinct mechanisms exist for stabilizing a larger number of bits, and what is the physics of such memories? A first step towards addressing this question would be to look at $\zn$ generalizations of squeezing codes (of which there are several), whose coarsening dynamics and critical phenomena should also be interesting to understand. One may also consider extensions to dynamics with continuous symmetries, which may have the potential of storing a thermodynamically large amount of logical information, given that the memory states may be defined with an angular resolution that scales with system size.

	\ss{Universality classes of noise-robust dynamics}
	
	From a condensed matter perspective, it is natural to ask whether or not the various squeezing codes defined in this work belong to the same phase of matter. Answering this question requires defining exactly what one should mean by a ``phase'' in the context of non-equilibrium dynamics, and since doing so is in fact rather nontrivial, we will not elaborate on the details here (see e.g.  \cite{rakovszky2024defining}). We instead content ourselves with a few general remarks. 
	
	One notion of phase equivalence comes from the renormalization group, with two models of dynamics being in the same phase if they flow to the same fixed point under RG. In this sense, the different squeezing codes studied here should constitute distinct phases of matter, since they perform error correction in qualitatively different ways and with different symmetries, and a good RG scheme should preserve information about how error correction is performed along its flow. However, a satisfactory RG for noisy cellular automata has not yet been fully worked out, and developing one would be an important next step towards understanding the different universality classes of robust memories.
	
	Another notion of phase equivalence can be formulated in terms of ``smooth'' deformations to the dynamics. To this end, we may say that two probabilistic cellular automata $\mca_a,\mca_b$ are in the same ``memory phase'' if there exists a one-parameter family $\mca(\l)\, \l \in [0,1]$ of probabilistic cellular automata such that 
	\begin{enumerate} 
		\item $\mca(0) = \mca_a$ and $\mca(1) = \mca_b$;
		\item $\mca(\l)$ is a robust memory preserving the same logical bit(s) of information for all $\l \in [0,1]$; and 
		\item the logical states of $\mca_a$ are mapped bijectively to those of $\mca_b$ under a path where $\l$ is slowly tuned from $0$ to $1$.
	\end{enumerate}
	$\mca_a$ and $\mca_b$ are thus in the same memory phase if $\mca_a$ can be smoothly turned into $\mca_b$ while faithfully retaining encoded information during this process. The extent to which relationships hold between this notion of phase equivalence and a notion based on RG is not something we will concern ourselves with here. 
	
	The simplest one-parameter family of rules which interpolates between $\mca_a$ and $\mca_b$ is the linear interpolation $\mca(\l) = (1-\l)\mca_a+ \l \mca_b$, here meaning that the rule $\mca_a$ is applied with probability $1-\l$ and $\mca_b$ with probability $\l \in [0,1]$. 
    In general, the dynamics $\mca(\l)$ constructed in this way will {\it not} be a robust memory for intermediate values of $\l$. 	
	Nevertheless, in App.~\ref{app:samephase}, we give an explicit construction of a different one-parameter family of rules which smoothly interpolates between {\it any} two rules $\mca_a,\mca_b$ (possibly with different synchronicity parameters), provided that both rules stabilize the same number of logical bits, and that all logical states possess local order parameters. 
	This means that all of squeezing codes considered in this work, as well as Toom's rule and the models of Ref.~\cite{pajouheshgar2025exploring}, are in the same ``memory phase''. In this interpretation, the results of Sec.~\ref{sec:crit} mean that this single memory phase thus likely has multiple distinct universality classes of continuous phase transitions on the boundary between it and the trivial disordered phase.

	\ss{New non-equilibrium critical points}
	
	When it comes to phase transitions out of memory phases, the most immediate open question concerns the nature of the universality class(es) at the zero-bias phase transitions of the $\sfR,\sfF,\sfM$ squeezing codes. All three squeezing codes appear to be in distinct dynamic universality classes, none of which agree with the model-A expectation. It is also perhaps surprising that $\sfF$ and $\sfM$ appear to be in distinct universality classes, even though they possess essentially the same symmetry (viz. spin-flip combined with a spatial reflection). Most importantly, $\sfF,\sfM$ appear to be intrinsically non-equilibrium, not admitting any effective description in terms of an equilibrium free energy (unlike Toom's rule), and do not appear to be related to KPZ or percolation. The theoretical characterization of these critical points thus presents an outstanding challenge, which will be important to address in future work. 
	
	One may try to perturb model-A dynamics using the non-equilibrium terms in the Langevin equations of \eqref{langevins}, but unless the conclusions of the 1-loop $\ep$ expansion calculations of \cite{bassler1994critical,tauber2002effects} fail at higher orders, model-A dynamics will remain stable to all such perturbations. One possibility may be to instead use a non-equilibrium extension of conformal perturbation theory directly about the 2D Ising CFT (although we have not checked if the static correlations at these critical points obey conformal invariance).
	To shed more light onto the phase transition with numerics, one may consider continuously interpolating between a model with a model-A critical point (such as Toom's rule or thermal Ising dynamics) and one of the squeezing codes, and seeing how the universality class of the phase transition changes along the interpolation. 
	
	This work has focused on squeezing codes with a discrete $\zt$ symmetry intertwining spin flips with a spatial symmetry.
	Squeezing codes can also produce continuous phase transitions in models without any symmetry. An example is the automaton with $\mcr^\wedge = \{\uvy,\bfzero,-\uvy\}, \mcr^\vee = \{\uvx,-\uvx\}$ mentioned in Sec.~\ref{sec:synch_protection}, which numerics suggest has a continuous phase transition at a particular nonzero value of the bias $\eta$. Many more examples can be produced using the neural cellular automaton framework of \cite{pajouheshgar2025exploring}, whose critical points have not yet been systematically studied. 
	
	The numerics in Sec.~\ref{sec:crit} were presented exclusively for asynchronous updates, and we observe significantly different critical exponents for synchronous ones. The discrepancy between synchronous and asynchronous exponents is known to also occur in Toom's rule \cite{makowiec2000study,takeuchi2006can,ray2024protecting}, but the reasons for this difference, as well as the nature of the universality class to which the synchronous critical point belongs, remain unexplained. It is currently unclear to us what types of theoretical tools can be used to study the critical points of synchronous automata, as traditional methods like the MSR formalism require that the generator of the dynamics be expressible as a sum of local terms. 
	
	A final class of critical points yet to be explored are the transitions driven by synchronicity, such as those occurring in the $(p,\a)$ phase diagram of Fig.~\ref{fig:synch_fig}. The synchronicity-driven transitions appear to be related to percolation, but it is again unclear how to best study such transitions analytically.

    \ms 
	
	\section{Acknowledgments: } 
	
	We thank Ehud Altman, Shankar Balasubramanian, Margarita Davydova, Sarang Gopalakrishnan, David Huse, Vedika Khemani, Jong-Yeon Lee, Yaodong Li, Nick O'Dea, Tibor Rakovsky, Shengqi Sang, Grace Sommers, and Charles Stahl for discussions. We are especially grateful to Ruihua Fan for helpful discussions and collaboration on a related RG calculation. E.L. thanks Aditya Bhardwaj, Ehsan Pajouheshgar, and Nathaniel Selub for collaboration on related work, and Ehsan Pajouheshgar for help with the interactive visualization. E.L. was supported by a Miller research fellowship. All code supporting the numerical results in this paper is available at \cite{code}. 
	
	\appendix
	% Inside the appendix, suppress writing of any sub-section-level
	% (\subsection / \subsubsection / \paragraph / starred forms) entries to
	% the TOC, so the main TOC at the start of the paper lists only the
	% appendix titles. Body entries above are unaffected because this
	% redefinition is only in effect from here on.
	\makeatletter
	\let\@orig@addcontentsline\addcontentsline
	\renewcommand{\addcontentsline}[3]{%
		\edef\@app@file{#1}\edef\@app@toc{toc}%
		\edef\@app@type{#2}\edef\@app@keep{section}%
		\ifx\@app@file\@app@toc
			\ifx\@app@type\@app@keep
				\@orig@addcontentsline{#1}{#2}{#3}%
			\fi
		\else
			\@orig@addcontentsline{#1}{#2}{#3}%
		\fi}
	\makeatother

	\begin{widetext}

		\ss*{{\bf Appendix contents}}

		\begin{itemize}
			\setlength\itemsep{3pt}
			\item Sec.~\ref{app:numerics}, \emph{Details on the numerics} --- the full pipeline used to extract the critical exponents $\t = (\nu,\b,z,p_c)$ and their statistical and systematic uncertainties, comprising:
			\begin{itemize}
				\setlength\itemsep{1pt}
				\item Sec.~\ref{ss:collapse}, the determination of static scaling exponents: jackknife per-point errors, the Houdayer-Hartmann collapse cost, a joint Nelder-Mead fit across observables, a parametric bootstrap for $\s_{\rm stat}$, a leave-one-$L$-out estimate of $\s_{\rm sys}$, and a restricted-range $L_{\rm min}$ sweep probing residual confluent corrections;
				\item Sec.~\ref{ss:quench}, refined quench-based determinations of $z$: describes how collapses of the time-dependent Binder cumulant and power-law fits of the magnetization following quenches can be used to estimate $z$;
				\item Sec.~\ref{ss:jointchi}, joint collapses with $\chi$: provides scaling collapses that include the magnetic susceptibility;
				\item Sec.~\ref{ss:modelA_collapse}, collapses with model-A exponents: compares against collapses with model-A values.
			\end{itemize}
			\item Sec.~\ref{app:proof}, \emph{Proofs} --- proves the robustness statements in Sec.~\ref{ss:synchproof} using Toom's monotone eroder theorem. Verifies the eroder condition (and its dual) for each squeezing rule, gives the non-robustness obstruction when $\mcr^\vee = R_\pi(\mcr^\wedge)$, and rederives stability through the zero-set criterion.
			\item Sec.~\ref{app:mf}, \emph{Doi-Peliti and cluster mean-field} --- supplies the details of the cluster mean-field analysis of Sec.~\ref{sec:meanfield}. Develops a Doi-Peliti operator formalism for arbitrary noisy CA, applies it as a warmup to Toom's rule, and then derives the cluster MF equations for the $\sfR_2$ and $\sfR_3$ squeezing codes.
			\item Sec.~\ref{app:magnus}, \emph{Langevin equation from the Magnus expansion} --- derives the effective Langevin equation~\eqref{langevins} for the coarsening dynamics of the $\sfR$ rule in the Floquet Glauber setting of Sec.~\ref{ss:glauber}, via a Floquet-Magnus expansions.
			\item Sec.~\ref{app:samephase}, \emph{Connectivity of locally testable classical memories} --- proves that any two  classical memories with local order parameters that store the same number of logical bits are continuously connectible.
		\end{itemize}

		\section{Details on the numerics}\label{app:numerics}
		
		In this appendix, we provide details on the methodology behind the numerical results discussed in Sec.~\ref{sec:crit}. The code producing these results is available in full at \cite{code}.

		\ss{Finite-size scaling collapse and exponent uncertainties in equilibrium} \label{ss:collapse}

		The quantities we are interested in extracting from our Monte Carlo numerics are the three independent critical exponents $\nu,\b,z$ and the value $p_c$ of the critical noise strength, which we will notationally combine into a single vector 
		\be \t =(\nu,\b,z,p_c).\ee 
		In this section, we describe how we use Monte Carlo to estimate both the value of $\t$ and the uncertainty thereof. We will estimate all components of $\t$ using results {\it equilibrated} to the best of our numerical capability. An additional more powerful method for estimating $z$ from a dynamic quench---which relies on our ability to first accurately identify $p_c$---will be described in a subsequent subsection. 
		
		Our analysis combines three ingredients: 
		\begin{enumerate} 
			\item per-point
			errors on each measured observable, estimated by block jackknife (see below); 
			\item a Houdayer--Hartmann \cite{houdayer2004low} collapse cost that measures how well a trial set
			of exponents brings the measured curves onto a single master curve; and 
			\item a single joint fit of all observables to that cost,
			producing a central estimate for $\theta$. 
		\end{enumerate} 
		We then attach two kinds of uncertainty to $\theta$: statistical component $\sigma_{\rm stat}$ from a parametric
		bootstrap that propagates the per-point jackknife errors into the
		fit, and a systematic component $\sigma_{\rm sys}$ from a
		leave-one-$L$-out diagnostic that probes the sensitivity of the fit
		to confluent corrections beyond the leading scaling ansatz (both explained in detail below). The total
		quoted error in our reporting of the components of $\t$ is then 
		\be \label{sigmatot} \sigma_{\rm tot}=\sqrt{\sigma_{\rm stat}^2+\sigma_{\rm sys}^2}.\ee 		
		The results of running this procedure for the magnetization and Binder cumulant (which determines $p_c,\nu,$ and $\b$) were shown in the scaling collapses of Sec.~\ref{sec:crit} of the main text. 
		
		In the following, we 
		describe in more detail the procedure used to determine the values of the optimal exponents, as well as the errors in each term of \eqref{sigmatot}. 
		
		\sss{Ingredient 1: jackknife errors on measured observables}
		\label{sec:jackknife}
		
		For each system size $L$ and noise value $p$ we collect a time series
		of $n_{\rm data}$ samples of the magnetization on an
		equilibrated trajectory. From this series we measure 
		\be m = \left \langle \left | \frac1{L^2} \sum_i s_i \right|\right\ran,\ee
		the Binder cumulant
		\be B = \tfrac{1}{2}(3 - \langle m^4\rangle/\langle m^2\rangle^2),\ee 
		and
		the autocorrelation time $\tau_M$ of the magnetization as extracted from an
		exponential fit to the connected autocorrelation function of $M = |\sum_i s_i|$.\footnote{We also measure the susceptibility $\chi$, but do not include it in the fit because its associated critical exponent $\g$ is determined from $\nu$ and $\b$ by hyperscaling. In the results quoted in the main text, we have verified that including $\c$ in the fit gives a value of $\g$ that is consistent within error bars of the value dictated by hyperscaling, and does not significantly change the extracted values of $p_c,\nu,\b$.}
		%Since the samples are correlated and $B$ is nonlinear in the magnetization, 
		
		We estimate the statistical error on each observable by a
		\emph{block jackknife}. The $n_{\rm data}$-long time series is divided
		into $n_b = 32$ contiguous blocks of equal block size
		$b = n_{\rm data}/n_b$. Consecutive samples within a block are
		typically correlated; samples in different blocks are effectively
		independent, provided $b \gg \tau_M$. We then form the $n_b$
		leave-one-block-out estimators
		\begin{equation}
			O_{(k)} \;=\; O\bigl(\{x_t\}_{t \notin \text{block }k}\bigr),
			\qquad k = 1,\dots,n_b,
		\end{equation}
		where $O \in \{ m,B,\tau_M\}$ and the input is the data $\{x _t\}$ of the time series with block $k$
		removed. The jackknife estimate of the standard error on $O$ is
		\begin{equation}
			\sigma_O^{\rm jk}
			\;=\;
			\sqrt{
				\frac{n_b - 1}{n_b}\,
				\sum_{k=1}^{n_b}\bigl(O_{(k)} - \bar{O}\bigr)^2
			},
			\qquad
			\bar{O} \;=\; \frac{1}{n_b}\sum_{k=1}^{n_b}O_{(k)}.
			\label{eq:jk-sigma}
		\end{equation}
		This estimator is asymptotically equivalent to the naive standard
		error for linear functionals and correctly inflates the error for
		nonlinear ones such as $B$. The $(n_b{-}1)/n_b$ prefactor
		reflects the fact that each leave-one-out resample shares $(n_b{-}1)/n_b$
		of the data with the full sample, so the individual $O_{(k)}$ are
		highly correlated and their naive sample variance underestimates the
		uncertainty. We obtain one value of $\sigma_O^{\rm jk}$ for every
		$(L, p, O)$ triplet.
		
		\sss{Ingredient 2: quantifying a scaling collapse}
		\label{sec:hh-cost}
		
		Near the critical point each observable $ O \in \{B,m,\tau_M\}$ is expected to satisfy a
		finite-size scaling form
		\begin{equation}
			O_L(p) \;\approx\; L^{\omega_O}\,
			\Phi_O\!\Bigl(\tfrac{p - p_c}{p_c}\,L^{1/\nu}\Bigr),
			\label{eq:fss-ansatz}
		\end{equation}
		with observable-specific rescaling exponents
		\begin{equation}
			\omega_B = 0,\qquad
			\omega_{m} = -\beta/\nu,\qquad
			\omega_{\tau_M} = +z,
		\end{equation}
		and a single dimensionless scaling function $\Phi_O$ per observable.
		If the true values of $\t$ are known, plotting the
		rescaled points $(X_{L,i}, Y_{L,i}) = (\,(p_{L,i} - p_c)/p_c \cdot
		L^{1/\nu},\; L^{-\omega_O}\,O_L(p_{L,i})\,)$ for all $L$ simultaneously
		collapses every curve onto a single master curve (modulo subleading corrections to scaling).
		
		Given a trial value of $\theta$, we quantify the
		residual deviation from an exact collapse using a variant of the
		Houdayer--Hartmann $\chi^2$~\cite{houdayer2004low}. For a given
		observable $O$, for each measured point $(L, i)$ we linearly
		interpolate the master curve at $X_{L,i}$ using the union of all
		measurements from the \emph{other} system sizes $L' \neq L$. Call the
		interpolated value $\hat Y_{L,i}$ and its propagated interpolation
		error $\hat\sigma_{L,i}$. Points whose $X_{L,i}$ falls outside the
		combined $X$-range of the other sizes are skipped, as the master curve
		cannot be sampled there. The cost is then the weighted mean-square
		deviation of each measured $Y_{L,i}$ from the others' interpolant:
		\begin{equation}
			S_O(\theta) \;=\;
			\frac{1}{|\mcc^O|}
			\sum_{(L,i) \in \mathcal{C}^O}
			\frac{(Y_{L,i}(\theta) - \hat Y_{L,i}(\theta))^2}
			{\sigma_{L,i}^2(\theta) + \hat\sigma_{L,i}^2(\theta)}\,.
			\label{eq:hh-cost}
		\end{equation}
		where $\mathcal{C}^O$ is the overlap set\footnote{Concretely: for each $L$, form the union of rescaled $X$-coordinates from every other size $L' \neq L$, $\mathcal{X}_{\neq L} = \bigcup_{L' \neq L}\{X_{L',i}(\theta)\}$, and let $[X_{\min}^{\neq L}, X_{\max}^{\neq L}]$ be its range. A point $(L,i)$ is included in $\mathcal{C}^O$ iff $X_{L,i}(\theta) \in [X_{\min}^{\neq L}, X_{\max}^{\neq L}]$, so that $\hat Y_{L,i}$ is obtained by interpolation rather than extrapolation. Note that $\mathcal{C}^O$ depends on $\theta$; varying $p_c$ or $\nu$ shifts and stretches every $X_{L,i}$, so the set is recomputed at every evaluation of the cost. In practice the points most often excluded are those near the edges of each size's $(p-p_c)$ sweep window, especially for the largest $L$.} and $\sigma_{L,i}$ is the jackknife error of
		the scaled observable, propagated analytically through the multiplicative
		$L^{\omega_O}$ factor (here $Y_{L,i}(\t)$ is the ``data'' and $\hat Y_{L,i}(\t)$ is the ``model''). Equation~\eqref{eq:hh-cost} is essentially a
		reduced $\chi^2$: if the scaling ansatz \eqref{eq:fss-ansatz} is
		exactly correct and the jackknife errors are honest (so that $Y_{L,i}(\theta) - \hat Y_{L,i}(\theta)$ is Gaussian random with variance half of the denominator), $S_O$ should
		fluctuate about unity.
		
		\sss{Ingredient 3: Joint fit across observables}
		\label{sec:joint-fit}
		
		A single observable's cost $S_O$ uses only two or three of the
		exponents (e.g.\ $(p_c,\nu)$ for $B$, $(p_c,\nu,\beta)$ for
		$m$), so minimising it independently for each
		observable generically produces different estimates of $(p_c,\nu)$. Because all four observables describe the same
		underlying critical point, these disagreements are unphysical: any
		remaining discrepancy reflects either finite-size corrections beyond
		the leading scaling form \eqref{eq:fss-ansatz} or statistical noise.
		We eliminate this ambiguity by fitting all four observables
		simultaneously with a single shared $(p_c,\nu)$,
		\begin{equation}
			\theta_* \;=\;
			\arg\min_{\theta} S_{\rm tot}(\theta),
			\qquad
			S_{\rm tot}(\theta)
			\;=\;
			\sum_{O \in \{B,\,m,\,\tau_M\}}
			S_O(\theta),
			\label{eq:joint-cost}
		\end{equation}
		
		The scaling exponents $\beta,z$ only enter their  own summand in
		\eqref{eq:joint-cost}; $(p_c,\nu)$ are coupled across every term and
		must therefore simultaneously collapse all four scaling functions. 
		We minimise \eqref{eq:joint-cost} with SciPy's Nelder-Mead simplex
		algorithm, with the values of model-A dynamics used as a starting point. The result $\theta_*$ is the
		joint-fit central estimate we report in the main text.
		
		\sss{Statistical uncertainty on $\t_*$: parametric bootstrap}
		\label{sec:bootstrap}
		
		The first piece of the error budget on $\theta_*$ comes from the
		finite Monte Carlo statistics. We propagate the per-point jackknife
		errors of Section~\ref{sec:jackknife} through the joint fit using a
		\emph{parametric bootstrap}. In this scheme we assume that each
		measured mean $\bar O_{L,i}$ is a realization of a Gaussian random
		variable centred on the (unknown) true value with standard deviation
		$\sigma_{L,i}^{{\rm jk},O}$, and we generate synthetic data replicas
		by drawing
		\begin{equation}
			O_{L,i}^{(k)} \;\sim\; \mathcal{N}\!\bigl(\bar O_{L,i},\,
			\sigma_{L,i}^{{\rm jk},O}\bigr),
			\qquad k = 1,\dots,K,
			\label{eq:param-boot-draw}
		\end{equation}
		independently at every $(L,i)$ and for every observable, with $K = 500$
		replicas. For each replica we re-run the same minimisation of
		\eqref{eq:joint-cost} (started at $\theta_*$), producing a
		set of fitted parameter vectors $\{\theta_*^{(k)}\}_{k=1}^K$. The
		empirical distribution of $\theta_*^{(k)}$ approximates the
		sampling distribution of $\theta_*$ under the measurement noise.
		%Crucially, no new Monte Carlo trajectories are required, so the cost
		%of producing this part of the error budget scales only with the
		%number of Nelder--Mead refits, not with simulation time.
		
		We summarise the per-component statistical error by the half-width of
		the empirical 68\% percentile interval of the bootstrap distribution.
		Let $\alpha \in \{p_c, \nu, \beta, z\}$ index the components of
		$\theta$, write $(\theta_*^{(k)})_\alpha$ for the $\alpha$-th component
		of the $k$-th bootstrap-refitted parameter vector, and let $q_p(\cdot)$
		denote the empirical $p$-th percentile of its argument over the $K$
		replicas. We then define
		\begin{equation}
			\sigma_{{\rm stat},\alpha}
			\;=\;
			\tfrac12\bigl(q_{84}\!\bigl((\theta_*^{(k)})_\alpha\bigr) -
			q_{16}\!\bigl((\theta_*^{(k)})_\alpha\bigr)\bigr),
			\label{eq:sigma-stat}
		\end{equation}
		the half-width of the central 68\% interval of the marginal
		bootstrap distribution for that component.\footnote{We use this rather than the sample
			standard deviation because it avoids a
			Gaussian assumption on the bootstrap distribution, which is important
			for $p_c$ (whose posterior is typically skewed by the one-sided
			constraint $p_c > 0$), and for $\tau_M$ whose bootstrap
			distribution can be heavy-tailed when the fit window is narrow. }  This determines the first part of the error in \eqref{sigmatot}. 
		%The
		%two-dimensional bootstrap covariance between any pair of parameters is
		%visualised as a scatter of
		%$\{((\theta_*^{(k)})_\alpha,(\theta_*^{(k)})_\beta)\}_{k=1}^K$ with overlaid
		%$1\sigma$ and $2\sigma$ ellipses obtained from the empirical
		%$2\times 2$ covariance matrix (corner plot, produced by the plotter's
		%\texttt{--contours} flag); the strongest such correlation is typically
		%the one between $p_c$ and $\nu$, which captures the residual trade-off
		%between shifting the critical point and rescaling its approach.
		
		\sss{Systematic uncertainty: leave-one-$L$-out}
		\label{sec:sys-loo}
		
		The parametric bootstrap captures the statistical scatter that the
		finite-size data $\{(L,p,O)\}$ would exhibit \emph{if the leading
			scaling ansatz \eqref{eq:fss-ansatz} were exact}. It says nothing about
		how well that ansatz actually applies at the system sizes accessible
		to us. A residual $L^{-\omega_1}$ correction to scaling will pull the central fit by an amount that does not shrink as we increase the number of Monte Carlo trajectories, and is as such
		genuinely a systematic error.
		
		To estimate it, we use a simple leave-one-$L$-out
		diagnostic. Let $\mathcal{L} = \{L_1,\dots,L_n\}$ be the set of system
		sizes in the data. For each $j$, we drop the entire $L_j$-row from
		the data set, refit the remaining $\mathcal{L}\setminus\{L_j\}$ to
		\eqref{eq:joint-cost} with the same warm start and the same fixed
		parameters, and record the resulting fit
		$\theta_*^{(-j)}$. The spread of $\{\theta_*^{(-j)}\}_{j=1}^n$
		across exclusions tells us how strongly the central estimate depends
		on any one system size, viz. how much the leading-order ansatz
		would have to be amended at small $L$ to bring the points back onto a
		single master curve. We summarise that spread by the standard
		jackknife formula,
		\begin{equation}
			\sigma_{{\rm sys},\alpha}
			\;=\;
			\sqrt{
				\frac{n - 1}{n}\,
				\sum_{j=1}^{n}\bigl((\theta_*^{(-j)})_\alpha
				- (\bar{\theta_*})_\alpha\bigr)^2
			},
			\qquad
			\bar{\theta_*}_\alpha
			\;=\; \frac{1}{n}\sum_{j=1}^{n}(\theta_*^{(-j)})_\alpha,
			\label{eq:sigma-sys}
		\end{equation}
		mirroring the per-point jackknife formula \eqref{eq:jk-sigma} but
		applied at the level of the fit rather than the level of the time
		series (the diagnostic requires $n \geq 3$, since the
		Houdayer--Hartmann collapse needs at least two curves
		$\{L_{j'}\}_{j' \neq j}$ to interpolate between).
		In our data the dominant entry of $\sigma_{\rm sys}$ is typically the
		one obtained when the smallest $L$ is dropped; this is consistent with
		the smallest $L$ having the largest confluent correction. This determines the second component of the total error in \eqref{sigmatot}. 
		
		%As a
		%companion sanity check, we also run a restricted-range FSS sweep:
		%refit on $\{L \geq L_{\rm min}\}$ for each successive $L_{\rm min}$
		%and watch how the fitted exponents move. A monotonic drift signals
		%that confluent corrections are still pulling the fit (so $L_{\rm min}$
		%is too small), while a plateau signals that the asymptotic regime is
		%reached.
		
		%\sss{Total reported uncertainty}
		%\label{sec:sigma-tot}
		%
		%We combine the two pieces in quadrature,
		%\begin{equation}
		%	\sigma_{{\rm tot},\alpha} \;=\;
		%	\sqrt{\sigma_{{\rm stat},\alpha}^2 + \sigma_{{\rm sys},\alpha}^2}\,,
		%	\label{eq:sigma-tot}
		%\end{equation}
		%which is the symmetric error attached to each component of $\theta_*$
		%in Table~\ref{tab:critical-exponents} and on the figures. We treat
		%$\sigma_{\rm stat}$ and $\sigma_{\rm sys}$ as independent because they
		%arise from genuinely independent failure modes: $\sigma_{\rm stat}$
		%shrinks under more Monte Carlo time at fixed lattice geometry,
		%whereas $\sigma_{\rm sys}$ shrinks under enlarging $\mathcal{L}$ at
		%fixed Monte Carlo time. In a converged run the two contributions are
		%of comparable magnitude — if $\sigma_{\rm stat}$ dominates, more MC
		%sweeps are the cheapest improvement; if $\sigma_{\rm sys}$ dominates,
		%adding a larger $L$ is.
		
		\sss{Assessing exponent drift with restricted-range $L_{\rm min}$ sweep}
		\label{sec:lmin-sweep}
		
		The leave-one-$L$-out estimator of Sec.~\ref{sec:sys-loo} folds
		any sensitivity of the fit to single-$L$ removal into a single
		number $\sigma_{\rm sys}$, but the way that sensitivity arises is
		itself informative: a fit dominated by confluent corrections will
		respond \emph{monotonically} to dropping the smallest size, while a
		fit limited only by statistical noise will in general not. 
		To separate these two cases we run a complementary diagnostic which
		we call the \emph{$L_{\rm min}$ sweep}.
		
		Let $L_1 < L_2 < \dots < L_n$ be the sorted system sizes in the data.
		For each $k = 1,\dots,n-1$ we refit \eqref{eq:joint-cost} on the
		restricted set $\{L_k, L_{k+1}, \dots, L_n\}$, again warm-starting at
		$\theta_*$, and record the result $\theta_*^{(\geq k)}$. The
		$k=1$ entry is just the all-$L$ central fit; each subsequent entry
		discards one more small-$L$ row. Because the leading scaling form
		\eqref{eq:fss-ansatz} is only the asymptotic large-$L$ behaviour and
		the true Binder, magnetization, and autocorrelation
		time admit a confluent expansion of the form
		$O_L(p) = L^{\omega_O}\Phi_O\bigl(\tfrac{p-p_c}{p_c} L^{1/\nu}\bigr)
		\bigl[\,1 + a\,L^{-\omega_1} + \cdots\bigr]$ with a leading correction
		exponent $\omega_1 > 0$, the $L^{-\omega_1}$ contamination in the fit
		shrinks as the smallest $L$ in the data set is raised. The trajectory
		$\{\theta_*^{(\geq k)}\}$ as $k$ rises therefore probes whether the
		asymptotic regime is being reached:
		
		\begin{itemize}
			\item \textbf{Plateau.} If successive entries fall inside one
			$\sigma_{\rm stat}$ of each other, with no sign of a systematic
			trend, the fit has converged.
			\item \textbf{Monotone drift.} If the entries move in one
			consistent direction as $k$ rises, finite-size corrections are
			still pulling the fit at every kept $L$ and the asymptotic regime
			has not been reached.
			\item \textbf{Non-monotone scatter.} If the entries jump around
			without a consistent direction, the diagnostic is being limited by
			the statistical noise of the bootstrap-equivalent fit on a smaller
			data set, and the leave-one-$L$-out spread $\sigma_{\rm sys}$ is
			already accounted for elsewhere.
		\end{itemize}
		
		\begin{figure} 	\includegraphics[width=.31\tw]{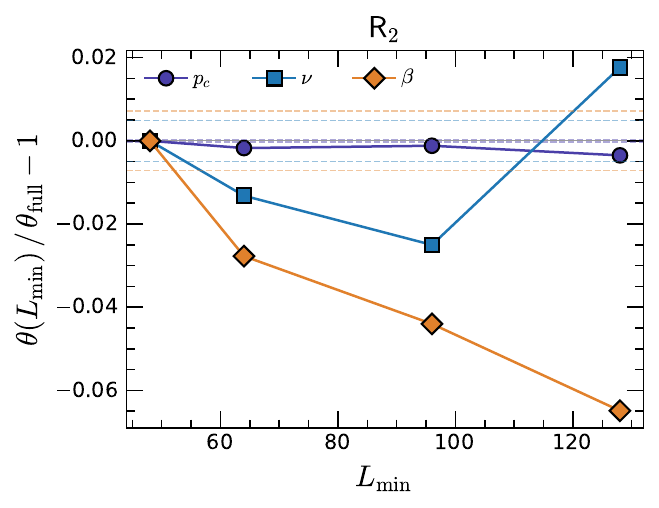} 
			\includegraphics[width=.31\tw]{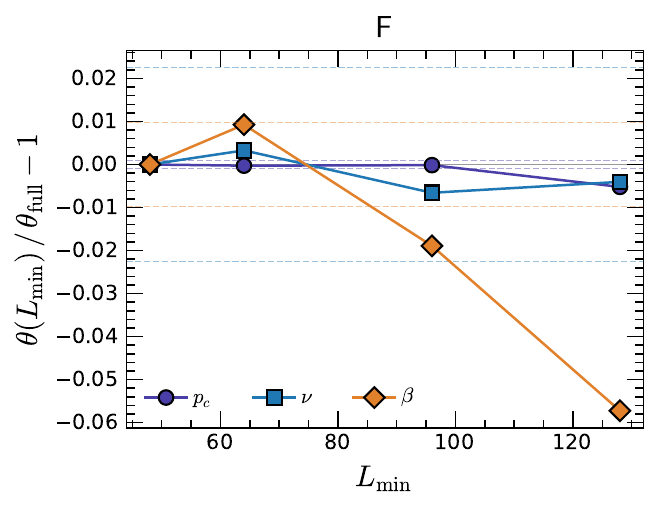} 
			\includegraphics[width=.31\tw]{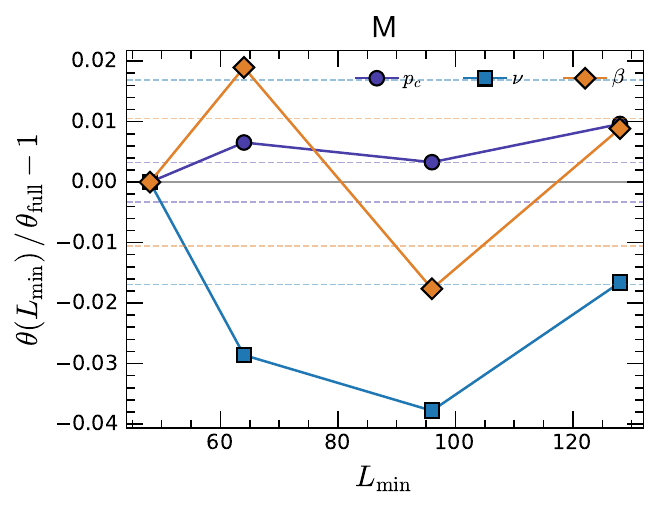} 
			\caption{\label{fig:exponent_drift} The result of doing the $L$-min sweep described in Sec.~\ref{sec:lmin-sweep} for $p_c,\nu$, and $\b$ on the data used to generate the scaling collapses in the main text (Fig.~\ref{fig:static_stuff}). Dashed lines are drawn at the statistical error $\sigma_{\rm stat}$ for each quantity.}
		\end{figure} 
	
		\begin{figure} 	\includegraphics[width=.31\tw]{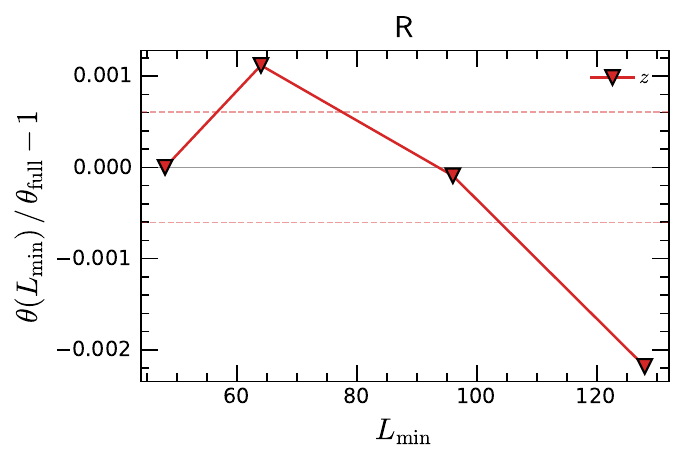} 
		\includegraphics[width=.31\tw]{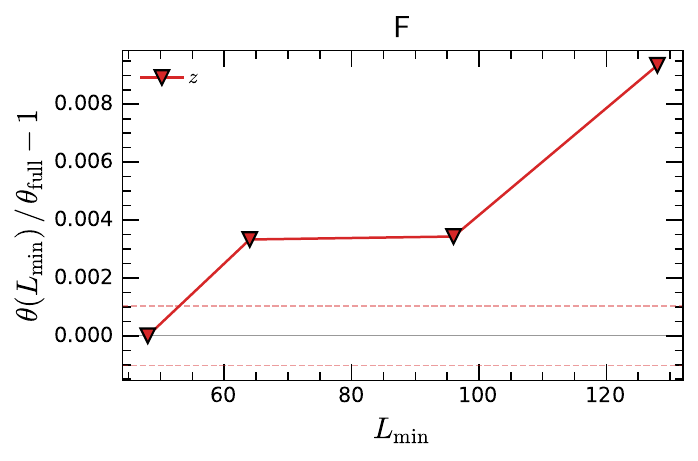} 
		\includegraphics[width=.31\tw]{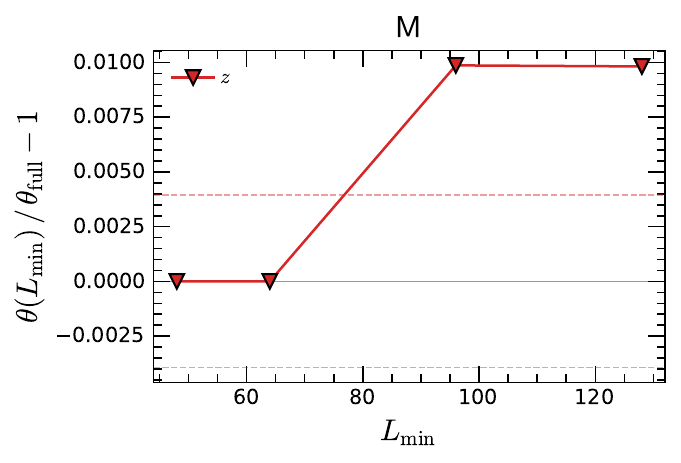} 
		\caption{\label{fig:z_drift} Drifts of the dynamic exponent $z$ as determined by the $L$-min sweep described in Sec.~\ref{sec:lmin-sweep} applied to the data used to generate the scaling collapses of the time-dependent Binder cumulant at criticality (Fig.~\ref{fig:Bt_collapses}). Dashed lines are drawn at the statistical error $\sigma_{\rm stat}$. }
	\end{figure} 
		
		The outcomes of doing this procedure for $p_c,\nu,\b$ for the three squeezing codes studied in the main text are shown in Fig.~\ref{fig:exponent_drift}, indicating that in most cases the exponents appear to have converged over the system sizes shown ($L=32$ to $L=196$); the biggest exception is $\b$ for $\sfR_2$, which may be systematically drifting towards a smaller value. The outcomes of doing this for the procedure of determining $z$ using the collapse of $B(t)$ at the critical point is shown in Fig.~\ref{fig:z_drift}. 
		
		\ss{Estimating $z$} \label{ss:quench}
		
		\sss{Magnetization autocorrelation time}

			\begin{figure}
			\includegraphics[width=.31\tw]{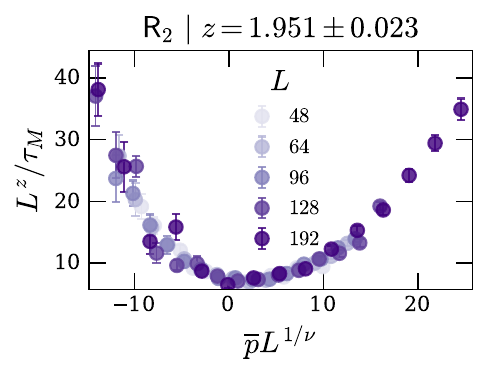}
			\includegraphics[width=.31\tw]{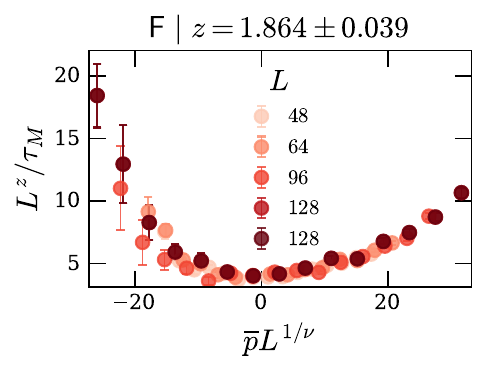}
			\includegraphics[width=.31\tw]{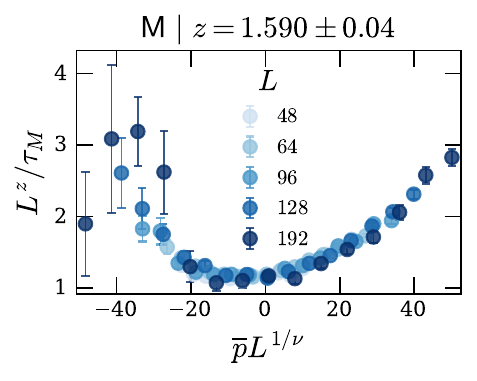}
			\caption{\label{fig:taumcollapse}
				Scaling collapses for the autocorrelation time of the magnetization, performed jointly with the collapses for the magnetization and binder cumulant shown in the main text.
			}
		\end{figure}
		Fig.~\ref{fig:taumcollapse} shows the collapses of $\tau_M$  (see below) that accompany these fits, which determines $z$. In practice, better-quality estimates for $z$ are given by looking at quenches (see the main text and Sec.~\ref{ss:quench} below), but we include this data here for completeness.  
		
		Fitting $\tau_M \sim L^z$ is however often a suboptimal way of determining $z$, since it relies on determining a slope of an exponential decay in the autocorrelation function, and in practice yields significantly noisier estimates for $z$ than the same data does for $\nu$ and $b$. In practice, we thus use two additional ways of estimating of the dynamical exponent $z$ from
		out-of-equilibrium quench simulations. Both run trajectories from
		an ordered initial state at $p = p_c$, sample observables every MC
		sweep, and average over $n_{\rm samples}$ independent trajectories;
		they differ in which observable they target and which scaling form
		they invert.
		
		\sss{Time-dependent Binder cumulant collapse}
		
		The trajectory-averaged Binder cumulant after the quench at the critical point obeys the
		dynamic finite-size scaling form
		\begin{equation}
			B(t, L) = \Psi_B(t/L^z),
			\label{eq:binds-t-collapse}
		\end{equation}
		with $\Psi_B(0) = 1$ and a non-trivial late-time plateau $\Psi_B(\infty)$ whose value is characteristic of the universality
		class. Treating $\{B(t_i, L_j)\}_{i,j}$ as a single
		``observable'' indexed by its time stamp instead of $p$, we minimise
		the same Houdayer-Hartmann collapse cost \eqref{eq:hh-cost} but with
		$z$ as the only free parameter, the rescaled coordinate set to $X =
		t/L^z$, and no multiplicative $L^{\omega}$ prefactor on $Y = B$. The
		per-point statistical errors $\sigma_{B(t_i, L_j)}$ feeding the cost
		come from a jackknife over the $n_{\rm samples}$ independent
		trajectories at each $L$ (each trajectory plays the role of a block
		in \eqref{eq:jk-sigma}). Statistical and systematic errors on $z$
		follow from running the bootstrap and leave-one-$L$-out machinery of
		Sections~\ref{sec:bootstrap}--\ref{sec:sys-loo} on the resulting
		cost surface, identically to the equilibrium fits.
		
		The collapse form \eqref{eq:binds-t-collapse} relies on the quench
		sitting exactly at $p = p_c$ (which we identify from the equilibrium Monte Carlo simulations discussed previously). Away from criticality, scaling function for $B$ also depends on $(p-p_c) L^{1/\nu}$; for large $L$ the curves at different $L$ approach
		distinct, non-universal late-time plateaus (closer to $1$ for $p < p_c$, and closer $0$ for $p > p_c$), and a one-parameter $t/L^z$ collapse onto a
		single master curve breaks down. This means that the
		quality of the collapse itself is a useful self-consistency check on
		the equilibrium-determined $p_c$. In the main text we do not investigate this explicitly, and directly use the value of $p_c$ reported from the equilibrium systems for all of the quenches. 
		
		\sss{Magnetization power-law decay}
		
		At $p = p_c$, the trajectory-averaged magnetization following the
		ordered-start quench decays. In the window
		$1 \ll t \ll L^z$, where $t$ is large enough to be in the asymptotic
		regime but small enough that finite-size effects have not yet
		arrested the decay, the magnetization $m(t)$ is independent of $L$. From the scaling form $m(t,L) = L^{-\b/\nu}\,\Psi_m(t/L^z)$, the $L$-dependence dictates that $\Psi_m(x) \approx x^{-\b/\nu z}$ for small $x$,  meaning that in this window 
		\begin{equation}
			m(t) \;\sim\; t^{-\theta},
			\qquad
			\theta =\b/(\nu z).
			\label{eq:mt-power-law}
		\end{equation}
		We extract $\theta$ from a linear fit of $\log m(t)$
		against $\log t$, restricted to a window in the asymptotic regime,\footnote{We default to the middle 50\% of the trajectory in $\log t$, with the
			endpoints cross-checked by varying them and confirming the fitted
			slope is stable.} and then invert \eqref{eq:mt-power-law} to obtain $z$, with $\b$ and $\nu$ taken from the equilibrium joint fit discussed previously. Statistical and systematic
		errors on $z$ are obtained by propagating the corresponding errors
		on $\theta$, $\b$, and $\nu$ in quadrature, treating the three as independent.
		
		\ss{Joint collapses with $\chi$}\label{ss:jointchi}
		
		In the main text, we performed scaling collapses using only the Binder cumulant and the  magnetization. These collapses give $p_c,\nu$, and $\b$, and hyperscaling relations allow the remaining static exponents to be recovered. As a sanity check, we may also perform a collapse that additionally involves the susceptibility $\chi$ of the magnetization (which we have not done in the main text, since $\chi$ is noisier). The results of doing this are shown in the following table: 		
		\begin{table}[h]
			\centering
			\renewcommand{\arraystretch}{1.3}
			\begin{tabular}{c|ccc}
				& $\sfR_2$ & $\sfF$ & $\sfM$ \\ \hline
				$p_c$       & $0.038273(32)$ & $0.011431(12)$ & $0.00315(4)$ \\
				$\nu$       & $0.950(18)$    & $0.965(22)$    & $1.03(4)$    \\
				$\beta$     & $0.1647(31)$   & $0.1799(32)$   & $0.219(9)$   \\
				$\gamma$    & $1.589(24)$    & $1.575(30)$    & $1.57(6)$    \\ \hline
				$2(\nu-\beta)$ & $1.570(37)$ & $1.571(45)$    & $1.62(8)$    \\
				$\gamma - 2(\nu-\beta)$ & $+0.02(4)$ & $+0.00(5)$ & $-0.05(10)$ \\
			\end{tabular}
			\caption{Critical point $p_c$ and critical exponents extracted from a joint Houdayer-Hartmann collapse of the Binder cumulant $B$, the magnetization $m$, and the susceptibility $\chi$ (the autocorrelation time $\tau_M$ is excluded from this fit). Uncertainties are $\sigma_{\rm tot}=\sqrt{\sigma_{\rm stat}^2+\sigma_{\rm sys}^2}$, with $\sigma_{\rm stat}$ obtained from a $K=100$ parametric bootstrap on the per-point block-jackknife errors and $\sigma_{\rm sys}$ from leave-one-$L$-out refits across the five system sizes $L\in\{48,64,96,128,192\}$. The bottom two rows compare the fitted $\gamma$ with the hyperscaling prediction $\gamma = d\nu - 2\beta = 2(\nu - \beta)$ in $d=2$.}
			\label{tab:four-param-exponents}
		\end{table}
		
		For all three rules the deviation of $p_c,\nu$, and $\b$ from the collapses in the main text are small (nearly within error bars), and extracted values of $\g$ agree with the predictions from hyperscaling with $d=2$ within statistical uncertainty: the deviation $\gamma - 2(\nu-\beta)$ is $+0.02 \pm 0.04$ ($\sfR$, $0.4\sigma$), $+0.00 \pm 0.05$ ($\sfF$, $0.1\sigma$), and $-0.05 \pm 0.10$ ($\sfM$, $0.5\sigma$). The notable feature is that $\gamma$ comes out essentially indistinguishable ($\approx 1.57 \pm 0.03$) across all three rules even though $\nu$ and $\beta$ individually shift between them; the variation in $\beta/\nu$ between rules tracks the variation in $\nu$ such that $\gamma$ stays roughly constant.

		\ss{Collapses with static model-A exponents}\label{ss:modelA_collapse}

			\begin{figure*}
			\centering
			\includegraphics[width=0.32\textwidth]{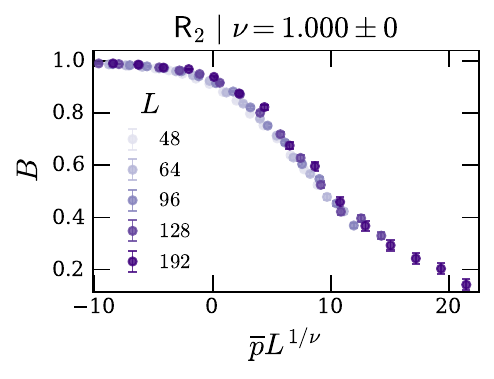} \includegraphics[width=.32\tw]{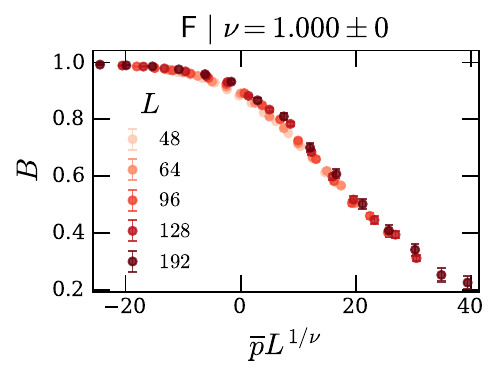}
			\includegraphics[width=.32\tw]{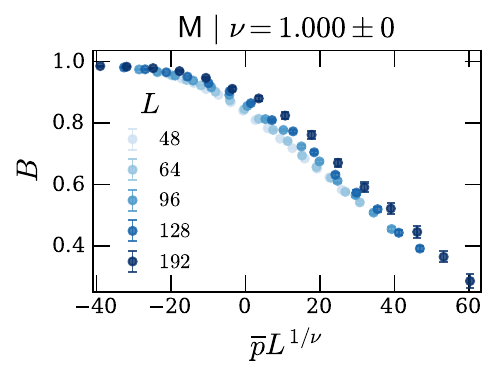} \\ 
				\includegraphics[width=0.32\textwidth]{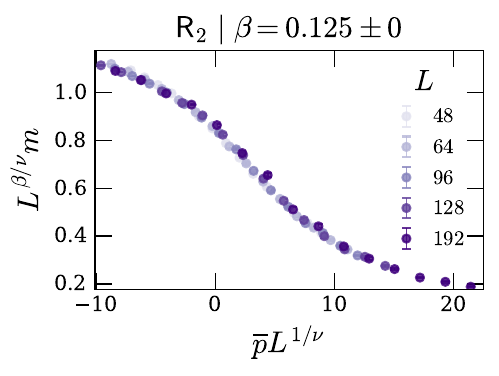} \includegraphics[width=.32\tw]{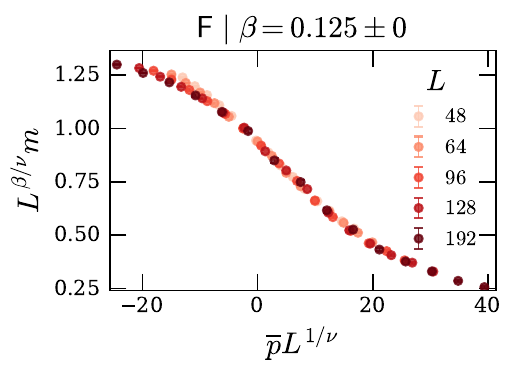}
			\includegraphics[width=.32\tw]{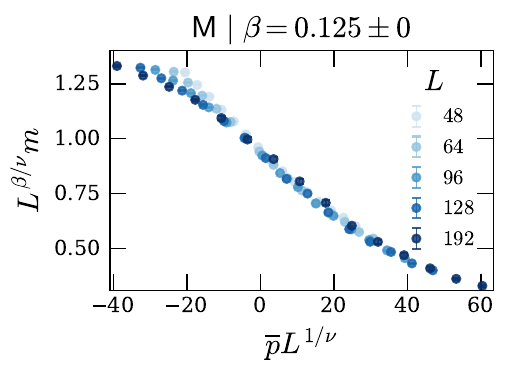}
			\caption{\label{fig:static_stuff_ising} Collapses of the Binder cumulant and magnetization using model-A exponents. The values of $p_c$ identified by this collapse are 
				$p_c = 0.037818 (\sfR)$, $p_c = 0.011213 (\sfF)$, and 
				$p_c = 0.0030319 (\sfM)$. }
		\end{figure*}

		For completeness, in this section we perform the same analysis as in Sec.~\ref{ss:statics}, but with the exponents $\nu,\b$ fixed to their model-A values of $\nu_A= 1, \b_A = 0.125$ (the only unknown we optimze the fit over is  thus the value of $p_c$). The results are shown in Fig.~\ref{fig:static_stuff_ising}. 
		
		Except in the case of $\sfM$ these collapses are perhaps visually not much worse than those in Fig.~\ref{fig:static_stuff}. However, examining the quality of the collapse reveals that the difference is indeed significant. To quantify this, we compare the Houdayer-Hartmann cost $S$ obtained at the two optima. Since $S$ is a reduced-$\chi^2$-like statistic (the average of $(Y-\hat Y)^2/(\sigma^2+\hat\sigma^2)$ over the overlap set $\mathcal{C}^O$), the appropriate test is the $\chi^2$-difference: under the null hypothesis that the model-A values describe the data, the statistic
		\be \chi^2_{\rm diff} = \sum_{O\in\{B,m\}} N^O_{\rm eff} \(S^{\rm model-A}_O - S^{\rm free}_O\) \ee
		is asymptotically distributed as $\chi^2_k$ with $k=2$ degrees of freedom (the number of fixed parameters $\nu,\b$). For each rule, the equivalent ``$n$-$\sigma$'' significance follows from the large-$\chi^2$ approximation $n_\sigma \approx \sqrt{\chi^2_{\rm diff} - k}$. The numbers, with $S_O$ averaged over $N^O_{\rm eff} \approx 70$ points per observable, are:
		\begin{center}
			\renewcommand{\arraystretch}{1.15}
			\begin{tabular}{c|cccc}
				& $S^{\rm free}_B + S^{\rm free}_m$ & $S^{\rm Ising}_B + S^{\rm Ising}_m$ & $\chi^2_{\rm diff}$ & $n_\sigma$ \\ \hline
				$\sfR$ & $2.6$  & $26.2$ & $1650$ & $41$ \\
				$\sfF$ & $2.2$  & $32.8$ & $2170$ & $47$ \\
				$\sfM$ & $4.7$  & $81.6$ & $5540$ & $74$ \\
			\end{tabular}
		\end{center}
		For all three rules the unconstrained per-observable $S$ values sit respectably close to unity (the value expected at the true optimum if the leading scaling ansatz is exact and the jackknife errors are honest), while the model-A-constrained $S$ values are an order of magnitude or more above this. The resulting $\chi^2_{\rm diff}$ exceeds the $k=2$ degrees of freedom by factors of $10^3$--$10^4$, corresponding to formal significances of $40\sigma$ ($\sfR$), $47\sigma$ ($\sfF$), and $74\sigma$ ($\sfM$). Even allowing for substantial underestimation of the per-point errors (e.g., inflating $\sigma_{\rm jk}$ by a factor of two would reduce each $n_\sigma$ by a factor of two), the static critical behaviour of all three rules is statistically incompatible with the model-A universality class at the system sizes studied.

		\section{Proofs}\label{app:proof}
		
		In this appendix we prove the results stated in Sec.~\ref{ss:synchproof}. To do this, we first need to explain Toom's monotone eroder theorem \cite{toom1980stable} (see also \cite{pajouheshgar2025exploring}). 
		
		We define a CA rule $\mca$ to be {\it monotonic} if it implements a monotonic Boolean function, viz. if $x \prec y \implies \mca(x) \prec \mca(y)$, where $x \prec y$ means that $x_\bfr \leq y_\bfr \, \forall \, \bfr$. Additionally, we define the {\it dual} $\mca^\neg$ of $\mca$ as 
		\be  \label{dualdef} \mca^\neg(s) = -\mca(-s).\ee 
		With these definitions, Toom's result \cite{toom1980stable} then states the following: {\it if $\mca$ and $\mca^\neg$ are both monotone eroders, then $\mca$ is a robust memory under synchronous updates}. 
		
		\ss{Proving stability}
		
		We now use this to prove theorem~\ref{thm:robust}. 
		
		\begin{proof} 
			If $\{s_i\}$ are a set of Boolean variables, any function of the form $f(\{s_i\}) = s_{i_1} \odot_1 s_{i_2} \cdots\odot_{n-1} s_{i_n}$ with $\odot_m \in \{\vee,\wedge\}$ is monotonic (non-monotonic functions would involve both variables $s_i$ and their negations $-s_i$). Thus the $\sfR,\sfF,\sfM,\sfT$ automata, as well as their duals, are all monotonic. Furthermore, given the $\wt X$ symmetries they obey, if they are eroders, then so too are their duals. We thus only need to show that each of $\sfR,\sfF,\sfM,\sfT$ are linear eroders under synchronous updates. That this is true is strongly suggested by Fig.~\ref{fig:megafig}, and here we sharpen this into a rigorous argument. 	
			An abstract proof using Toom's zero-sets can be given rather quickly, which we do in a latter subsection. In what follows, we instead give a more elementary proof. 
			
			Consider a state $s_{-,A}$ for which $A$ is a square of linear size $R$ centered at the origin $\bfzero$ (as in Fig.~\ref{fig:megafig}). By monotonicity, showing that each automaton erodes $s_{-,A}$ is sufficient for showing erosion. Indeed, monotonicity means that if $\mca$ erodes all square regions, then it erodes $s_{-,A}$ for any choice of $A$; this follows by enlarging an arbitrary $A$ to the smallest square region in which it is contained. 
			
			Showing erosion of square domains can be done quite directly by examining how large domain walls move under the dynamics. To explain this, we need some notation. For a CA rule $\mca$ and initial state, define the {\it damage} $\dam(t)$ as 
			\be \dam_\mca(t) = \{ \bfr \, : \, [\mca^t( s)]_\bfr = 1 \}.\ee
			In all of what follows, we will take the initial state to be $s_{-,A}$ with $A$ a square of linear size $R$, as above (with the total system size always taken to be much larger than $R$). 
			Additionally, define the half-spaces 
			\bea H^\pm_|(a) & = \{ (x,y) \, : \, x \gtrless a \} \\
			H^\pm_-(a) & = \{ (x,y) \, : \, y \gtrless a \} \\ 
			H^\pm_\diagup(a) & = \{(x,y) \, : \, x-y \gtrless a \} \\ 
			H^\pm_\diagdown (a) & = \{(x,y) \, : \, x+y \gtrless a\}.\eea 
			We will also use the symbol $d^\pm_o(a)$ to denote the spin configuration  $s_{+,H^\pm_o(a)}$, viz. a state on an infinite system with a domain wall along the border defined by $o \in \{|,-,\diagup,\diagdown\}$. 
			Finally, in what follows we will find it convenient to define a single time step as containing both a $\vee$ update {\it and} a $\wedge$ update, with the former coming first (the exact choice of ordering is unimportant). We will adopt this convention in the rest of this proof. 
			
			Consider first $\sfR$. We claim that\footnote{To lighten the notation, we will avoid putting floor functions on fractions like $R/2$ in what follows.}
			\be \label{rdam} \dam_\sfR(t) \subset H^-_-(R/2-t) \cap H^+_-(-R/2+t). \ee 
			This implies $\dam_\sfR(t) = \emptyset \,\, \forall \,\, t > R/2$, thus implying (linear) erosion. Showing \eqref{rdam} is straightforward. Indeed, consider a state $d^+_-(y)$, and consider a site $\bfr=(x,y)$. The $\wedge$ part of the updates in \eqref{sqzrule_r} then guarantees that $\sfR(s)_{(x,y)} = -1$, and this shows that $\dam_\sfR(t) \subset H^-_-(R/2-t)$. An identical argument shows that $\dam_\sfR(t) \subset H^+_-(-R/2+t)$. 
			
			For $\sfF$, we claim 
			\be \dam_\sfF(t) \subset  H^-_-(R/2) \cap H^+_-(-R/2+t), \ee
			so that $\dam_\sfR(t) = \emptyset \,\, \forall \,\, t > R$. This can be shown by considering the action of $\sfF$ on horizontal domain walls. Note that the domain is eroded by time $R$, giving an erosion speed twice as slow as that of $\sfR$; when the updates are made asynchronous, this difference is reflected in \eqref{erosion_speeds}. 
			
			For $\sfM$, we have
			\be \dam_\sfM(t) \subset H^+_\diagdown (-R) \cap H^-_\diagdown(R-2t)\ee 
			giving $\dam_\sfM(t) = \emptyset \,\, \forall \,\, t > R$; this similarly follows by examining the action of $\sfM$ on states $d^\pm_\diagdown(a)$, and also gives an erosion speed twice as slow as $\sfR$. 
			
			$\sfT$ is slightly more complicated: 
			\be \dam_\sfT(t) \subset H^+_\diagup(R-t) \cap H^+_-(-R/2+t) \cap H^-_|(R/2-t),  \ee 
			which follows from analyzing action of $\sfT$ on states $d^\pm_-,d^\pm_|,d^\pm_\diagup$. The intersection of the sets on the RHS is empty when $t > 2R$, giving an erosion speed four times as slow as $\sfR$.
		\end{proof} 
		Toom's rule, as well as the automata rigorously proved to be memories in \cite{pajouheshgar2025exploring}, are defined using majority votes over particular patterns of cells. Like the $\sfT$ rule, these types of dynamics have $X$ symmetry. This means that the domain walls in the states $d^\pm_o$ must move in the same direction under the dynamics. By contrast, the rules $\sfR,\sfF,\sfM$, which lack $X$ symmetry, operate differently: for these rules, there are choices of $o$ for which the domain walls in the states $d^+_o,d^-_o$ move in {\it opposite} directions. This behavior qualitatively distinguishes the dynamics from rules based on majority votes.  
		
		\ss{Proving non-robustness if $\mcr^\vee = R_\pi (\mcr^\wedge)$}
		
		In this subsection, we prove theorem~\ref{thm:nonrobust} from the main text. 
		
		\begin{proof}
			For conciseness, we use schematic language in this proof. The basic idea is that, as explained below, dynamics with this symmetry must erode minority domains in a curvature-driven way; when biased noise is added to drive ballistic expansion of minority domains, the error correcting dynamics is accordingly unable to correct errors quickly enough. 
			
			To illustrate this, consider any domain wall state $d^\pm_o$. Under $\mca$, the domain wall must either move, or remain fixed. Suppose it moves with velocity $\bfv$. Now if $\mcr^\vee = R_\pi(\mcr^\wedge)$, then the dynamics is invariant under $X \circ R_\pi$. $d^\pm_o$ is invariant under this symmetry, but $X \circ R_\pi \, : \, \bfv \mt - \bfv$. Thus the domain wall must also move along $-\bfv$, and this implies that in fact $\bfv = \bfzero$. 
			Therefore all states $d^\pm_o(a)$ must in fact be fixed under $\mca$. 
			
			This argument clearly does not rely on the DW being infinite in extent; any DW extending over more than 3 sites will similarly be invariant. This means that $\mca$ can act nontrivial only at corners between DWs of different orientations, which implies that the error-correcting dynamics of $\mca$ is curvature driven, and standard arguments then show that $\mca$ cannot be a robust memory (since any erosion of minority domains is not fast enough to compete against the ballistic expansion induced by biased noise). 	
		\end{proof}
		
		\ss{Proving stability using zero sets}\label{app:zerosets}
		
		Finally, we provide a quick alternate proof of theorem~\ref{thm:robust} using the ``zero-set'' technology developed by Toom \cite{toom1980stable}. 
		
		\begin{proof}
			For a site $\bfr$, define a {\it zero set} $Z_\bfr$ as a collection of sites for which 
			\be s_{\bfr'}(t) = 0 \, \,\forall \,\,\bfr'\in Z_\bfr \implies s_\bfr(t+1) = 0.\ee 
			Define a {\it minimal zero set} to a zero set of minimal size. Then Toom's zero-set theorem \cite{toom1980stable} says that $\mca$ is an eroder iff the intersection of the convex hulls of the minimal zero sets at any given point is empty: 
			\be \mca\,\, \text{an eroder} \lra \bigcap_{Z_\bfr \in {\sf MinZero}_\bfr } {\sf conv}(Z_\bfr) = \emptyset,\ee 
			where ${\sf MinZero}_\bfr$ is the set of minimal zero sets at $\bfr$, and ${\sf conv}$ denotes the convex hull. 
			
			Consider then calculating the minimal zero sets of the automata under consideration. For convenience, we will group the $\wedge$ and $\vee$ updates into a single larger update (in that order). We will thus use notation where $t\in 2\zz+1$ is in the ``middle'' of a time step. In order for $s_\bfr(t+2) = 0$, we require that $s_{\bfr'}(t+1) = 0 \,\, \forall \,\, \bfr' \in \mcr^\vee_\bfr$, and the ways of ensuring that this happens define the elements of ${\sf MinZero}_\bfr$. 
			
			To complete the proof, we explicitly find 1) two disjoint elements of ${\sf MinZero}_\bfr$ for $\sfR,\sfM,\sfF$, and 2) three elements of ${\sf MinZero}_\bfr$ with trivial intersection for $\sfT$. This is done in Fig.~\ref{fig:zerosets}. 		
		\end{proof}

		\begin{figure*}
			\includegraphics[width=.85\tw]{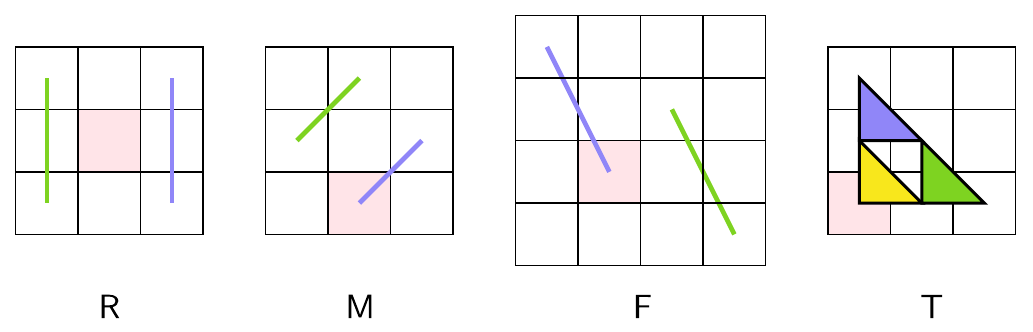} 
			\caption{\label{fig:zerosets} Examples of  minimal zero sets for the rules $\sfR,\sfM,\sfF,\sfT$, used in the proof of theorem~\ref{thm:robust}. The green, purple, and (for $\sfT$) yellow regions are the convex hulls of particular zero sets for the lattice point (square cell) shaded in red. For $\sfR,\sfM,\sfF$ the indicated sets are disjoint, while for $\sfT$ they overlap pairwise, but their intersection is empty.  }
		\end{figure*}

		\section{Doi-Peliti and cluster mean-field}\label{app:mf}
		
		In this appendix we provide the details of the cluster mean-field approach discussed in Sec.~\ref{sec:meanfield}. 
		
		We will begin in Sec.~\ref{ss:operator_formalism} by using the Doi-Peliti technique to derive a framework from which the dynamic cluster MF equations for any noisy cellular automaton may be derived. In Sec.~\ref{app:toom_mf}, we will apply this formalism to Toom's rule, where the analysis is very simple. Then in Secs.~\ref{ss:s2mf} and \ref{ss:r3mf} we derive cluster MF equations for the $\sfR_2$ and $\sfR_3$ squeezing codes, respectively. A similar analysis can be carried out for the $\sfF,\sfM$ squeezing codes, but the details are rather tedious, and will be omitted. 
		
		%	with the stability of the ordered state made possible only by spatial fluctuations in the magnetization. This fact is related to why the order dissappears in larger dimensions, since fluctuations are weaker in higher dimensions. For $\sfR_2$ $p_c$ is (roughly speaking) proportional to the strength of fluctuations, while for $\sfR_3$ a critical fluctuation strength is required to maintain order, even in the absence of other noise.
		%	%	 This is very roughly analagous to an intrinsically non-equilibrium version of the Pomeranchuk effect. 
		
		\ss{Operator formalism of automaton dynamics}	\label{ss:operator_formalism}
		
		In the calculations to follow, we will use the Doi-Peliti operator formalism \cite{doi1976second} to derive dynamic cluster mean-field equations for general noisy two-state asynchronous automata. To lighten the notation somewhat, we will use italic roman letters $i,j,\dots$ for site indices rather than the boldface $\bfr,\bfr',\dots$ of the main text. 
		%	This formalism is slight overkill for the simplest mean field treatments (which as we will see fail for $\sfR_{2,3}$), but significantly streamlines calculations for more complicated treatments, including the cluster mean field approach adopted below. 
		
		The operator formalism is based on writing the evolution of the full probability distribution $P(\{ s_i\})$ as a Markov equation of the form $\k{\p_t P} = +H \k{P}$ for a Fock-space operator $H$ which generates the dynamics. In what follows, we will abuse notation slightly by defining $|+\ran$ as the un-normalized uniform sum over all configurations, so that normalization of $P$ means $\lan +|P\ran =1$. Conservation of probability mandates that $\k+ \in {\rm ker}(H^T)$, and means that for any operator $O$, we may write an evolution equation for the expectation value $\lan O\ran$ as 
		\be \p_t \lan O \ran = \lan + | [O,H] | P\ran.\ee 
		The general approach adopted below will be to calculate these evolution equations for a few simple observables (such as the magnetization and short-ranged correlation functions thereof), simplifying the expectation value on the RHS using an appropriate type of mean field ansatz. 
		
		We are interested in two-state automata, and as such will write $H$ in terms of the operators 
		\be a_i \equiv \frac{X_i-iY_i}2,\qq n_i \equiv \frac{1-Z_i}2,\qq \ob n_i \equiv 1-n_i,\ee
		where $X_i,Y_i,Z_i$ are Pauli matrices at site $i$, with the $n_i,a_i$ satisfying the usual relations
		\be [n_i,a_j] = -\d_{i,j} a_j,\qq  [n_i,a_j^\da] = \d_{i,j}a_j^\da.\ee  
		$H$ is constructed by writing down and enumerating the ways that spins can flip under the dynamics, while taking care to ensure that $\lan +|H = 0$. As a simple example, diffusion of hardcore particles on a lattice would be controlled by the Hamiltonian 
		\be H = \sum_{\lan i,j \ran } (a_i^\da a_j - \ob n_i n_j),\ee 
		where the second term ensures probability conservation $\lan + | H = 0$ on account of the relations 
		\be \lan +| a_i = \lan +|n_i,\qq \lan +|a_i^\da = \lan +|\ob n_i.\ee 
		
		Suppose that the noise-free automaton applies $l$ different types of updates, each occurring with probability $q_a, a = 1,\dots,l$ (e.g. in two dimensions, $\sfR_{2,3}$ both contain two types of updates, each occurring with probability $1/2$). Then in general, the noisy automata dynamics under consideration---where, as in the main text, noise of bias $\eta$ is applied with strength $p$---can always be formulated using the Hamiltonian 
		\bea H & = \sum_i \Bigg( (a^\da_i - \ob n_i) \((1-p) \sum_{a=1}^l q_a \Pi^{a,-1\ra 1}_i + p \frac{1+\eta}2\) + (a_i - n_i) \((1-p) \sum_{a =1}^l q_a \Pi^{a,1\ra -1}_i + p \frac{1-\eta}2\) \Bigg),\eea 
		where $\Pi^{a,\pm 1 \ra \mp 1}_i$ is a projector onto the states that would result in the spin at site $i$ flipping from $\pm1$ to $\mp1$ under the action of the update rule labeled by $a$. For the dynamics of interest to us, these projectors will always be constructed as a polynomial in the operators $n_j$ for $j$ nearby $i$. 
		Since we will only be interested in the dynamics of expectation values of products of the $n_i$, we only need to calculate the commutator of $H$ with $n_i$. This is easily done: let a given such operator be written as $O = \prod_{i \in I} n_i$ for some finite set of indices $I$. Then we just need the commutator  
		\bea \lan +| [n_i,H] & =  \lan + | \ob n_i \((1-p)\sum_a q_a \Pi^{a,-1\ra1}_i + p \frac{1+\eta}2\) - \lan + | n_i \( (1-p)\sum_a q_a \Pi_i^{a,1\ra -1} + p \frac{1-\eta}2\) \\ 
		& = \lan + | (1-p) \sum_a q_a (\Pi_i^{a,-1\ra 1} - \Pi_i^{a,1\ra -1}) + \lan + | p \( \frac{1+\eta}2 - n_i\).\eea 
		To simplify this, define $\Pi^{a,\ra 1}_i$ as the projector onto all states for which an application of the automaton update $a$ would result in the spin at site $i$ being $+1$. We may then write 
		\be \Pi^{a,-1\ra1}_i = \Pi^{a,\ra1}_i - \Pi^{a,1\ra1}_i,\ee 
		so that 
		\be \Pi^{a,-1\ra1}_i - \Pi^{a,1\ra-1}_i = \Pi^{a,\ra1}_i - \Pi^{a,1\ra1}_i - \Pi^{a,1\ra-1}_i = \Pi^{a,\ra1}_i - n_i.\ee 
		
		With this, a short calculation gives 
		%	\bea \p_t \lan O \ran & = \sum_{i\in I} \left \lan \prod_{j \in I_i}n_j \( (1-p) \( \ob n_i \sum_{a = 1}^l q_a  \Pi^{a,-1\ra 1}_i - n_i \sum_{a=1}^l q_a \Pi^{a,1\ra -1}_i  \) + p \(\frac{1+\eta}2 - n_i\)\) \right \ran.\eea 
		\bea \p_t \lan O \ran & = \sum_{i\in I} \left \lan \prod_{j \in I_i}n_j \( (1-p) \( \sum_{a = 1}^l q_a(\Pi^{a,\ra1}_i  - n_i)  \) + p \(\frac{1+\eta}2 - n_i\)\) \right \ran\eea 
		where we defined $I_i = I \setminus \{ i \}$.
		
		Even in the simplest case where $O = n_i$, this equation is in general untractable: the RHS will involve expectation values of products of the $n_j$ (unless each $\Pi^{a,\ra1}_i = n_{k_\a}$ for some site $k_a$, in which case the dynamics is trivial). To solve for the dynamics of $\lan n_i\ran$ we must therefore also calculate evolution equations for multi-body products of the $n_i$, which in turn involve yet higher-order expectation values; as discussed in the main text, the infinite hierarchy of equations thus generated must be truncated by assuming that correlation functions factorize at sufficiently large distances.

		\ss{Warmup: Toom's rule} \label{app:toom_mf}
		
		Before discussing how the truncation works for the squeezing codes---for which higher body correlation functions need to be included in the hierarchy---we first review how mean field theory works for the simpler (and less interesting) case of Toom's rule, which we will analyze in $d=2$ (see also \cite{bennett1985role} and \cite{lebowitz1990statistical}).
		
		The noiseless automaton possesses only a single type of update, for which a $-1$ spin at site $i$ flips under the dynamics if both $s_{i+x}, s_{i+y}$ are $+1$, and analogously for a $+1$ spin. The relevant projector is accordingly 
		\be \Pi^{\ra 1}_i = n_i n_{i+x} + n_{i+x} n_{i+y} + n_i n_{i+y} - 2 n_in_{i+x} n_{i+y}.\ee 
		
		The simplest mean field ansatz assumes that correlation functions on different sites factorize: 
		\be \left\lan \prod_{i \in I} n_i \right\ran \ra \prod_{i\in I} \lan n_i\ran.\ee 
		Under this assumption, and rewriting things in terms of the expected magnetization 
		\be \lan n_i \ran = \frac{1+m_i}2,\ee 
		we obtain 
		\be \p_t m_i = p(\eta - m_i) + \frac{1-p}2\(-m_i + m_{i+x} + m_{i+y} - m_i m_{i+x}m_{i+y}\).\ee 
		
		If we specify to spatially uniform configurations $m_i = m$, the mean-field equations will involve only a single variable $m$. In this case, there is no possibility of non-reciprocity in the equations of motion, and we may always write $\p_t m = -\d F_{\rm eff} / \d m$ for some effective free energy $F_{\rm eff}$. In the present setting, doing so gives 
		\be \p_t m = p\eta + \frac{1-3p}2 m - \frac{1-p}2 m^3,\ee 
		corresponding to an effective free energy of 
		\be F_{\rm eff} = -p\eta m + \frac{3p-1}4 m^2 + \frac{1-p}8 m^4.\ee 
		This free energy is bounded, and at $\eta=0$ has a second-order phase transition which in mean field occurs at $p_{c,MF} = 1/3$.

		\ss{$\sfR_2$} \label{ss:s2mf}
		
		We now turn to the $\sfR_2$ squeezing code defined in \eqref{sqzrule_r2}. 
		We will in fact consider a slightly more general setup in which $\wedge$ squeezing ($-1$ expansion) occurs along the first $(d-r)/2$ spatial directions, and $\vee$ squeezing ($+1$ expansion) occurs along the remaining $(d+r)/2$ directions, with each direction of update occurring with equal probability $q_a = 1/d$, $a = 1,\dots,d$. This dynamics possesses a $\zt$ symmetry generated by a spin flip and a reflection only when $d$ is even and $r=0$, and we will see momentarily that order (with asynchronous updates) is possible only when this symmetry is present. 
		
		The projectors in question are 
		\be \Pi^{a,\ra1}_i = \begin{dcases} n_{i+a}n_{i-a}  \qq & a \leq \frac{d-r}2 \\ 1 - \ob n_{i-a} \ob n_{i+a}
			\qq& a > \frac{d-r}2 \end{dcases} \ee 
		For any integer $l>0$, define the ``domain wall'' operators
		\bea  d^{la}_i & \equiv n_i n_{i+la} \\ 
		\ob d^{la}_i & \equiv \ob n_i \ob n_{i+la}\eea 
		and define 
		\be D_i \equiv \sum_{a=1}^{(d-r)/2} d^{2a}_i - \sum_{a=(d-r)/2+1}^d \ob d^{2a}_i .\ee 
		Then the time evolution of a generic product of $n_i$ operators is given by 
		\be \p_t \lan O \ran  = \sum_{i \in I} \left\lan \prod_{j\in I_i} \( p\( \frac{1+\eta}2 - n_i\) +  \frac{1-p}d\( \frac{d+r}{2} - dn_i + D_i\)\)\right\ran \ee 
		
		\sss{Fluctuationless mean field} 
		
		Let us first examine the fluctuationless limit, where we take all connected correlation functions of $n_i$ to vanish. We will furthermore assume that the expectation value of $n_i$ is translation invariant,\footnote{An assumption which empirically is seen to be true in the steady state.} as keeping track of spatial derivatives in the mean field equations will not be important for any of the main points to follow. With these assumptions we have 
		\be \lan D \ran = (d+r) \( \lan n \ran - \frac12 \) - r\lan n\ran^2 ,\ee 
		and so 
		\be \p_t \lan n \ran = p \( \frac{1+\eta}2 - \lan n \ran \) + \frac{1-p}d r \lan n \ran(1 - \lan n \ran).\ee 
		We will usually find it more intuitive to write our final evolution equations in terms of the magnetization 
		\be m_i = 2n_i -1,\ee 
		for which (omitting the $\lan\,\ran$s)
		\be \p_t m = p(\eta_{\rm eff} - m) -r \frac{1-p}{2d}  m^2 \ee 
		where the effective bias is $\eta_{\rm eff} = \eta + \frac{1-p}{2dp}r$. This evolution equation corresponds to an effective free energy of 
		\be F_{\rm eff} = - p \eta_{\rm eff} m +  \frac p2 m^2 + r\frac{1-p}{6d} m^3.\ee 
		
		There are a few things to take away from this. First, consider the symmetric case where $d$ is even and $r=0$. We then get the extremely simple
		\be \p_t m = p(\eta - m),\ee 
		which is the evolution equation we would get in the presence of noise alone: $m$ simply relaxes to the average value $\eta$ of the noise over a timescale $1/p$. As discussed in the main text, the effects of the error-correction being performed by the $\sfR_2$ automaton are completely absent when fluctuations of the magnetization are neglected, and there is no possibility of having an ordered phase in this limit. It is possible for the presence of fluctuations to change this conclusion, but since fluctuations enter only at an order of at most $1/d$, the critical noise strength must vanish with $d$ at least as fast as $p_c \sim 1/d$ (we will see that in fact $p_c \sim 1/d^2$). 
		
		Secondly, when $r\neq 0$, the absence of symmetry is reflected in both the modification of $\eta$, and in the generation of an $m^3$ term in the effective free energy. This gives a free energy with only one local minimum (the other extremal point is a local maximum), which is not compatible with the existence of a stable memory. Furthermore, the size of both the induced bias and the $m^2$ term appearing in $\p_t m$ are $O(p^0)$ as $p\ra 0$. Thus our conclusion about the absence of two stable fixed points for $m$ in the range $[-1,1]$ will remain true beyond the present mean field approximation as long as $d$ is large enough. This means that there can be no order in large enough $d$ with asynchronous updates if $r\neq 0$, viz. if the $\zt$ symmetry is not present. We will later numerically check that in fact $p_c = 0$ when $d=3$: thus with asynchronous updates, $\sfR_2$ can only order when $d \in 2\zz$ and $r=0$. We will accordingly specify to this case in the remainder of this subsection. 
		
		\sss{Adding fluctuations} 
		
		Assuming that the symmetries that permute among the first $d/2$ dimensions and the last $d/2$ are not spontaneously broken (this statement is vacuous when $d=2$, and when $d>2$ it is numerically observed to be true), we may write 
		\be \frac{\lan D \ran}d = \frac{\lan d^{2x} \ran - \lan \ob d^{2y} \ran}2,\ee 
		where $x$ is one of the first $d/2$ coordinates, and $y$ is one of the last $d/2$. The exact evolution equation for $\lan n \ran$ is then 
		\bea \p_t \lan n\ran & = \frac12\( 1+\eta p - 2\lan n\ran + (1-p)(\lan d^{2x}\ran - \lan \ob d^{2y} \ran) \).
		%	\\ & =  p\(\frac{1+\eta}2 - \lan n \ran \) + \frac{1-p}2( \lan d^{2x}\ran - \lan  d^{2y}\ran ) \\ 
		%	& =p\(\frac{1+\eta}2 - \lan n \ran \) + \frac{1-p}2( \lan d^{2x}\ran_c - \lan  d^{2y}\ran_c ) 
		\eea 
		We thus need to solve for the as-yet unknown expectation values $\lan d^{2x}\ran,\lan \ob d^{2y}\ran$. A straightforward calculation gives 
		\bea \p_t \lan d^{la}\ran & = (1+\eta p)\lan n \ran - 2 \lan d^{la}\ran +2 \frac{1-p}d \lan n_j D_{j+la}\ran\eea 
		where we have used $\lan n_j D_{j+la}\ran = \lan n_{j+la} D_j\ran$.\footnote{Translation implies $\lan n_{j+la}D_j\ran = \lan n_j D_{j-la}\ran$ and then $\lan d^{la}\ran =\lan d^{-la} \ran$.} By symmetry, we similarly obtain
		\be \p_t \lan \ob d^{la}\ran = (1-\eta p)\lan \ob n \ran - 2\lan \ob d^{la}\ran - 2 \frac{1-p}d \lan \ob n_j D_{j+la}\ran.\ee 
		Again under these assumptions, we have 
		% \be \lan n_j D_{j+2a}\ran= \begin{dcases} \lan n_j n_{j+x}n_{j+3x}\ran  + (d/2-1) \lan n_j n_{j+2x-w} n_{j+2x+w}\ran - \frac d2 \lan n_j \ob n_{j+2x-y} \ob n_{j+2x+y}\ran & a \leq  d/2 \\ 
			% 	\frac d2 \lan n_j n_{j+2y-x} n_{j+2y+x}\ran - 
			% 	\lan n_j \ob n_{j+y} \ob n_{j+3y}\ran  - (d/2-1) \lan n_j \ob n_{j+2y-z} \ob n_{j+2y+z}\ran & a >  d/2 \end{dcases}
		% \ee 
		\be \lan n_j D_{j+2a}\ran=  \lan n_j n_{j+x}n_{j+3x}\ran  + (d/2-1) \lan n_j n_{j+2x-w} n_{j+2x+w}\ran - \frac d2 \lan n_j \ob n_{j+2x-y} \ob n_{j+2x+y}\ran \qq (a \leq  d/2) \ee 
		and 
		\be \lan \ob n_j D_{j+2a}\ran= - \lan \ob n_j \ob n_{j+y}\ob n_{j+3y}\ran  - (d/2-1) \lan \ob n_j \ob n_{j+2y-z} \ob n_{j+2y+z}\ran + \frac d2 \lan \ob n_j n_{j+2x-y} n_{j+2x+y}\ran \qq (a > d/2) \ee
		where $w\neq x$ is an arbitrary coordinate with $w \leq d/2$ and $z\neq y$ is a coordinate with $z>d/2$. 
		
		We will truncate the hierarchy of cluster mean-field equations by only keeping correlations between nearest-neighbors and those next-nearest-neighbors which are of the form $(i,i\pm2a)$ for some $a$ (thus for example we do not keep correlations between nnn pairs like $(i,i+a+b)$ for $a\neq b$---this choice keeps the minimal number of variables needed to get an ordered phase). This means that we take, e.g.,
		\be \lan n_j n_{j+2a-b} n_{j+2a+b} \ran \ra \lan n\ran \lan d^{2b} \ran \ee 
		and
		\be \lan n_j n_{j+a} n_{j+3a}\ran \ra \frac{\lan d^{1a} \ran \lan d^{2a}\ran}{\lan n\ran}.\ee 
		Letting $P(\Pi_{i\in I} n_i)$ be the probability of finding $n_i=+1$ at all sites in the set $I$, these equations follow from 
		\bea P(n_j n_{j+a} n_{j+3a}) & = P(n_j|n_{j+a}n_{j+3a}) P(n_{j+a}n_{j+3a}) \\ & = P(n_j|n_{j+a}) P(n_{j+a}n_{j+3a}) \\ & = P(n_jn_{j+a})P(n_{j+a}n_{j+3a})/P(n_{j+a}) \\ & = \lan d^{1a}\ran \lan d^{2a}\ran / \lan n\ran,\eea
		where our factorization assumption enters in the second equality. 
		%This approach should always be valid close to $p=1$, where the correlation length must vanish. Thus while it may not 
		
		Using these factorizations, we have 
		\be \lan n_j D_{j+2a}\ran = \frac{\lan d^{1x} \ran \lan d^{2x}\ran }{\lan n \ran} + (d/2-1) \lan n \ran \lan d^{2x}\ran - \frac d2\lan n \ran \lan \ob d^{2y}\ran  \qq (a\leq d/2) \ee 
		and 
		\be \lan \ob n_j D_{j+2a}\ran = -\frac{\lan \ob d^{1y} \ran \lan \ob d^{2y}\ran }{\lan \ob n \ran} - (d/2-1) \lan \ob n \ran \lan \ob d^{2y}\ran + \frac d2\lan \ob n \ran \lan d^{2x}\ran  \qq (a>d/2). \ee 
		Similarly, 
		\be \lan n_j D_{j+a}\ran = \lan d^{2x}\ran (1-\lan n \ran) + \frac d2 \lan n \ran (\lan d^{2x}\ran - \lan \ob d^{2y}\ran ) \qq(a\leq d/2)\ee 
		and 
		\be \lan \ob n_j D_{j+a}\ran = -\lan \ob d^{2y}\ran (1-\lan \ob n \ran) + \frac d2\lan \ob n \ran (\lan d^{2x} \ran - \lan \ob d^{2y}\ran ) \qq (a>d/2)\ee 
		where we used relations like $n_j^2 = n_j$ and $\lan n_j n_{j+a-b} n_{j+a+b} \ran \ra  \lan n \ran \lan d^{2b}\ran$ for $a \neq b$ (the latter following from our neglect of connected diagonal next-nearest-neighbor correlations). 
		
		Collecting results, our cluster mean field equations read 
		\bea 
		\p_t \lan n\ran & = \frac12\( 1+\eta p - 2\lan n\ran + (1-p)(\lan d^{2x}\ran - \lan \ob d^{2y} \ran) \) \\ 
		\p_t \lan d^{1x}\ran & = (1+\eta p)\lan n \ran -2\lan d^{1x}\ran + 2\frac{1-p}d\( \lan d^{2x}\ran \lan \ob n \ran + \frac d2 \lan n \ran (\lan d^{2x} \ran - \lan \ob d^{2y} \ran )\) \\ 
		\p_t \lan \ob d^{1y}\ran & = (1-\eta p)\lan \ob n \ran -2\lan \ob d^{1y}\ran-  2\frac{1-p}d\( -\lan\ob  d^{2y}\ran \lan n \ran + \frac d2 \lan \ob n \ran (\lan d^{2x} \ran - \lan \ob d^{2y} \ran )\) \\ 
		\p_t \lan d^{2x} \ran & = (1+\eta p)\lan n \ran - 2\lan d^{2x}\ran + 2 \frac{1-p}d \(\lan d^{2x}\ran \(\frac{\lan d^{1x}\ran }{\lan n \ran} - \lan n \ran \) + \frac d2 \lan n \ran (\lan d^{2x}\ran - \lan \ob d^{2y} \ran ) \)\\ 
		\p_t \lan \ob d^{2y} \ran & = (1-\eta p)\lan \ob n \ran - 2\lan \ob d^{2y}\ran - 2 \frac{1-p}d \(-\lan \ob d^{2y}\ran \(\frac{\lan \ob d^{1y}\ran }{\lan \ob n \ran} - \lan \ob n \ran \) + \frac d2 \lan \ob n \ran (\lan d^{2x}\ran - \lan \ob d^{2y} \ran ) \)
		\eea 
		
		It will be more useful to write these equations in terms of variables that transform in definite representations of the $\zt$ symmetry, viz. the magnetization and the linear combinations 
		\be d^{l\pm} \equiv \frac{d^{lx} \pm \ob d^{ly}}2.\ee 
		$d^{l+}$ is neutral under the $\zt$ symmetry, while $d^{l-}$ transforms in the same way as $m$. 
		Some algebra turns the MF equations into (omitting $\lan\, \ran$s)
		\bea \p_t m & = p\eta - m + 2(1-p) d^{2-} \\ 
		\p_t d^{1+} & = \frac{1+\eta p m}2 - 2d^{1+} + \frac{1-p}d \(d^{2+} + (d-1)md^{2-}\) \\ 
		\p_t d^{1-} & = \frac{\eta p + m}2 - 2d^{1-} + \frac{1-p}d\( -m d^{2+} + (1+d)d^{2-}\) \\ 
		\p_t d^{2+} & = \frac{1+\eta pm}2 - 2d^{2+} + 4\frac{1-p}{d(1-m^2 )} \( d^{2+} d^{1+} + d^{2-}d^{1-} -m(d^{2+}d^{1-} + d^{2-}d^{1+})\)   + \frac{1-p}d\((d-1)md^{2-} - d^{2+}\) \\ 
		\p_t d^{2-} & = \frac{m+\eta p}2 - 2 d^{2-} + 4\frac{1-p}{d(1-m^2)}\(d^{1+}d^{2-} + d^{1-}d^{2+} - m(d^{2+}d^{1+} + d^{2-}d^{1-})\) \\ 
		& \qq + \frac{1-p}d\((d-1)d^{2-} - md^{2+} \) \eea 
		
		The equations further simplify if we write them not in terms of the $d^{l\pm}$, but rather in terms of the connected correlation functions of the magnetization. Define the correlation functions 
		\be f^{la}_i \equiv \lan m_i m_{i+la} \ran  - \lan m_i \ran\lan m_{i+la} \ran, \ee 
		together with their linear combinations 
		\be f^{l \pm} \equiv \frac{f^{lx} \pm f^{ly} }2.\ee 
		The $f^{l\pm}$ will vanish both at $p=0$ and $p=1$, where the system has no nontrivial spatial correlations. A short calculation gives the following relations between $d^{l \pm}$ and $f^{l \pm}$: (again assuming translation invariance and dropping spatial indices)
		\bea \lan d^{l+} \ran & = \frac14(f^{l+} + 1+ m^2) \\
		\lan d^{l-} \ran & = \frac14(f^{l-} + 2 m ).
		\eea 
		Substituting this in our previous MF equations yields, after algebraic simplification, the equations quoted in the main text: 
		%	\footnote{For this it is helpful to use $(1\pm \eta p)\lan n \ran - 2\lan d^{la}\ran = \frac12(\pm \eta p + \lan m \ran - \lan m\ran^2 - f^{la})$, where the $-$ sign is applied if the $n$s and $d$s are replaced by their PH conjugates.}
		\bea
		\p_t m  & = p(\eta-m) + \frac{1-p}2 f^{2-} \\ 
		\p_t f^{1-} & = -2f^{1-} + \frac{1-p}d \(m(1-m^2) + f^{2-} - mf^{2+} \) \\ 
		\p_t f^{2-} & =  -2f^{2-} + \frac{1-p}d \( f^{1-} + mf^{1+} + \frac1{1-m^2} \( f^{1-}f^{2+} + f^{1+}f^{2-} - m(f^{1-}f^{2-} + f^{1+} f^{2+})\)\)\\ 
		\p_t f^{1+} & = -2f^{1+} + \frac{1-p}d \(1-m^2 + f^{2+} - mf^{2-} \) \\ 
		\p_t f^{2+} & =  -2f^{2+} + \frac{1-p}d \( f^{1+} + mf^{1-} + \frac1{1-m^2} \( f^{1+}f^{2+} + f^{1-}f^{2-} - m(f^{1+}f^{2-} + f^{1-} f^{	2+})\)\)
		\eea

		\begin{figure*}[t]
			\centering{\includegraphics[width=0.32\textwidth]{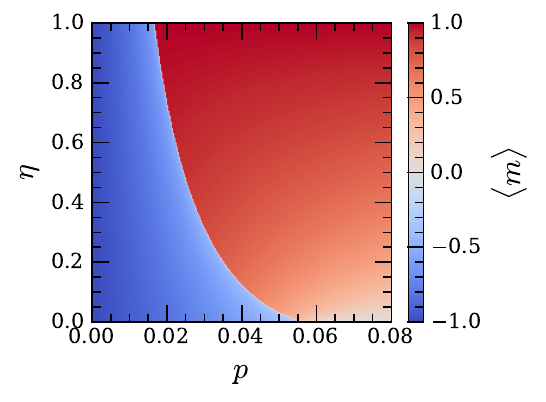}} 
    {\includegraphics[width=0.32\textwidth]{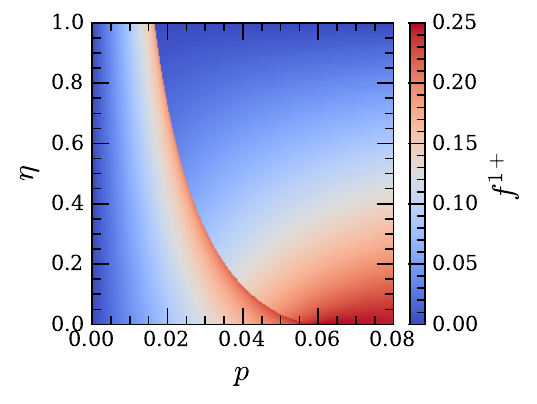}} 
			{\includegraphics[width=0.32\textwidth]{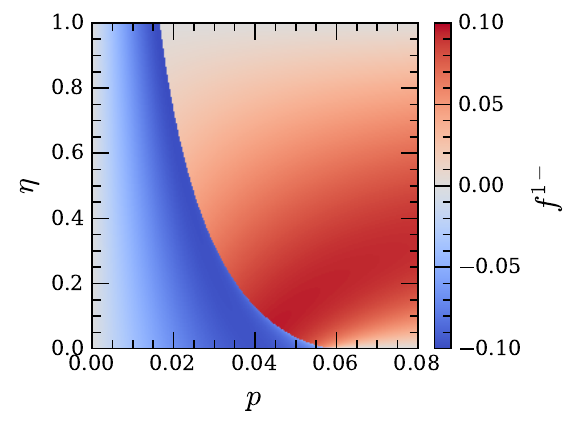}} 
			\caption{Mean-field phase diagrams for the $\sfR_2$ automaton in $d=2$. {\it left:} magnetization, {\it center:} $\zt$-neutral nearest neighbor fluctuations, and {\it right:} $\zt$-odd nearest neighbor fluctuations. The plots are made by starting from a logical state with $m = -1, f^{1\pm} = 0$ and evolving under the dynamical mean-field equations until a steady state is reached.  }
			\label{fig:s2mfpd}
		\end{figure*}
		
		The phase diagram in $d=2$ one obtains from these equations is shown in Fig.~\ref{fig:s2mfpd}.

		\ss{$\sfR_3$} \label{ss:r3mf}
		
		We now consider the three-variable rule $\sfR_3$, which as we will see has rather different physics. Like $\sfR_2$, similar arguments show that $\sfR_3$ is disordered under asynchronous updates in all cases where a $\zt$ symmetry is not present. For this reason we will continue to specify to even dimensions, with the number of $\wedge$ and $\vee$ squeezing directions both equal to $d/2$. 
		
		The relevant projector that determines the dynamics is now cubic in the number operators: 
		\be \Pi^{a,\ra1}_i = \begin{dcases} n_{i+a}n_in_{i-a}  \qq & a \leq \frac{d}2 \\ 1 - \ob n_{i-a} \ob n_i\ob n_{i+a}
			\qq& a > \frac{d}2 \end{dcases}. \ee 
		Define the operators 
		\be D^1_i \equiv \sum_{a=1}^{d/2} d^{2a}_i,\qq D^2_i \equiv \sum_{a=d/2+1}^d \ob d^{2a}_i.\ee 
		Then a generic product of number operators evolves as
		\be \label{s3product} \p_t\lan O\ran =  \sum_{i \in I} \left\lan \prod_{j\in I_i} \( p\( \frac{1+\eta}2 - n_i\) +  (1-p)\( \frac{1}{2} - n_i + \frac1d(n_iD_i^1-\ob n_i D^2_i)\)\)\right\ran. \ee 
		
		\sss{Fluctuationless mean field}
		
		We begin by looking at the limit in which all correlation functions factorize and are translation invariant. In this limit, we have 
		\be \frac1d\lan n_i D^1_i - \ob n_i D^2_i \ran = \frac18(3m + m^3).\ee 
		We then find 
		\be \p_t m = p\eta - \frac{1+3p}4 m + \frac{1-p}4 m^3\ee 
		which corresponds to a free energy of 
		\be F_{\rm eff} = -p \eta m + \frac{1+3p}8m^2 - \frac{1-p}{16}m^4.\ee 
		
		Like $\sfR_2$, no ordering is possible within mean field theory. Unlike $\sfR_2$, fluctuation corrections to this result cannot produce order once $d$ is large enough, since in the present case the mass term in $F_{\rm eff}$ is {\it  positive} and non-vanishing in the limit $p\ra0$. Thus there must be a critical dimension $d_c$ such that the system is disordered for {\it all} $p$ as long as $d > d_c$. This dimension turns out to be $d_c=2$, so that $\sfR_3$ with asynchronous updates orders {\it only} in two dimensions. 
		
		\sss{Adding fluctuations}
		
		The expectation values $\lan n_j D^1_j\ran, \lan \ob n_j D^2_j\ran$ on the RHS of \eqref{s3product} involve sums of operators of the form $\lan n_i n_{i+a} n_{i+2a}\ran$. Since these do not factorize even when only {\it nearest} neighbors are taken into account, we will simplify the calculation by doing an extended cluster mean field theory involving only the $d^{1a}$, taking correlation functions of the $n_j$ to factorize beyond nearest neighbors (for $\sfR_2$ we were forced to account for next nearest neighbor correlations since $\p_t\lan n \ran$ would otherwise trivially lead to $\lan n \ran = (1+\eta)/2$). In this scheme, and working with the assumption that translation and the relevant reflection symmetries are unbroken, we have 
		\be \lan n_j D^1_j\ran =  \sum_{a=1}^{d/2}\frac{\lan d^{1a}_{j-a}\ran\lan d^{1a}_j \ran}{\lan n_j\ran},\qq \lan \ob n_j D^2_j\ran = \sum_{b=d/2+1}^d \frac{\lan \ob d^{1b}_{i-b}\ran\lan \ob d^{1b}_i \ran}{\lan \ob n_j \ran}.\ee

		We then just need the evolution equations for $\lan d^{1a}_j\ran$. These are (here as before $x$ is any coordinate less than $d/2+1$, and $y$ is any coordinate greater than $d/2$)
		\bea \p_t \lan d^{1x}_j \ran & =\frac{1+\eta p}2 \lan n_j +n_{j+x}\ran - 2\lan d^{1x}_j\ran + \frac{1-p}d\( \lan d^{1x}_j (D^1_{j+x} +D^1_j)\ran - \lan n_j \ob n_{j+x} D^2_{j+x} + n_{j+x}\ob n_j D^2_j \ran \)  \\ 
		\p_t \lan \ob d^{1y}_j \ran & =\frac{1-\eta p}2 \lan \ob n_j +\ob n_{j+y}\ran - 2\lan \ob d^{1y}_j\ran + \frac{1-p}d\( \lan \ob d^{1y}_j (D^2_{j+y} +D^2_j)\ran - \lan \ob n_j n_{j+y} D^1_{j+y} + \ob n_{j+y} n_j D^1_j \ran \)  \eea 
		For spatially uniform solutions, this simplifies to 
		\bea \p_t \lan d^{1x}\ran & = (1+\eta p)\lan n \ran -2\lan d^{1x}\ran + 2 \frac{1-p}d\lan n_{j+x}n_j D^1_j - n_{j+x}\ob n_j D^2_j\ran \\ 
		\p_t \lan d^{1y} \ran & = (1-\eta p) \lan \ob n\ran - 2 \lan \ob d^{1y}\ran - 2 \frac{1-p}d \lan \ob n_j n_{j+y} D_{j+y}^1 - \ob n_j \ob n_{j+y} D^2_{j+y} \ran.\eea 
		
		Our factorization assumption lets us simplify the four-point functions above, giving e.g. 
		\bea \lan d^{1x}_j (D^1_{j+x} + D^1_j)\ran & = \lan d^{1x}_j \ran \De^{1x} \( \frac{\lan d^{1x}_{j-x} \ran}{\lan n_j\ran} + \sum_{b\neq a}^{d/2} \frac{\lan d^{1b}_{j-b} \ran\lan d^{1b}_j \ran }{\lan n_j\ran^2} \)\eea 
		where 
		\be \De^{la}(\mco_j) \equiv \mco_j + \mco_{j+la}.\ee 
		We also have 
		\bea \lan n_j \ob n_{j+x} D^2_{j+x} + n_{j+x}\ob n_j D^2_j\ran & = \sum_{b=d/2+1}^d \( \frac{ (\lan n_j\ran -\lan d^{1x}_j\ran)\lan\ob d^{1b}_{j+x} \ran \lan \ob d^{1b}_{j-b+x} \ran }{\lan \ob n_{j+x}\ran^2} + \frac{(\lan n_{j+x}\ran - \lan d^{1x}_j \ran) \lan \ob d^{1b}_j\ran \lan \ob d^{1b}_{j-b} \ran }{\lan \ob n_j \ran^2} \) \eea 
		
		\bea\lan n_{j+x}n_j D^1_j \ran &= \lan n_{j+x} n_j n_{j-x}\ran + (d/2-1)\lan n_{j+x}n_jn_{j-w}n_{j+w} \ran \\ 
		& \ra \frac{\lan d^{1x}\ran^2}{\lan n \ran} + (d/2-1) \frac{\lan d^{1x}\ran^3}{\lan n\ran^2}\eea 
		and similarly 
		\bea \lan n_{j+x} \ob n_j D^2_j\ran & \ra \frac d2\frac{\lan \ob d^{1y}\ran^2}{\lan \ob n \ran^2} (\lan n \ran - \lan d^{1x}\ran)\\ 
		\lan \ob n_j n_{j+y} D^1_{j+y}\ran& \ra \frac d2\frac{\lan d^{1x}\ran^2}{\lan n\ran^2} (\lan \ob n \ran - \lan \ob d^{1y}\ran ) \\ 
		\lan \ob n_j \ob n_{j+y} D^2_{j+y}\ran & \ra \frac{\lan \ob d^{1y}\ran^2}{\lan \ob n \ran} + (d/2-1) \frac{\lan \ob d^{1y} \ran^3}{\lan \ob n \ran^2}.
		\eea 
		
		Recapitulating, the mean field equations are 
		\bea \p_t \lan n \ran & = \frac{1+\eta p}2 - \lan n \ran + \frac{1-p}2 \( \frac{\lan d^{1x} \ran^2}{\lan n \ran } - \frac{\lan \ob d^{1y}\ran^2}{\lan \ob n \ran} \)\\ 
		\p_t \lan d^{1x} \ran & = (1+\eta p)\lan n \ran - 2 \lan d^{1x}\ran + (1-p)\( \frac{\lan d^{1x}\ran^2}{\lan n\ran} \( \frac2d + \(1-\frac2d\)\frac{\lan d^{1x}\ran}{\lan n\ran}\) - \frac{\lan \ob d^{1y} \ran^2}{\lan \ob n\ran^2}(\lan n \ran - \lan d^{1x}\ran) \)  \\ 
		\p_t \lan \ob d^{1y}\ran & = (1-\eta p)\lan \ob n \ran - 2 \lan \ob d^{1y}\ran + (1-p)\( \frac{\lan \ob d^{1y}\ran^2}{\lan \ob n\ran} \( \frac2d + \(1-\frac2d\)\frac{\lan \ob d^{1y}\ran}{\lan \ob n\ran}\) - \frac{\lan d^{1x} \ran^2}{\lan  n\ran^2}(\lan \ob n \ran - \lan \ob d^{1y}\ran) \)  
		\eea

		Written in terms of variables $m,d^\pm$ which have definite charges under the $\zt$ symmetry, this becomes 
		\bea \p_t m & = \eta p -m + \frac{4(1-p)}{1-m^2} \(2d^+d^- - m((d^+)^2 + (d^-)^2)\) \\ 
		\p_t d^+ & = \frac{1+\eta pm}2 - 2d^+ \\ & + 2\frac{1-p}{(1-m^2)^2} \Big( ((d^+)^2+(d^-)^2) \(- 1 + 2/d - (3+2/d)m^2 + 4(1-1/d)(1+m^2)d^+ + (8/d)m d^- \) \Big)  \\ 
		&  + 2d^+d^- \( (3-2/d)m + (2/d+1)m^3 - 4/d(1+m^2)d^-  - 8(1-1/d)md^+\) \\ 
		\p_t d^- & = \frac{m+\eta p}2 - 2d^-  \\ & + 2\frac{1-p}{(1-m^2)^2} \Big( ((d^+)^2+(d^-)^2) \( - (3+2/d)m   -(1-2/d)m^3 +4(1-1/d)(1+m^2)d^- + (8/d)md^+ \)  \\ 
		&  + 2d^+d^- \( 1+2/d + (3-2/d)m^2 - 4/d (1+m^2)d^+ - 8(1-1/d)md^-\)   \Big)  
		\eea 
		Finally, switching from $d^\pm$ to the connected correlation functions $f^\pm$ of the magnetization, after lengthy algebraic steps, produces 
		\bea \p_t m & = \eta p - \frac{1+3 p}4 m + \frac{1-p}4m^3 + \frac{1-p}4 \( \frac{2 f^+f^- - m((f^+)^2 + (f^-)^2)}{1-m^2} + 2(mf^+ + f^-) \) \\ 
		\p_t f^- & = -2f^-  + \frac{1-p}{2d} \Big( 2m(1-m^2) + (1+d+(d-3)m^2)f^- - 2f^+m + 2df^-f^+  \\ 
		& \quad -\frac1{1-m^2} \(2m ((f^-)^2 + (f^+)^2) + 2(1-3m^2)f^+f^- \) \\ 
		& \quad + \frac1{(1-m^2)^2} \big(2m(f^+)^3 - 2(2d-3)m(f^-)^2f^+ + (d-3)(1+m^2)f^-(f^+)^2  + (d-1)(1+m^2)(f^-)^3 \big) \Big)
		\\ 
		\p_t f^+ & = -2f^+  + \frac{1-p}{2d} \Big( 1-m^4 + (1+d+(d-3)m^2)f^+ - 2f^-m + 2d(f^+)^2  \\ 
		& \quad -\frac1{1-m^2} \(4m f^-f^+ +  (1-3m^2)((f^+)^2 + (f^-)^2) \) \\ 
		& \quad + \frac1{(1-m^2)^2} \big(2m(f^-)^3 - 2(2d-3)mf^-(f^+)^2 + (d-3)(1+m^2)(f^-)^2f^+  + (d-1)(1+m^2)(f^+)^3 \big) \Big).
		\eea 
		%	One rather annoying aspect of these equations is that the ordered fixed point at $p\ra0$ can only be obtained by cancelling the vanishing of $f^\pm$ in the numerators by the vanishing of $1-m^2$ in the denomenators. 
		
		The existence of an ordered phase is most easily established by examining the stability of the disordered fixed point.  Spontaneous symmetry breaking can be ruled out when $d$ is large enough by showing that the disordered fixed point remains stable for all $p$ when $d$ is sufficiently large. To this end, we first solve $\p_t f^+ = 0$ in the $d\ra\infty$ limit. This is done by expanding $f^+$ as a series in $1/d$. The leading $O(d^0)$ part is easily checked to vanish\footnote{The equation for the $O(1)$ part is solved either by $f^+=0$ or by $f^+= \frac2{\sqrt{1-p}} -1$. This latter solution is unphysical, since $f^+$ is bounded from above by $1$.} (as it should, on account of fluctuations vanishing as $d\ra\infty$), while the leading $1/d$ part 
		gives an equation that is easily solved to produce 
		\be f^+ = \frac{1-p}{d(3+p)} + O(1/d^2).\ee 
		
		\begin{figure*}[t]
			\centering
            {\includegraphics[width=0.32\textwidth]{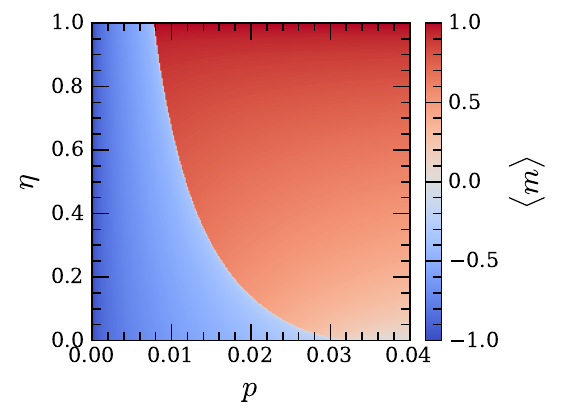}} 
			{\includegraphics[width=0.32\textwidth]{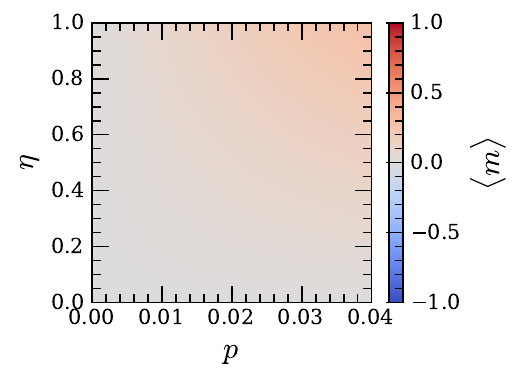}} 
			\caption{Mean-field phase diagrams for magnetization of the $\sfR_3$ automaton's steady state in $d=2$ (left), where the automaton retains an ordered phase, and $d=4$ (right), where it is disordered for all $p,\eta$.  }
			\label{fig:s3mfpd}
		\end{figure*}

		The mass matrix $M^{ab} = - \d \p_t \phi^a / \d \phi^b$ for $\phi = (m,f^-)$ is 
		\be M = \bpm \frac{1+3p}4 - \frac{1-p}4 f^+(2-f^+) & -\frac{1-p}2 (1+f^+) \\ 
		\frac{1-p}d (-(f^+)^3 + (f^+)^2 + f^+ - 1) & 2-\frac{1-p}{2d} ((f^+)^2 (d-3) + 2(d-1) f^+ + d + 1 ) \epm,\ee 
		with $M \neq M^T$ by virtue of the non-equilibrium nature of the present problem. 
		Since $f^+ \sim 1/d$, we have 
		\be M|_{d\ra\infty} = \bpm \frac{1+3p}4 & -\frac{1-p}2 \\ 0 & \frac{3+p}2 \epm + O(1/d).\ee 
		Both eigenvalues of $M$ are thus positive and $O(1)$ even when $p = 0$, provided $d$ is large enough. This shows that the disordered fixed point is always stable for $d > d_c$, with finite $d_c$. In the present mean field analysis it turns out that $d_c = 2$, which agrees with the critical dimension seen in numerics. Solutions of the mean field equations for the magnetization in $d=2,4$ are shown in Fig.~\ref{fig:s3mfpd}.

		\section{Langevin equation from the Magnus expansion} \label{app:magnus}
		
		In this section, we derive the Langevin equation \eqref{langevins} describing the coarsening dynamics of the $\sfR$ squeezing code within the context of the Floquet Glauber dynamics discussed in Sec.~\ref{ss:glauber}. We will assume that the driving period $\o\inv$ is much longer than the local relaxation time, so that the system's evolution is locally in equilibrium. In this limit we expect the dynamics to be well described by a Langevin equation of the form (ignoring the noise for simplicity and using slightly different notation from the main text; here $a$ denotes a spatial direction and repeated indices are summed)
		\be \p_t m =  (\ob K_a + \wt K_a \cos(\o t)) \p_a^2 m - (\ob h + \wt h \cos(\o t)) - V'(m)\ee 
		where the prime on the interaction potential $V$ denotes a derivative with respect to $m$. 
		
		In the limit where the amplitudes $\wt K_a, \wt h$ of the drive are small compared to $\o$, the micromotion induced by the drive can be effectively rotated away using the Floquet-Magnus expansion, producing an effective Langevin equation driven by a time-independent drift operator. For this to be well-controlled, we need to thus work in the regime 
		\be \wt J, \wt h \ll \o \ll \ob J,\ee 
		with the second inequality coming from the local relaxation time being $\sim 1/\ob J$. From our general picture of how squeezing works, we do not expect working in this regime to compromise the existence of a threshold. 
		
		The Magnus expansion is performed using the techniques of \cite{higashikawa2018floquet}, which perform a usual Magnus expansion familiar in quantum mechanics to the differential operator that appears in the Fokker-Planck equation. We will write the Langevin equation as $\p_t m = f(m,\o t)$, with the drift term expanded in Fourier harmonics as 
		\be f(m,\o t) = \sum_l f_l(m) e^{il\o t}.\ee  
		For us only $f_0$ and $f_1 = f_{-1}$ are nonzero, with 
		\be f_0 = \ob K_a \p_a^2 m - \ob h - V'(m),\qq f_{\pm1} = \frac12 (\wt K_a \p_a^2m - \wt h).\ee
		In this case, the results of \cite{higashikawa2018floquet} simplify to give an effective drift force $f_{\rm eff}(m)$ (with no explicit time dependence) of 
		\be f_{\rm eff} = - \frac1{\o^2} [f_1,[f_0,f_1]]\ee 
		to leading order in $1/\o$, where the commutator is defined as a Lie bracket on field space: for two functionals $g,h$ of $m$, 
		\be ([g,h])[m(x)] = \int d^dy \, \( g[m(y)] \frac{\d }{\d m(y)} h[m(x)] - h[m(y)] \frac{\d }{\d m(y)} g[m(x)] \). \ee
		Note that this expansion is essentially done in $f_{l\neq 0} / \o$: thus we only require that $\o$ be large compared to the amplitude of the drive, and not necessarily compared to the size of the static part $f_0$. 
		
		To calculate $f_{\rm eff}$, we use the following commutators: 
		\bea [c, V'] & = c V'' \\ 
		[c, \p_a^2 m] & = 0 \\ 
		[\p_a^2m,\p_b^2m] & = 0 \\ 
		[\p_a^2 m, V'] & = -(\p_b m)^2V''' ,\eea 
		where $c$ stands for a constant operator (for us either $\wt h $ or $\ob h$).\footnote{The second line comes from $([c,\p_a^2m])[m(x)] = c\int_y \p_{a,y}^2 \d(x-y) = 0$. } Then one derives 
		\be [f_0,f_1] = -\wt h V'' - \wt K_b (\p_bm)^2V''' .\ee 
		Taking the commutator of this with $f_1$, some algebra gives 
		\bea [f_1,[f_0,f_1]] & = 2\wt h \wt K_a (\p_am)^2V''''  + \wt h^2 V''' - \wt K_a \wt K_b [\p_a^2 m , (\p_bm)^2V''' ] .\eea 
		Evaluating the last commutator under the simplifying assumption that $\d^5 V / \d m^5 = 0$ and collecting terms, we arrive at\footnote{If we were to keep track of the noise as well, we would find a renormalization of the diffusion constant. Since we are, in any case, not assuming FDT, this renormalization will not be important to explicitly calculate. }
		\bea \p_t m & = \ob K_a \p_a^2 m - \ob h - \( V' + \frac{\wt h^2}{\o^2} V'''\) - \l_a (\p_am)^2   - 2\wt K_a \wt K_b \( 2\p_a^2 \p_b m \p_b m V''' + 2 \p_am\p_bm \p_a\p_bm  V '''' + (\p_a \p_bm)^2 V'''\)\eea 
		where 
		\be \l_a \equiv \frac{2 \wt h \wt K_a V''''}{\o^2}.\ee 
		This yields 
		\be \p_tm = K\D^2m + \wt K_a (\D_a m)^2 - V'(m) - \ob h,\ee 
		together with additional interactions that are unimportant for determining the qualitative behavior of the coarsening dynamics. This is of the same form as the equation for $\sfR$ in \eqref{langevins} after taking $\wt K_x = -\wt K_y$ and adjusting notation appropriately.

		\section{Connectivity of locally testable classical memories} \label{app:samephase}
		
		Let us define a {\it locally testable classical memory} to be a classical memory whose logical state can be read off by measuring expectation values of local operators. In this section, we argue that any two locally testable classical memories with the same number of logical bits are ``continuously connectible''. For us, this means that given any two memories $\mca,\mcb$ with the same number of logical bits, there always exists a 1-parameter family of memories $\mcm(t)$ with $\mcm(0) = \mca, \mcm(1)=\mcb$, such that logical information is preserved along the path $\mcm(t)$, and that for small $\ep$, the rules $\mcm(t)$ and $\mcm(t+\epsilon)$ differ by a small amount, in the sense that the probability for $\mcm(t)$ and $\mcm(t+\epsilon)$ to act differently on a given site at a given time vanishes with $\ep\ra0$. While this definition and the statements made in the section below can be made mathematically rigorous, for space reasons we will content ourselves with a more concise schematic treatment.
		
		The aforementioned claim is perhaps not so obvious a priori, since taking $\mcm(t) = (1-t) \mca + t \mcb$ to linearly interpolate between $\mca$ and $\mcb$---meaning that $\mca$ updates are applied with probability $1-t$ and $\mcb$ updates with probability $t$---will usually {\it not} define a memory-preserving path. As an explicit example, this linear interpolation fails if $\mca$ is a squeezing code and $\mcb = \mca^\neg$ is its dual, c.f. \eqref{dualdef} (it however appears to succeed if $\mca$ is a 2D squeezing code and $\mcb$ is Toom's rule). Nevertheless, we argue below that a memory-preserving path can always be constructed. 
		
		Our argument proceeds by explicitly constructing such a path. Suppose $\mca$ acts on spin system with local state space $\mch$. Consider the automaton $\mca \tp \unit$ acting on $\mch \tp \mch_{an}$, where $\unit$ is a trivial automaton and $\mch_{an}$ is a copy of $\mch$. We will construct a memory-preserving path between $\mca \tp \unit$ and $\mcb \tp \unit$.\footnote{If one does not want to double the number of bits, the ancillae in $\mch_{an}$ can be taken instead to be a subsystem of the original bits.} This path is constructed according to the following procedure:
		\begin{enumerate}
			\item Let $\mch_{an}$ evolve under an arbitrary memory-preserving automaton $\mcc$ which has the same number of locally-testable logical states as $\mca$, with $\mch_{an}$ initialized in an arbitrary logical state. During this stage, the automaton dynamics is a smooth interpolation between $\mca \tp \unit$ and $\mca \tp \mcc$, with the first factor remaining constant. 
			\item Copy the logical information preserved by $\mca$ into $\mch_{an}$. This can be done using non-reciprocal dynamics in a smooth way by slowly turning on a bias which aligns bits in $\mch_{an}$ with bits in $\mch$.
			\item Now that the original logical information is safely stored in $\mch_{an}$, deform the automaton from $\mca \tp \mcc$ to $\mcb \tp \mcc$ along an arbitrary path which keeps the second factor fixed. 
			\item Re-copy the information from $\mch_{an}$ to $\mch$. 
			\item Finally, turn off the error-correcting dynamics on $\mch_{an}$ by moving along a path from $\mcb \tp \mcc$ to $\mcb \tp \unit$ that keeps the first factor fixed. 
		\end{enumerate}
		This provides a map between a memory evolving under $\mca$ (plus a collection of trivial ancillae) to a memory evolving under $\mcb$ (plus the same ancillae), which maps logical states of $\mca$ to those of $\mcb$ in a bijective way, which is the desired result. 
		
		%We note in passing that since this argument involved copying logical information into a collection of ancillae, it does not go through for {\it quantum} memories.
		% (where the result is likely false on e.g. account of the existence of multiple distinct topological orders with the same ground state degenercy on a given manifold). 
	\end{widetext} 

        \bibliography{refs}

%apsrev4-2.bst 2019-01-14 (MD) hand-edited version of apsrev4-1.bst
%Control: key (0)
%Control: author (8) initials jnrlst
%Control: editor formatted (1) identically to author
%Control: production of article title (0) allowed
%Control: page (0) single
%Control: year (1) truncated
%Control: production of eprint (0) enabled
\begin{thebibliography}{56}%
\makeatletter
\providecommand \@ifxundefined [1]{%
 \@ifx{#1\undefined}
}%
\providecommand \@ifnum [1]{%
 \ifnum #1\expandafter \@firstoftwo
 \else \expandafter \@secondoftwo
 \fi
}%
\providecommand \@ifx [1]{%
 \ifx #1\expandafter \@firstoftwo
 \else \expandafter \@secondoftwo
 \fi
}%
\providecommand \natexlab [1]{#1}%
\providecommand \enquote  [1]{``#1''}%
\providecommand \bibnamefont  [1]{#1}%
\providecommand \bibfnamefont [1]{#1}%
\providecommand \citenamefont [1]{#1}%
\providecommand \href@noop [0]{\@secondoftwo}%
\providecommand \href [0]{\begingroup \@sanitize@url \@href}%
\providecommand \@href[1]{\@@startlink{#1}\@@href}%
\providecommand \@@href[1]{\endgroup#1\@@endlink}%
\providecommand \@sanitize@url [0]{\catcode `\\12\catcode `\$12\catcode
  `\&12\catcode `\#12\catcode `\^12\catcode `\_12\catcode `\%12\relax}%
\providecommand \@@startlink[1]{}%
\providecommand \@@endlink[0]{}%
\providecommand \url  [0]{\begingroup\@sanitize@url \@url }%
\providecommand \@url [1]{\endgroup\@href {#1}{\urlprefix }}%
\providecommand \urlprefix  [0]{URL }%
\providecommand \Eprint [0]{\href }%
\providecommand \doibase [0]{https://doi.org/}%
\providecommand \selectlanguage [0]{\@gobble}%
\providecommand \bibinfo  [0]{\@secondoftwo}%
\providecommand \bibfield  [0]{\@secondoftwo}%
\providecommand \translation [1]{[#1]}%
\providecommand \BibitemOpen [0]{}%
\providecommand \bibitemStop [0]{}%
\providecommand \bibitemNoStop [0]{.\EOS\space}%
\providecommand \EOS [0]{\spacefactor3000\relax}%
\providecommand \BibitemShut  [1]{\csname bibitem#1\endcsname}%
\let\auto@bib@innerbib\@empty
%</preamble>
\bibitem [{\citenamefont {Ponselet}(2013)}]{ponselet2013phase}%
  \BibitemOpen
  \bibfield  {author} {\bibinfo {author} {\bibfnamefont {L.}~\bibnamefont
  {Ponselet}},\ }\bibfield  {title} {\bibinfo {title} {Phase transitions in
  probabilistic cellular automata},\ }\href@noop {} {\bibfield  {journal}
  {\bibinfo  {journal} {arXiv preprint arXiv:1312.3612}\ } (\bibinfo {year}
  {2013})}\BibitemShut {NoStop}%
\bibitem [{\citenamefont {Rakovszky}\ \emph {et~al.}(2024)\citenamefont
  {Rakovszky}, \citenamefont {Gopalakrishnan},\ and\ \citenamefont
  {Von~Keyserlingk}}]{rakovszky2024defining}%
  \BibitemOpen
  \bibfield  {author} {\bibinfo {author} {\bibfnamefont {T.}~\bibnamefont
  {Rakovszky}}, \bibinfo {author} {\bibfnamefont {S.}~\bibnamefont
  {Gopalakrishnan}},\ and\ \bibinfo {author} {\bibfnamefont {C.}~\bibnamefont
  {Von~Keyserlingk}},\ }\bibfield  {title} {\bibinfo {title} {Defining stable
  phases of open quantum systems},\ }\href@noop {} {\bibfield  {journal}
  {\bibinfo  {journal} {Physical Review X}\ }\textbf {\bibinfo {volume} {14}},\
  \bibinfo {pages} {041031} (\bibinfo {year} {2024})}\BibitemShut {NoStop}%
\bibitem [{\citenamefont {Sang}\ \emph {et~al.}(2024)\citenamefont {Sang},
  \citenamefont {Zou},\ and\ \citenamefont {Hsieh}}]{sang2024mixed}%
  \BibitemOpen
  \bibfield  {author} {\bibinfo {author} {\bibfnamefont {S.}~\bibnamefont
  {Sang}}, \bibinfo {author} {\bibfnamefont {Y.}~\bibnamefont {Zou}},\ and\
  \bibinfo {author} {\bibfnamefont {T.~H.}\ \bibnamefont {Hsieh}},\ }\bibfield
  {title} {\bibinfo {title} {Mixed-state quantum phases: Renormalization and
  quantum error correction},\ }\href@noop {} {\bibfield  {journal} {\bibinfo
  {journal} {Physical Review X}\ }\textbf {\bibinfo {volume} {14}},\ \bibinfo
  {pages} {031044} (\bibinfo {year} {2024})}\BibitemShut {NoStop}%
\bibitem [{\citenamefont {Pajouheshgar}\ \emph {et~al.}(2025)\citenamefont
  {Pajouheshgar}, \citenamefont {Bhardwaj}, \citenamefont {Selub},\ and\
  \citenamefont {Lake}}]{pajouheshgar2025exploring}%
  \BibitemOpen
  \bibfield  {author} {\bibinfo {author} {\bibfnamefont {E.}~\bibnamefont
  {Pajouheshgar}}, \bibinfo {author} {\bibfnamefont {A.}~\bibnamefont
  {Bhardwaj}}, \bibinfo {author} {\bibfnamefont {N.}~\bibnamefont {Selub}},\
  and\ \bibinfo {author} {\bibfnamefont {E.}~\bibnamefont {Lake}},\ }\bibfield
  {title} {\bibinfo {title} {Exploring the landscape of non-equilibrium
  memories with neural cellular automata},\ }\href@noop {} {\bibfield
  {journal} {\bibinfo  {journal} {arXiv preprint arXiv:2508.15726}\ } (\bibinfo
  {year} {2025})}\BibitemShut {NoStop}%
\bibitem [{\citenamefont {Gacs}\ \emph {et~al.}(1978)\citenamefont {Gacs},
  \citenamefont {Kurdyumov},\ and\ \citenamefont {Levin}}]{gacs1978one}%
  \BibitemOpen
  \bibfield  {author} {\bibinfo {author} {\bibfnamefont {P.}~\bibnamefont
  {Gacs}}, \bibinfo {author} {\bibfnamefont {G.~L.}\ \bibnamefont
  {Kurdyumov}},\ and\ \bibinfo {author} {\bibfnamefont {L.~A.}\ \bibnamefont
  {Levin}},\ }\bibfield  {title} {\bibinfo {title} {One-dimensional uniform
  arrays that wash out finite islands},\ }\href@noop {} {\bibfield  {journal}
  {\bibinfo  {journal} {Problemy Peredachi Informatsii}\ }\textbf {\bibinfo
  {volume} {14}},\ \bibinfo {pages} {92} (\bibinfo {year} {1978})}\BibitemShut
  {NoStop}%
\bibitem [{\citenamefont {Toom}(1980)}]{toom1980stable}%
  \BibitemOpen
  \bibfield  {author} {\bibinfo {author} {\bibfnamefont {A.~L.}\ \bibnamefont
  {Toom}},\ }\bibfield  {title} {\bibinfo {title} {Stable and attractive
  trajectories in multicomponent systems},\ }\href@noop {} {\bibfield
  {journal} {\bibinfo  {journal} {Multicomponent random systems}\ }\textbf
  {\bibinfo {volume} {6}},\ \bibinfo {pages} {549} (\bibinfo {year}
  {1980})}\BibitemShut {NoStop}%
\bibitem [{\citenamefont {Cirelson}(2006)}]{cirel2006reliable}%
  \BibitemOpen
  \bibfield  {author} {\bibinfo {author} {\bibfnamefont {B.}~\bibnamefont
  {Cirelson}},\ }\bibfield  {title} {\bibinfo {title} {Reliable storage of
  information in a system of unreliable components with local interactions},\
  }in\ \href@noop {} {\emph {\bibinfo {booktitle} {Locally Interacting Systems
  and Their Application in Biology}}}\ (\bibinfo {organization} {Springer},\
  \bibinfo {year} {2006})\ pp.\ \bibinfo {pages} {15--30}\BibitemShut {NoStop}%
\bibitem [{\citenamefont {Gacs}(2001)}]{gacs2001reliable}%
  \BibitemOpen
  \bibfield  {author} {\bibinfo {author} {\bibfnamefont {P.}~\bibnamefont
  {Gacs}},\ }\bibfield  {title} {\bibinfo {title} {Reliable cellular automata
  with self-organization},\ }\href@noop {} {\bibfield  {journal} {\bibinfo
  {journal} {Journal of Statistical Physics}\ }\textbf {\bibinfo {volume}
  {103}},\ \bibinfo {pages} {45} (\bibinfo {year} {2001})}\BibitemShut
  {NoStop}%
\bibitem [{\citenamefont {Mordvintsev}\ \emph {et~al.}(2020)\citenamefont
  {Mordvintsev}, \citenamefont {Randazzo}, \citenamefont {Niklasson},\ and\
  \citenamefont {Levin}}]{mordvintsev2020growing}%
  \BibitemOpen
  \bibfield  {author} {\bibinfo {author} {\bibfnamefont {A.}~\bibnamefont
  {Mordvintsev}}, \bibinfo {author} {\bibfnamefont {E.}~\bibnamefont
  {Randazzo}}, \bibinfo {author} {\bibfnamefont {E.}~\bibnamefont
  {Niklasson}},\ and\ \bibinfo {author} {\bibfnamefont {M.}~\bibnamefont
  {Levin}},\ }\bibfield  {title} {\bibinfo {title} {Growing neural cellular
  automata},\ }\href@noop {} {\bibfield  {journal} {\bibinfo  {journal}
  {Distill}\ }\textbf {\bibinfo {volume} {5}},\ \bibinfo {pages} {e23}
  (\bibinfo {year} {2020})}\BibitemShut {NoStop}%
\bibitem [{\citenamefont {Challa}\ \emph {et~al.}(2024)\citenamefont {Challa},
  \citenamefont {Hao}, \citenamefont {Rozum},\ and\ \citenamefont
  {Rocha}}]{challa2024effect}%
  \BibitemOpen
  \bibfield  {author} {\bibinfo {author} {\bibfnamefont {A.}~\bibnamefont
  {Challa}}, \bibinfo {author} {\bibfnamefont {D.}~\bibnamefont {Hao}},
  \bibinfo {author} {\bibfnamefont {J.~C.}\ \bibnamefont {Rozum}},\ and\
  \bibinfo {author} {\bibfnamefont {L.~M.}\ \bibnamefont {Rocha}},\ }\bibfield
  {title} {\bibinfo {title} {The effect of noise on the density classification
  task for various cellular automata rules},\ }in\ \href@noop {} {\emph
  {\bibinfo {booktitle} {Artificial Life Conference Proceedings 36}}},\ Vol.\
  \bibinfo {volume} {2024}\ (\bibinfo {organization} {MIT Press One Rogers
  Street, Cambridge, MA 02142-1209, USA journals-info~…},\ \bibinfo {year}
  {2024})\ p.~\bibinfo {pages} {83}\BibitemShut {NoStop}%
\bibitem [{\citenamefont {Pajouheshgar}\ \emph {et~al.}(2024)\citenamefont
  {Pajouheshgar}, \citenamefont {Xu},\ and\ \citenamefont
  {S{\"u}sstrunk}}]{pajouheshgar2024noisenca}%
  \BibitemOpen
  \bibfield  {author} {\bibinfo {author} {\bibfnamefont {E.}~\bibnamefont
  {Pajouheshgar}}, \bibinfo {author} {\bibfnamefont {Y.}~\bibnamefont {Xu}},\
  and\ \bibinfo {author} {\bibfnamefont {S.}~\bibnamefont {S{\"u}sstrunk}},\
  }\bibfield  {title} {\bibinfo {title} {Noisenca: Noisy seed improves
  spatio-temporal continuity of neural cellular automata},\ }in\ \href@noop {}
  {\emph {\bibinfo {booktitle} {Artificial Life Conference Proceedings 36}}},\
  Vol.\ \bibinfo {volume} {2024}\ (\bibinfo {organization} {MIT Press One
  Rogers Street, Cambridge, MA 02142-1209, USA journals-info~…},\ \bibinfo
  {year} {2024})\ p.~\bibinfo {pages} {57}\BibitemShut {NoStop}%
\bibitem [{\citenamefont {Harrington}(2004)}]{Harrington2004}%
  \BibitemOpen
  \bibfield  {author} {\bibinfo {author} {\bibfnamefont {J.~W.}\ \bibnamefont
  {Harrington}},\ }\emph {\bibinfo {title} {Analysis of Quantum
  Error-Correcting Codes: Symplectic Lattice Codes and Toric Codes}},\
  \href@noop {} {Ph.D. thesis},\ \bibinfo  {school} {Caltech} (\bibinfo {year}
  {2004})\BibitemShut {NoStop}%
\bibitem [{\citenamefont {Ray}\ \emph {et~al.}(2024)\citenamefont {Ray},
  \citenamefont {Laflamme},\ and\ \citenamefont {Kubica}}]{ray2024protecting}%
  \BibitemOpen
  \bibfield  {author} {\bibinfo {author} {\bibfnamefont {A.}~\bibnamefont
  {Ray}}, \bibinfo {author} {\bibfnamefont {R.}~\bibnamefont {Laflamme}},\ and\
  \bibinfo {author} {\bibfnamefont {A.}~\bibnamefont {Kubica}},\ }\bibfield
  {title} {\bibinfo {title} {Protecting information via probabilistic cellular
  automata},\ }\href@noop {} {\bibfield  {journal} {\bibinfo  {journal}
  {Physical Review E}\ }\textbf {\bibinfo {volume} {109}},\ \bibinfo {pages}
  {044141} (\bibinfo {year} {2024})}\BibitemShut {NoStop}%
\bibitem [{\citenamefont {Balasubramanian}\ \emph {et~al.}(2024)\citenamefont
  {Balasubramanian}, \citenamefont {Davydova},\ and\ \citenamefont
  {Lake}}]{balasubramanian2024local}%
  \BibitemOpen
  \bibfield  {author} {\bibinfo {author} {\bibfnamefont {S.}~\bibnamefont
  {Balasubramanian}}, \bibinfo {author} {\bibfnamefont {M.}~\bibnamefont
  {Davydova}},\ and\ \bibinfo {author} {\bibfnamefont {E.}~\bibnamefont
  {Lake}},\ }\bibfield  {title} {\bibinfo {title} {A local automaton for the 2d
  toric code},\ }\href@noop {} {\bibfield  {journal} {\bibinfo  {journal}
  {arXiv preprint arXiv:2412.19803}\ } (\bibinfo {year} {2024})}\BibitemShut
  {NoStop}%
\bibitem [{\citenamefont {Paletta}\ \emph {et~al.}(2025)\citenamefont
  {Paletta}, \citenamefont {Leverrier}, \citenamefont {Mirrahimi},\ and\
  \citenamefont {Vuillot}}]{paletta2025high}%
  \BibitemOpen
  \bibfield  {author} {\bibinfo {author} {\bibfnamefont {L.}~\bibnamefont
  {Paletta}}, \bibinfo {author} {\bibfnamefont {A.}~\bibnamefont {Leverrier}},
  \bibinfo {author} {\bibfnamefont {M.}~\bibnamefont {Mirrahimi}},\ and\
  \bibinfo {author} {\bibfnamefont {C.}~\bibnamefont {Vuillot}},\ }\bibfield
  {title} {\bibinfo {title} {High-performance local automaton decoders for
  defect matching in 1d},\ }\href@noop {} {\bibfield  {journal} {\bibinfo
  {journal} {arXiv preprint arXiv:2505.10162}\ } (\bibinfo {year}
  {2025})}\BibitemShut {NoStop}%
\bibitem [{\citenamefont {Lake}(2025{\natexlab{a}})}]{lake2025fast}%
  \BibitemOpen
  \bibfield  {author} {\bibinfo {author} {\bibfnamefont {E.}~\bibnamefont
  {Lake}},\ }\bibfield  {title} {\bibinfo {title} {Fast offline decoding with
  local message-passing automata},\ }\href@noop {} {\bibfield  {journal}
  {\bibinfo  {journal} {arXiv preprint arXiv:2506.03266}\ } (\bibinfo {year}
  {2025}{\natexlab{a}})}\BibitemShut {NoStop}%
\bibitem [{\citenamefont {Vasil'Ev}\ \emph {et~al.}(1969)\citenamefont
  {Vasil'Ev}, \citenamefont {Petrovskaya},\ and\ \citenamefont
  {Piatetskii-Shapiro}}]{vasil1969modelling}%
  \BibitemOpen
  \bibfield  {author} {\bibinfo {author} {\bibfnamefont {N.}~\bibnamefont
  {Vasil'Ev}}, \bibinfo {author} {\bibfnamefont {M.}~\bibnamefont
  {Petrovskaya}},\ and\ \bibinfo {author} {\bibfnamefont {I.}~\bibnamefont
  {Piatetskii-Shapiro}},\ }\bibfield  {title} {\bibinfo {title} {Modelling of
  voting with random error},\ }\href@noop {} {\bibfield  {journal} {\bibinfo
  {journal} {Avtomatika i Telemekhanika}\ }\textbf {\bibinfo {volume} {10}},\
  \bibinfo {pages} {103} (\bibinfo {year} {1969})}\BibitemShut {NoStop}%
\bibitem [{\citenamefont {Toom}(1974)}]{toom1974nonergodic}%
  \BibitemOpen
  \bibfield  {author} {\bibinfo {author} {\bibfnamefont {A.~L.}\ \bibnamefont
  {Toom}},\ }\bibfield  {title} {\bibinfo {title} {Nonergodic multidimensional
  system of automata},\ }\href@noop {} {\bibfield  {journal} {\bibinfo
  {journal} {Problemy Peredachi Informatsii}\ }\textbf {\bibinfo {volume}
  {10}},\ \bibinfo {pages} {70} (\bibinfo {year} {1974})}\BibitemShut {NoStop}%
\bibitem [{\citenamefont {Swart}\ \emph {et~al.}(2022)\citenamefont {Swart},
  \citenamefont {Szab{\'o}},\ and\ \citenamefont
  {Toninelli}}]{swart2022peierls}%
  \BibitemOpen
  \bibfield  {author} {\bibinfo {author} {\bibfnamefont {J.~M.}\ \bibnamefont
  {Swart}}, \bibinfo {author} {\bibfnamefont {R.}~\bibnamefont {Szab{\'o}}},\
  and\ \bibinfo {author} {\bibfnamefont {C.}~\bibnamefont {Toninelli}},\
  }\bibfield  {title} {\bibinfo {title} {Peierls bounds from toom contours},\
  }\href@noop {} {\bibfield  {journal} {\bibinfo  {journal} {arXiv preprint
  arXiv:2202.10999}\ } (\bibinfo {year} {2022})}\BibitemShut {NoStop}%
\bibitem [{\citenamefont {Gacs}(1989)}]{gacs1989self}%
  \BibitemOpen
  \bibfield  {author} {\bibinfo {author} {\bibfnamefont {P.}~\bibnamefont
  {Gacs}},\ }\bibfield  {title} {\bibinfo {title} {Self-correcting
  two-dimensional arrays.},\ }\href@noop {} {\bibfield  {journal} {\bibinfo
  {journal} {Adv. Comput. Res.}\ }\textbf {\bibinfo {volume} {5}},\ \bibinfo
  {pages} {223} (\bibinfo {year} {1989})}\BibitemShut {NoStop}%
\bibitem [{\citenamefont {Poulin}\ \emph {et~al.}(2019)\citenamefont {Poulin},
  \citenamefont {Melko},\ and\ \citenamefont {Hastings}}]{poulin2019self}%
  \BibitemOpen
  \bibfield  {author} {\bibinfo {author} {\bibfnamefont {D.}~\bibnamefont
  {Poulin}}, \bibinfo {author} {\bibfnamefont {R.~G.}\ \bibnamefont {Melko}},\
  and\ \bibinfo {author} {\bibfnamefont {M.~B.}\ \bibnamefont {Hastings}},\
  }\bibfield  {title} {\bibinfo {title} {Self-correction in wegner's
  three-dimensional ising lattice gauge theory},\ }\href@noop {} {\bibfield
  {journal} {\bibinfo  {journal} {Physical Review B}\ }\textbf {\bibinfo
  {volume} {99}},\ \bibinfo {pages} {094103} (\bibinfo {year}
  {2019})}\BibitemShut {NoStop}%
\bibitem [{\citenamefont {Kubica}\ and\ \citenamefont
  {Preskill}(2019)}]{kubica2019cellular}%
  \BibitemOpen
  \bibfield  {author} {\bibinfo {author} {\bibfnamefont {A.}~\bibnamefont
  {Kubica}}\ and\ \bibinfo {author} {\bibfnamefont {J.}~\bibnamefont
  {Preskill}},\ }\bibfield  {title} {\bibinfo {title} {Cellular-automaton
  decoders with provable thresholds for topological codes},\ }\href@noop {}
  {\bibfield  {journal} {\bibinfo  {journal} {Physical review letters}\
  }\textbf {\bibinfo {volume} {123}},\ \bibinfo {pages} {020501} (\bibinfo
  {year} {2019})}\BibitemShut {NoStop}%
\bibitem [{\citenamefont {Fern{\'a}ndez}\ and\ \citenamefont
  {Toom}(2001)}]{fernandez2001non}%
  \BibitemOpen
  \bibfield  {author} {\bibinfo {author} {\bibfnamefont {R.}~\bibnamefont
  {Fern{\'a}ndez}}\ and\ \bibinfo {author} {\bibfnamefont {A.}~\bibnamefont
  {Toom}},\ }\bibfield  {title} {\bibinfo {title} {Non-gibbsianness of the
  invariant measures of non-reversible cellular automata with totally
  asymmetric noise},\ }\href@noop {} {\bibfield  {journal} {\bibinfo  {journal}
  {arXiv preprint math-ph/0101014}\ } (\bibinfo {year} {2001})}\BibitemShut
  {NoStop}%
\bibitem [{sup(2026)}]{supp}%
  \BibitemOpen
  \href {https://ethanlake.github.io/squeezing-codes/} {\bibinfo {title} {See
  the demo at ethanlake.github.io/squeezing-codes/}} (\bibinfo {year}
  {2026})\BibitemShut {NoStop}%
\bibitem [{\citenamefont {Bowick}\ \emph {et~al.}(2022)\citenamefont {Bowick},
  \citenamefont {Fakhri}, \citenamefont {Marchetti},\ and\ \citenamefont
  {Ramaswamy}}]{bowick2022symmetry}%
  \BibitemOpen
  \bibfield  {author} {\bibinfo {author} {\bibfnamefont {M.~J.}\ \bibnamefont
  {Bowick}}, \bibinfo {author} {\bibfnamefont {N.}~\bibnamefont {Fakhri}},
  \bibinfo {author} {\bibfnamefont {M.~C.}\ \bibnamefont {Marchetti}},\ and\
  \bibinfo {author} {\bibfnamefont {S.}~\bibnamefont {Ramaswamy}},\ }\bibfield
  {title} {\bibinfo {title} {Symmetry, thermodynamics, and topology in active
  matter},\ }\href@noop {} {\bibfield  {journal} {\bibinfo  {journal} {Physical
  Review X}\ }\textbf {\bibinfo {volume} {12}},\ \bibinfo {pages} {010501}
  (\bibinfo {year} {2022})}\BibitemShut {NoStop}%
\bibitem [{\citenamefont {Fruchart}\ \emph {et~al.}(2021)\citenamefont
  {Fruchart}, \citenamefont {Hanai}, \citenamefont {Littlewood},\ and\
  \citenamefont {Vitelli}}]{fruchart2021non}%
  \BibitemOpen
  \bibfield  {author} {\bibinfo {author} {\bibfnamefont {M.}~\bibnamefont
  {Fruchart}}, \bibinfo {author} {\bibfnamefont {R.}~\bibnamefont {Hanai}},
  \bibinfo {author} {\bibfnamefont {P.~B.}\ \bibnamefont {Littlewood}},\ and\
  \bibinfo {author} {\bibfnamefont {V.}~\bibnamefont {Vitelli}},\ }\bibfield
  {title} {\bibinfo {title} {Non-reciprocal phase transitions},\ }\href@noop {}
  {\bibfield  {journal} {\bibinfo  {journal} {Nature}\ }\textbf {\bibinfo
  {volume} {592}},\ \bibinfo {pages} {363} (\bibinfo {year}
  {2021})}\BibitemShut {NoStop}%
\bibitem [{\citenamefont {Dinelli}\ \emph {et~al.}(2023)\citenamefont
  {Dinelli}, \citenamefont {O’Byrne}, \citenamefont {Curatolo}, \citenamefont
  {Zhao}, \citenamefont {Sollich},\ and\ \citenamefont
  {Tailleur}}]{dinelli2023non}%
  \BibitemOpen
  \bibfield  {author} {\bibinfo {author} {\bibfnamefont {A.}~\bibnamefont
  {Dinelli}}, \bibinfo {author} {\bibfnamefont {J.}~\bibnamefont {O’Byrne}},
  \bibinfo {author} {\bibfnamefont {A.}~\bibnamefont {Curatolo}}, \bibinfo
  {author} {\bibfnamefont {Y.}~\bibnamefont {Zhao}}, \bibinfo {author}
  {\bibfnamefont {P.}~\bibnamefont {Sollich}},\ and\ \bibinfo {author}
  {\bibfnamefont {J.}~\bibnamefont {Tailleur}},\ }\bibfield  {title} {\bibinfo
  {title} {Non-reciprocity across scales in active mixtures},\ }\href@noop {}
  {\bibfield  {journal} {\bibinfo  {journal} {Nature Communications}\ }\textbf
  {\bibinfo {volume} {14}},\ \bibinfo {pages} {7035} (\bibinfo {year}
  {2023})}\BibitemShut {NoStop}%
\bibitem [{\citenamefont {O’Malley}\ \emph {et~al.}(2006)\citenamefont
  {O’Malley}, \citenamefont {Kozma}, \citenamefont {Korniss}, \citenamefont
  {R{\'a}cz},\ and\ \citenamefont {Caraco}}]{o2006fisher}%
  \BibitemOpen
  \bibfield  {author} {\bibinfo {author} {\bibfnamefont {L.}~\bibnamefont
  {O’Malley}}, \bibinfo {author} {\bibfnamefont {B.}~\bibnamefont {Kozma}},
  \bibinfo {author} {\bibfnamefont {G.}~\bibnamefont {Korniss}}, \bibinfo
  {author} {\bibfnamefont {Z.}~\bibnamefont {R{\'a}cz}},\ and\ \bibinfo
  {author} {\bibfnamefont {T.}~\bibnamefont {Caraco}},\ }\bibfield  {title}
  {\bibinfo {title} {Fisher waves and front roughening in a two-species
  invasion model with preemptive competition},\ }\href@noop {} {\bibfield
  {journal} {\bibinfo  {journal} {Physical Review E—Statistical, Nonlinear,
  and Soft Matter Physics}\ }\textbf {\bibinfo {volume} {74}},\ \bibinfo
  {pages} {041116} (\bibinfo {year} {2006})}\BibitemShut {NoStop}%
\bibitem [{\citenamefont {Giometto}\ \emph {et~al.}(2021)\citenamefont
  {Giometto}, \citenamefont {Nelson},\ and\ \citenamefont
  {Murray}}]{giometto2021antagonism}%
  \BibitemOpen
  \bibfield  {author} {\bibinfo {author} {\bibfnamefont {A.}~\bibnamefont
  {Giometto}}, \bibinfo {author} {\bibfnamefont {D.~R.}\ \bibnamefont
  {Nelson}},\ and\ \bibinfo {author} {\bibfnamefont {A.~W.}\ \bibnamefont
  {Murray}},\ }\bibfield  {title} {\bibinfo {title} {Antagonism between killer
  yeast strains as an experimental model for biological nucleation dynamics},\
  }\href@noop {} {\bibfield  {journal} {\bibinfo  {journal} {Elife}\ }\textbf
  {\bibinfo {volume} {10}},\ \bibinfo {pages} {e62932} (\bibinfo {year}
  {2021})}\BibitemShut {NoStop}%
\bibitem [{\citenamefont {Kikuchi}(1951)}]{kikuchi1951theory}%
  \BibitemOpen
  \bibfield  {author} {\bibinfo {author} {\bibfnamefont {R.}~\bibnamefont
  {Kikuchi}},\ }\bibfield  {title} {\bibinfo {title} {A theory of cooperative
  phenomena},\ }\href@noop {} {\bibfield  {journal} {\bibinfo  {journal}
  {Physical review}\ }\textbf {\bibinfo {volume} {81}},\ \bibinfo {pages} {988}
  (\bibinfo {year} {1951})}\BibitemShut {NoStop}%
\bibitem [{\citenamefont {Pelizzola}(2005)}]{pelizzola2005cluster}%
  \BibitemOpen
  \bibfield  {author} {\bibinfo {author} {\bibfnamefont {A.}~\bibnamefont
  {Pelizzola}},\ }\bibfield  {title} {\bibinfo {title} {Cluster variation
  method in statistical physics and probabilistic graphical models},\
  }\href@noop {} {\bibfield  {journal} {\bibinfo  {journal} {Journal of Physics
  A: Mathematical and General}\ }\textbf {\bibinfo {volume} {38}},\ \bibinfo
  {pages} {R309} (\bibinfo {year} {2005})}\BibitemShut {NoStop}%
\bibitem [{\citenamefont {Hohenberg}\ and\ \citenamefont
  {Halperin}(1977)}]{hohenberg1977theory}%
  \BibitemOpen
  \bibfield  {author} {\bibinfo {author} {\bibfnamefont {P.~C.}\ \bibnamefont
  {Hohenberg}}\ and\ \bibinfo {author} {\bibfnamefont {B.~I.}\ \bibnamefont
  {Halperin}},\ }\bibfield  {title} {\bibinfo {title} {Theory of dynamic
  critical phenomena},\ }\href@noop {} {\bibfield  {journal} {\bibinfo
  {journal} {Reviews of Modern Physics}\ }\textbf {\bibinfo {volume} {49}},\
  \bibinfo {pages} {435} (\bibinfo {year} {1977})}\BibitemShut {NoStop}%
\bibitem [{\citenamefont {Grinstein}\ \emph {et~al.}(1985)\citenamefont
  {Grinstein}, \citenamefont {Jayaprakash},\ and\ \citenamefont
  {He}}]{grinstein1985statistical}%
  \BibitemOpen
  \bibfield  {author} {\bibinfo {author} {\bibfnamefont {G.}~\bibnamefont
  {Grinstein}}, \bibinfo {author} {\bibfnamefont {C.}~\bibnamefont
  {Jayaprakash}},\ and\ \bibinfo {author} {\bibfnamefont {Y.}~\bibnamefont
  {He}},\ }\bibfield  {title} {\bibinfo {title} {Statistical mechanics of
  probabilistic cellular automata},\ }\href@noop {} {\bibfield  {journal}
  {\bibinfo  {journal} {Physical review letters}\ }\textbf {\bibinfo {volume}
  {55}},\ \bibinfo {pages} {2527} (\bibinfo {year} {1985})}\BibitemShut
  {NoStop}%
\bibitem [{\citenamefont {Bassler}\ and\ \citenamefont
  {Schmittmann}(1994)}]{bassler1994critical}%
  \BibitemOpen
  \bibfield  {author} {\bibinfo {author} {\bibfnamefont {K.}~\bibnamefont
  {Bassler}}\ and\ \bibinfo {author} {\bibfnamefont {B.}~\bibnamefont
  {Schmittmann}},\ }\bibfield  {title} {\bibinfo {title} {Critical dynamics of
  nonconserved ising-like systems},\ }\href@noop {} {\bibfield  {journal}
  {\bibinfo  {journal} {Physical review letters}\ }\textbf {\bibinfo {volume}
  {73}},\ \bibinfo {pages} {3343} (\bibinfo {year} {1994})}\BibitemShut
  {NoStop}%
\bibitem [{\citenamefont {Tauber}\ \emph {et~al.}(2002)\citenamefont {Tauber},
  \citenamefont {Akkineni},\ and\ \citenamefont {Santos}}]{tauber2002effects}%
  \BibitemOpen
  \bibfield  {author} {\bibinfo {author} {\bibfnamefont {U.~C.}\ \bibnamefont
  {Tauber}}, \bibinfo {author} {\bibfnamefont {V.~K.}\ \bibnamefont
  {Akkineni}},\ and\ \bibinfo {author} {\bibfnamefont {J.~E.}\ \bibnamefont
  {Santos}},\ }\bibfield  {title} {\bibinfo {title} {Effects of violating
  detailed balance on critical dynamics},\ }\href@noop {} {\bibfield  {journal}
  {\bibinfo  {journal} {Physical review letters}\ }\textbf {\bibinfo {volume}
  {88}},\ \bibinfo {pages} {045702} (\bibinfo {year} {2002})}\BibitemShut
  {NoStop}%
\bibitem [{\citenamefont {Makowiec}(2000)}]{makowiec2000study}%
  \BibitemOpen
  \bibfield  {author} {\bibinfo {author} {\bibfnamefont {D.}~\bibnamefont
  {Makowiec}},\ }\bibfield  {title} {\bibinfo {title} {Study of continuous
  phase transition with toom cellular automata},\ }\href@noop {} {\bibfield
  {journal} {\bibinfo  {journal} {TASK Quarterly}\ }\textbf {\bibinfo {volume}
  {4}},\ \bibinfo {pages} {19} (\bibinfo {year} {2000})}\BibitemShut {NoStop}%
\bibitem [{\citenamefont {Adzhemyan}\ \emph {et~al.}(2022)\citenamefont
  {Adzhemyan}, \citenamefont {Evdokimov}, \citenamefont {Hnatic}, \citenamefont
  {Ivanova}, \citenamefont {Kompaniets}, \citenamefont {Kudlis},\ and\
  \citenamefont {Zakharov}}]{adzhemyan2022dynamic}%
  \BibitemOpen
  \bibfield  {author} {\bibinfo {author} {\bibfnamefont {L.~T.}\ \bibnamefont
  {Adzhemyan}}, \bibinfo {author} {\bibfnamefont {D.}~\bibnamefont
  {Evdokimov}}, \bibinfo {author} {\bibfnamefont {M.}~\bibnamefont {Hnatic}},
  \bibinfo {author} {\bibfnamefont {E.}~\bibnamefont {Ivanova}}, \bibinfo
  {author} {\bibfnamefont {M.}~\bibnamefont {Kompaniets}}, \bibinfo {author}
  {\bibfnamefont {A.}~\bibnamefont {Kudlis}},\ and\ \bibinfo {author}
  {\bibfnamefont {D.}~\bibnamefont {Zakharov}},\ }\bibfield  {title} {\bibinfo
  {title} {The dynamic critical exponent z for 2d and 3d ising models from
  five-loop epsilon expansion},\ }\href@noop {} {\bibfield  {journal} {\bibinfo
   {journal} {Physics Letters A}\ }\textbf {\bibinfo {volume} {425}},\ \bibinfo
  {pages} {127870} (\bibinfo {year} {2022})}\BibitemShut {NoStop}%
\bibitem [{\citenamefont {Nightingale}\ and\ \citenamefont
  {Blote}(1996)}]{nightingale1996dynamic}%
  \BibitemOpen
  \bibfield  {author} {\bibinfo {author} {\bibfnamefont {M.}~\bibnamefont
  {Nightingale}}\ and\ \bibinfo {author} {\bibfnamefont {H.}~\bibnamefont
  {Blote}},\ }\bibfield  {title} {\bibinfo {title} {Dynamic exponent of the
  two-dimensional ising model and monte carlo computation of the subdominant
  eigenvalue of the stochastic matrix},\ }\href@noop {} {\bibfield  {journal}
  {\bibinfo  {journal} {Physical review letters}\ }\textbf {\bibinfo {volume}
  {76}},\ \bibinfo {pages} {4548} (\bibinfo {year} {1996})}\BibitemShut
  {NoStop}%
\bibitem [{\citenamefont {Masaoka}\ \emph {et~al.}(2024)\citenamefont
  {Masaoka}, \citenamefont {Soejima},\ and\ \citenamefont
  {Watanabe}}]{masaoka2024rigorous}%
  \BibitemOpen
  \bibfield  {author} {\bibinfo {author} {\bibfnamefont {R.}~\bibnamefont
  {Masaoka}}, \bibinfo {author} {\bibfnamefont {T.}~\bibnamefont {Soejima}},\
  and\ \bibinfo {author} {\bibfnamefont {H.}~\bibnamefont {Watanabe}},\
  }\bibfield  {title} {\bibinfo {title} {Rigorous lower bound of dynamic
  critical exponents in critical frustration-free systems},\ }\href@noop {}
  {\bibfield  {journal} {\bibinfo  {journal} {arXiv preprint arXiv:2406.06415}\
  } (\bibinfo {year} {2024})}\BibitemShut {NoStop}%
\bibitem [{\citenamefont {Fates}(2007)}]{fates2007asynchronism}%
  \BibitemOpen
  \bibfield  {author} {\bibinfo {author} {\bibfnamefont {N.~A.}\ \bibnamefont
  {Fates}},\ }\bibfield  {title} {\bibinfo {title} {Asynchronism induces second
  order phase transitions in elementary cellular automata},\ }\href@noop {}
  {\bibfield  {journal} {\bibinfo  {journal} {arXiv preprint nlin/0703044}\ }
  (\bibinfo {year} {2007})}\BibitemShut {NoStop}%
\bibitem [{\citenamefont {Solon}\ and\ \citenamefont
  {Tailleur}(2015)}]{solon2015flocking}%
  \BibitemOpen
  \bibfield  {author} {\bibinfo {author} {\bibfnamefont {A.~P.}\ \bibnamefont
  {Solon}}\ and\ \bibinfo {author} {\bibfnamefont {J.}~\bibnamefont
  {Tailleur}},\ }\bibfield  {title} {\bibinfo {title} {Flocking with discrete
  symmetry: The two-dimensional active ising model},\ }\href@noop {} {\bibfield
   {journal} {\bibinfo  {journal} {Physical Review E}\ }\textbf {\bibinfo
  {volume} {92}},\ \bibinfo {pages} {042119} (\bibinfo {year}
  {2015})}\BibitemShut {NoStop}%
\bibitem [{\citenamefont {Lake}(2025{\natexlab{b}})}]{active_message_passing}%
  \BibitemOpen
  \bibfield  {author} {\bibinfo {author} {\bibfnamefont {E.}~\bibnamefont
  {Lake}},\ }\bibfield  {title} {\bibinfo {title} {Local active error
  correction from simulated confinement},\ }\href@noop {} {\bibfield  {journal}
  {\bibinfo  {journal} {arXiv preprint arXiv:2510.08056}\ } (\bibinfo {year}
  {2025}{\natexlab{b}})}\BibitemShut {NoStop}%
\bibitem [{\citenamefont {Marsan}\ \emph {et~al.}(2025)\citenamefont {Marsan},
  \citenamefont {Sablik},\ and\ \citenamefont {Torma}}]{marsan2025perturbed}%
  \BibitemOpen
  \bibfield  {author} {\bibinfo {author} {\bibfnamefont {H.}~\bibnamefont
  {Marsan}}, \bibinfo {author} {\bibfnamefont {M.}~\bibnamefont {Sablik}},\
  and\ \bibinfo {author} {\bibfnamefont {I.}~\bibnamefont {Torma}},\ }\bibfield
   {title} {\bibinfo {title} {A perturbed cellular automaton with two phase
  transitions for the ergodicity},\ }\href@noop {} {\bibfield  {journal}
  {\bibinfo  {journal} {arXiv preprint arXiv:2507.03485}\ } (\bibinfo {year}
  {2025})}\BibitemShut {NoStop}%
\bibitem [{\citenamefont {Spirin}\ \emph {et~al.}(2001)\citenamefont {Spirin},
  \citenamefont {Krapivsky},\ and\ \citenamefont
  {Redner}}]{spirin2001freezing}%
  \BibitemOpen
  \bibfield  {author} {\bibinfo {author} {\bibfnamefont {V.}~\bibnamefont
  {Spirin}}, \bibinfo {author} {\bibfnamefont {P.}~\bibnamefont {Krapivsky}},\
  and\ \bibinfo {author} {\bibfnamefont {S.}~\bibnamefont {Redner}},\
  }\bibfield  {title} {\bibinfo {title} {Freezing in ising ferromagnets},\
  }\href@noop {} {\bibfield  {journal} {\bibinfo  {journal} {Physical Review
  E}\ }\textbf {\bibinfo {volume} {65}},\ \bibinfo {pages} {016119} (\bibinfo
  {year} {2001})}\BibitemShut {NoStop}%
\bibitem [{\citenamefont {Lebowitz}\ \emph {et~al.}(1990)\citenamefont
  {Lebowitz}, \citenamefont {Maes},\ and\ \citenamefont
  {Speer}}]{lebowitz1990statistical}%
  \BibitemOpen
  \bibfield  {author} {\bibinfo {author} {\bibfnamefont {J.~L.}\ \bibnamefont
  {Lebowitz}}, \bibinfo {author} {\bibfnamefont {C.}~\bibnamefont {Maes}},\
  and\ \bibinfo {author} {\bibfnamefont {E.~R.}\ \bibnamefont {Speer}},\
  }\bibfield  {title} {\bibinfo {title} {Statistical mechanics of probabilistic
  cellular automata},\ }\href@noop {} {\bibfield  {journal} {\bibinfo
  {journal} {Journal of statistical physics}\ }\textbf {\bibinfo {volume}
  {59}},\ \bibinfo {pages} {117} (\bibinfo {year} {1990})}\BibitemShut
  {NoStop}%
\bibitem [{\citenamefont {Young}\ \emph {et~al.}(2020)\citenamefont {Young},
  \citenamefont {Gorshkov}, \citenamefont {Foss-Feig},\ and\ \citenamefont
  {Maghrebi}}]{young2020nonequilibrium}%
  \BibitemOpen
  \bibfield  {author} {\bibinfo {author} {\bibfnamefont {J.~T.}\ \bibnamefont
  {Young}}, \bibinfo {author} {\bibfnamefont {A.~V.}\ \bibnamefont {Gorshkov}},
  \bibinfo {author} {\bibfnamefont {M.}~\bibnamefont {Foss-Feig}},\ and\
  \bibinfo {author} {\bibfnamefont {M.~F.}\ \bibnamefont {Maghrebi}},\
  }\bibfield  {title} {\bibinfo {title} {Nonequilibrium fixed points of coupled
  ising models},\ }\href@noop {} {\bibfield  {journal} {\bibinfo  {journal}
  {Physical Review X}\ }\textbf {\bibinfo {volume} {10}},\ \bibinfo {pages}
  {011039} (\bibinfo {year} {2020})}\BibitemShut {NoStop}%
\bibitem [{\citenamefont {Agrawal}\ \emph {et~al.}(2024)\citenamefont
  {Agrawal}, \citenamefont {Cugliandolo}, \citenamefont {Faoro}, \citenamefont
  {Ioffe},\ and\ \citenamefont {Picco}}]{agrawal2024dynamical}%
  \BibitemOpen
  \bibfield  {author} {\bibinfo {author} {\bibfnamefont {R.}~\bibnamefont
  {Agrawal}}, \bibinfo {author} {\bibfnamefont {L.~F.}\ \bibnamefont
  {Cugliandolo}}, \bibinfo {author} {\bibfnamefont {L.}~\bibnamefont {Faoro}},
  \bibinfo {author} {\bibfnamefont {L.~B.}\ \bibnamefont {Ioffe}},\ and\
  \bibinfo {author} {\bibfnamefont {M.}~\bibnamefont {Picco}},\ }\bibfield
  {title} {\bibinfo {title} {Dynamical critical behavior on the nishimori point
  of frustrated ising models},\ }\href@noop {} {\bibfield  {journal} {\bibinfo
  {journal} {Physical Review E}\ }\textbf {\bibinfo {volume} {110}},\ \bibinfo
  {pages} {034120} (\bibinfo {year} {2024})}\BibitemShut {NoStop}%
\bibitem [{cod(2026)}]{code}%
  \BibitemOpen
  \href {https://github.com/ethanlake/squeezing-codes} {\bibinfo {title} {See
  the linked github repository}} (\bibinfo {year} {2026})\BibitemShut {NoStop}%
\bibitem [{\citenamefont {Kamieniarz}\ and\ \citenamefont
  {Blote}(1993)}]{kamieniarz1993universal}%
  \BibitemOpen
  \bibfield  {author} {\bibinfo {author} {\bibfnamefont {G.}~\bibnamefont
  {Kamieniarz}}\ and\ \bibinfo {author} {\bibfnamefont {H.}~\bibnamefont
  {Blote}},\ }\bibfield  {title} {\bibinfo {title} {Universal ratio of
  magnetization moments in two-dimensional ising models},\ }\href@noop {}
  {\bibfield  {journal} {\bibinfo  {journal} {Journal of Physics A:
  Mathematical and General}\ }\textbf {\bibinfo {volume} {26}},\ \bibinfo
  {pages} {201} (\bibinfo {year} {1993})}\BibitemShut {NoStop}%
\bibitem [{\citenamefont {Janssen}\ \emph {et~al.}(1989)\citenamefont
  {Janssen}, \citenamefont {Schaub},\ and\ \citenamefont
  {Schmittmann}}]{janssen1989new}%
  \BibitemOpen
  \bibfield  {author} {\bibinfo {author} {\bibfnamefont {H.}~\bibnamefont
  {Janssen}}, \bibinfo {author} {\bibfnamefont {B.}~\bibnamefont {Schaub}},\
  and\ \bibinfo {author} {\bibfnamefont {B.}~\bibnamefont {Schmittmann}},\
  }\bibfield  {title} {\bibinfo {title} {New universal short-time scaling
  behaviour of critical relaxation processes},\ }\href@noop {} {\bibfield
  {journal} {\bibinfo  {journal} {Zeitschrift fur Physik B Condensed Matter}\
  }\textbf {\bibinfo {volume} {73}},\ \bibinfo {pages} {539} (\bibinfo {year}
  {1989})}\BibitemShut {NoStop}%
\bibitem [{\citenamefont {Hinrichsen}(2000)}]{hinrichsen2000non}%
  \BibitemOpen
  \bibfield  {author} {\bibinfo {author} {\bibfnamefont {H.}~\bibnamefont
  {Hinrichsen}},\ }\bibfield  {title} {\bibinfo {title} {Non-equilibrium
  critical phenomena and phase transitions into absorbing states},\ }\href@noop
  {} {\bibfield  {journal} {\bibinfo  {journal} {Advances in physics}\ }\textbf
  {\bibinfo {volume} {49}},\ \bibinfo {pages} {815} (\bibinfo {year}
  {2000})}\BibitemShut {NoStop}%
\bibitem [{\citenamefont {Takeuchi}(2006)}]{takeuchi2006can}%
  \BibitemOpen
  \bibfield  {author} {\bibinfo {author} {\bibfnamefont {K.}~\bibnamefont
  {Takeuchi}},\ }\bibfield  {title} {\bibinfo {title} {Can the ising critical
  behaviour survive in non-equilibrium synchronous cellular automata?},\
  }\href@noop {} {\bibfield  {journal} {\bibinfo  {journal} {Physica D:
  Nonlinear Phenomena}\ }\textbf {\bibinfo {volume} {223}},\ \bibinfo {pages}
  {146} (\bibinfo {year} {2006})}\BibitemShut {NoStop}%
\bibitem [{\citenamefont {Houdayer}\ and\ \citenamefont
  {Hartmann}(2004)}]{houdayer2004low}%
  \BibitemOpen
  \bibfield  {author} {\bibinfo {author} {\bibfnamefont {J.}~\bibnamefont
  {Houdayer}}\ and\ \bibinfo {author} {\bibfnamefont {A.~K.}\ \bibnamefont
  {Hartmann}},\ }\bibfield  {title} {\bibinfo {title} {Low-temperature behavior
  of two-dimensional gaussian ising spin glasses},\ }\href@noop {} {\bibfield
  {journal} {\bibinfo  {journal} {Physical Review B—Condensed Matter and
  Materials Physics}\ }\textbf {\bibinfo {volume} {70}},\ \bibinfo {pages}
  {014418} (\bibinfo {year} {2004})}\BibitemShut {NoStop}%
\bibitem [{\citenamefont {Doi}(1976)}]{doi1976second}%
  \BibitemOpen
  \bibfield  {author} {\bibinfo {author} {\bibfnamefont {M.}~\bibnamefont
  {Doi}},\ }\bibfield  {title} {\bibinfo {title} {Second quantization
  representation for classical many-particle system},\ }\href@noop {}
  {\bibfield  {journal} {\bibinfo  {journal} {Journal of Physics A:
  Mathematical and General}\ }\textbf {\bibinfo {volume} {9}},\ \bibinfo
  {pages} {1465} (\bibinfo {year} {1976})}\BibitemShut {NoStop}%
\bibitem [{\citenamefont {Bennett}\ and\ \citenamefont
  {Grinstein}(1985)}]{bennett1985role}%
  \BibitemOpen
  \bibfield  {author} {\bibinfo {author} {\bibfnamefont {C.~H.}\ \bibnamefont
  {Bennett}}\ and\ \bibinfo {author} {\bibfnamefont {G.}~\bibnamefont
  {Grinstein}},\ }\bibfield  {title} {\bibinfo {title} {Role of irreversibility
  in stabilizing complex and nonergodic behavior in locally interacting
  discrete systems},\ }\href@noop {} {\bibfield  {journal} {\bibinfo  {journal}
  {Physical review letters}\ }\textbf {\bibinfo {volume} {55}},\ \bibinfo
  {pages} {657} (\bibinfo {year} {1985})}\BibitemShut {NoStop}%
\bibitem [{\citenamefont {Higashikawa}\ \emph {et~al.}(2018)\citenamefont
  {Higashikawa}, \citenamefont {Fujita},\ and\ \citenamefont
  {Sato}}]{higashikawa2018floquet}%
  \BibitemOpen
  \bibfield  {author} {\bibinfo {author} {\bibfnamefont {S.}~\bibnamefont
  {Higashikawa}}, \bibinfo {author} {\bibfnamefont {H.}~\bibnamefont
  {Fujita}},\ and\ \bibinfo {author} {\bibfnamefont {M.}~\bibnamefont {Sato}},\
  }\bibfield  {title} {\bibinfo {title} {Floquet engineering of classical
  systems},\ }\href@noop {} {\bibfield  {journal} {\bibinfo  {journal} {arXiv
  preprint arXiv:1810.01103}\ } (\bibinfo {year} {2018})}\BibitemShut {NoStop}%
\end{thebibliography}%
		\end{document}